\documentclass[acmsmall,screen,nonacm]{acmart}

\authorsaddresses{}

\usepackage{caption}
\usepackage[noabbrev,nameinlink,capitalise]{cleveref}

\usepackage{todonotes}
\usepackage{bm}
\usepackage{multirow}
\usepackage[shortlabels]{enumitem}
\usepackage{graphics}
\usepackage{wrapfig}

\usepackage{algorithm2e}
\SetKw{Real}{real}
\SetKw{Int}{int}
\SetKw{Assert}{assert}

\newcommand{\Z}{\mathbb{Z}}
\newcommand{\D}{\mathbb{D}}
\newcommand{\N}{\mathbb{N}}
\newcommand{\R}{\mathbb{R}}
\newcommand{\Q}{\mathbb{Q}}

\newcommand{\NP}{\mathsf{NP}}
\newcommand{\ram}{\exists^{\mathsf{ram}}}

\newcommand{\bx}{\bm{x}}
\newcommand{\ba}{\bm{a}}

\newcommand{\bd}{\bm{d}}

\newcommand{\bu}{\bm{u}}

\newcommand{\bv}{\bm{v}}
\newcommand{\bA}{\bm{A}}
\newcommand{\bB}{\bm{B}}

\newcommand{\bzero}{\bm{0}}

\newcommand{\NEXPTIME}{\mathsf{NEXPTIME}}
\newcommand{\EXPSPACE}{\mathsf{EXPSPACE}}
\newcommand{\TwoNEXPTIME}{\mathsf{2NEXPTIME}}
\newcommand{\PSPACE}{\mathsf{PSPACE}}
\newcommand{\coNP}{\mathsf{coNP}}

\newcommand{\OMIT}[1]{}
\newcommand{\struct}{\mathfrak{A}}
\newcommand{\ZLin}{\langle \mathbb{Z}; +,<,1,0\rangle}
\newcommand{\RLin}{\langle \mathbb{R}; +,<,1,0\rangle}

\newcommand{\QfLin}{\langle \mathbb{Q}; \lfloor\cdot\rfloor,+,<,1,0\rangle}
\newcommand{\QLin}{\langle \mathbb{Q};+,<,1,0\rangle}

\newif\ifdraft\drafttrue
\ifdraft
\newcommand\anthony[1]{{\color{blue}
[#1 - \textbf{Anthony}]}}

\else
\newcommand\anthony[1]{}

\fi

\newcommand{\Theory}{\mathcal{T}}
 
\begin{document}

\title{Ramsey Quantifiers in Linear Arithmetics}

\author{Pascal Bergsträßer}
\email{bergstraesser@cs.uni-kl.de}
\orcid{0000-0002-4681-2149}
\affiliation{\institution{University of Kaiserslautern-Landau}
  \city{Kaiserslautern}
  \country{Germany}
}

\author{Moses Ganardi}
\email{ganardi@mpi-sws.org}
\orcid{0000-0002-0775-7781}
\affiliation{\institution{Max Planck Institute for Software Systems (MPI-SWS)}
  \city{Kaiserslautern}
  \country{Germany}
}

\author{Anthony W. Lin}
\email{lin@cs.uni-kl.de}
\orcid{0000-0003-4715-5096}
\affiliation{\institution{University of Kaiserslautern-Landau}
  \city{Kaiserslautern}
  \country{Germany}
}
\affiliation{\institution{Max Planck Institute for Software Systems (MPI-SWS)}
  \city{Kaiserslautern}
  \country{Germany}
}

\author{Georg Zetzsche}
\email{georg@mpi-sws.org}
\orcid{0000-0002-6421-4388}
\affiliation{\institution{Max Planck Institute for Software Systems (MPI-SWS)}
  \city{Kaiserslautern}
  \country{Germany}
}

\begin{abstract}
    We study Satisfiability Modulo Theories (SMT) enriched with the so-called 
    Ramsey quantifiers, which assert the existence of cliques (complete graphs)
    in the graph induced by some formulas. The extended framework is known to 
    have 
    applications in
    proving program termination (in particular, whether a transitive binary
    predicate is well-founded), and monadic decomposability of SMT formulas.
Our main result is a new algorithm for eliminating Ramsey
quantifiers from three common SMT theories: Linear Integer Arithmetic (LIA),
Linear Real Arithmetic (LRA), and Linear Integer Real Arithmetic (LIRA).
In particular, if we work only with existentially quantified formulas, then our
algorithm runs in polynomial time and produces a formula of linear size. One 
immediate consequence is
that checking well-foundedness of a given formula in the aforementioned theory
defining a transitive predicate
can be straightforwardly handled by highly optimized SMT-solvers. We show also 
    how this provides a
    uniform semi-algorithm for verifying termination and liveness with 
    completeness guarantee
    (in fact, with an optimal computational complexity) for several well-known
    classes of
    infinite-state systems, which include succinct timed systems, one-counter 
    systems, and monotonic counter systems.
    Another immediate consequence is a solution to an open problem on checking
    monadic decomposability of a given relation in quantifier-free fragments of
    LRA and LIRA, which is an 
    important problem in automated reasoning and constraint databases. Our
    result immediately implies decidability of this problem with an optimal
    complexity (coNP-complete) and enables exploitation of SMT-solvers. It also
    provides a termination guarantee for the generic monadic decomposition 
    algorithm of Veanes et al. for LIA, LRA, and LIRA. We report encouraging
    experimental results on a prototype implementation of our algorithms on
    micro-benchmarks.
\end{abstract}



\maketitle
\section{Introduction}\label{sec:introduction}
The last two decades have witnessed significant advances in software
verification \cite{SMC-survey}. One prominent and fruitful approach to software 
verification is that of
\emph{deductive verification} and \emph{program logics}
\cite{NO80,Sho84,Leino-book}, whereby one models a 
specification of a program $P$ as a formula $\varphi_P$ over some logical 
theory, usually in first-order logic (FO), or a fragment thereof (e.g. 
quantifier-free 
formulas or existential formulas). That way, the original problem is reduced to 
satisfiability of the formula $\varphi_P$ (i.e. whether it has a solution). 
One decisive factor of the success of
this software verification approach is that solvers for satisfiability of 
boolean formulas and extensions to quantifier-free and existential theories 
(a.k.a. SAT-solvers and SMT-solvers, respectively) have made an enormous stride
forward in the last decades to the extent that they are now capable of solving 
practical industrial instances. The cornerstone theories in the SMT
framework include the theory of Linear Integer Arithmetic (LIA),
the theory of Linear Real Arithmetic (LRA), and the mixed theory of 
Linear Integer Real Arithmetic (LIRA).
Among others, these can be used to naturally model numeric programs
\cite{DBLP:conf/cav/HagueL11}, 
programs with clocks \cite{DBLP:conf/cav/HagueL11,BH06,D01,DPK01,DIBKS00}, 
linear hybrid systems \cite{BH06},
and numeric abstractions of programs that manipulate lists and arrays
\cite{lists-counters,DBLP:conf/cav/HagueL11}.

Most program logics in software verification can be formulated directly within 
FO. For example, to specify a \emph{safety} property, a programmer may provide a
formula $Inv$ asserting a desired invariant for the program. In turn, the
property that $Inv$ is an invariant is definable in FO. In fact, if we stay
within the quantifier-free fragment of FO, this check can be easily and
efficiently verified by SMT-solvers.
Some program logics, however, require us to go beyond FO. Most notably, when
verifying that a program terminates, a programmer must provide well-founded
relations (or a finite disjunction thereof) and prove that this covers the 
transitive closure of the program (e.g. see \cite{PR04}). [Some techniques
realize the proof rules of \cite{PR04} (e.g. see \cite{CPR11}) by constructing
relations that are guaranteed to be well-founded by construction, but these
limit the shapes of the well-founded relations that can be constructed.]
In case such a relation is synthesized and not guaranteed to be well-founded,
one may want to check the well-foundedness property automatically.
Since well-foundedness of a transitive
predicate is in general not a first-order property (e.g. see Problem 1.4.1 
of \cite{model-theory-book}), an extension of FO is required to be able to
reason about well-foundedness. One solution is to simply enrich FO with
an ad-hoc condition for checking well-foundedness of a relation \cite{BPR13}.
A more general solution is to extend FO with \emph{Ramsey quantifiers}
\cite{lics22-ramsey} (see also Chapter VII of \cite{BF85}) and study
elimination of such quantifiers in the logical theory under consideration.
This latter solution is known \cite{lics22-ramsey} to also provide an approach 
to analyze variable dependencies (a.k.a. \emph{monadic decomposability})
in a first-order formula, which has
applications in formal verification \cite{DBLP:journals/jacm/VeanesBNB17}
and query optimization in constraint databases \cite{GRS01,CDB-book}.

\OMIT{
requires a second-order quantifier asserting the existence of an 
\emph{invariant} $Inv$ for the program, resulting in a formula of the form
$\exists Inv( \varphi )$, where $\varphi$ is first-order. That is also the 
reason why proving safety often requires a synthesis algorithm 
to synthesize such an invariant,
in combination with an SMT-solver for checking satisfiability/validity of
$\varphi$. Such synthesis algorithms are implemented in solvers for 
\emph{Constraint Horn Clauses} (CHC), where $Inv$ is treated as a
recursive predicate, e.g., see \cite{BMR12,BGMR15}. Proving termination is
typically more challenging, where the ability to express that a
transitive relation
is well-founded is required \cite{PR04,CPR11}. Since well-foundedness is not 
definable in first-order logic even with the help of recursive predicates
(equivalently, second-order quantifiers), the program logic is sometimes
simply enriched by an ad-hoc condition for checking well-foundedness of a 
relation \cite{BPR13}. 
}

\subsubsection*{SMT with Ramsey quantifiers} 
\OMIT{
Recently, it was shown
\cite{lics22-ramsey} that enriching first-order logic with the so-called 
\emph{Ramsey quantifiers} could lead to a sufficiently powerful logic for
expressing termination/liveness properties, while at the same time allowing
decision procedures for certain first-order theories to be extended with such
quantifiers.
}
In a nutshell,
a Ramsey quantifier asserts the existence of an infinite sequence
of elements forming a clique (i.e. a complete graph) in the 
graph induced by a given formula. [There are in fact two flavors of Ramsey 
quantifiers, of which one asserts the existence of an undirected clique (e.g.
see Chapter VII of \cite{BF85}), and the other of a \emph{directed} clique
\cite{lics22-ramsey}. In the sequel, we will only deal with the latter because
of the applications to reasoning about liveness and variable dependencies.]
More precisely,
if $\varphi(\bm x,\bm y)$ is a formula over a structure $\struct$ with universe $D$
and $\bm x$, $\bm y$ are $k$-tuples of variables,
the formula $\ram \bm x,\bm y\colon \varphi(\bm x,\bm y)$ asserts the existence of an \emph{infinite (directed) $\varphi$-clique}, i.e.\ a sequence $\bm{v}_1,\bm {v}_2,\ldots$ of pairwise
distinct $k$-tuples in $D^k$ such that $\struct \models \varphi(\bm{v}_i,\bm{v}_j)$ for all $i < j$. 
For example,
in the theory $T = \RLin$ of Linear Real Arithmetic, we have $T \models 
\ram x,y\colon ( x < y \wedge x > 99 \wedge y < 100)$ because there
are infinitely many numbers between 99 and 100. 

How do Ramsey quantifiers connect to proving termination/liveness? Let us take a
proof rule for termination/liveness from \cite{PR04}, which
concerns covering the transitive closure $R^+$ of a relation $R$ by well-founded
relations (or a finite disjunction thereof). At its simplest form, we obtain a
verification condition of asserting 
\begin{equation}
	R \subseteq T \quad\wedge\quad T \circ R \subseteq T \quad\wedge\quad T \text{ is well-founded}.\label{proof-rule-well-founded}
\end{equation}
Such a $T$ satisfying the first two conjuncts is said to be an \emph{inductive
relation} \cite{PR04}. The disjunctive well-founded version can be stated
similarly, but with a disjunction of relations $T_i$ instead of just a
single $T$. Here, one defines a relation to be \emph{well-founded} if 
there is no infinite $T$-chain, i.e., $s_1, s_2, \ldots$ such
that $(s_i,s_{i+1}) \in T$ for each $i$. 
Clearly, if $T$ is well-founded, then there is no $T$-loop (i.e.\ $x$ with $T(x,x)$) and no infinite $T$-clique. Thus, $T$ also satisfies
\begin{equation}
	R \subseteq T \quad\wedge\quad T \circ R \subseteq T \quad\wedge\quad \text{no $T$-loop}\quad \wedge\quad \text{no infinite $T$-clique}.\label{proof-rule-re}
\end{equation}
Note that \eqref{proof-rule-re} also implies termination of $R$, despite imposing a weaker requirement on $T$.
However, the conditions in \eqref{proof-rule-re} are easily expressed with the Ramsey quantifier: The absence of $T$-loops is a first-order
property (i.e. $\neg \exists x \colon T(x,x)$) and the
absence of an infinite $T$-clique is definable with the help of a Ramsey
quantifier $\neg \ram x,y \colon T(x,y)$.
As an example for a covering that satisfies \eqref{proof-rule-re} but not \eqref{proof-rule-well-founded}, 
consider the well-founded relation $R=\{(i+1,i) \mid i\in\mathbb{N}\}$ 
and the covering $T=\{(i,j) \mid i,j\in \mathbb{N},~i>j\} \cup \{(i,i-1) \mid i\le 0\}$. 
Then $R^+\subseteq T$ and $T$ is loop-free and contains no (directed) infinite clique, 
hence \eqref{proof-rule-re} proves termination of $R$. However, $T$ is not well-founded.
\OMIT{
There
are many proof rules for termination/liveness of a program (e.g.  
\cite{CPR11,PR04,AProVE,UGK16,LR16}), but most require 
some kind of reasoning of well-foundedness of a relation. 
In addition, well-foundedness
of a relation $R$ can easily be expressed with the help of a Ramsey quantifier:
\[
    \neg \ram x,y R(x,y)
\]
}

Most techniques for handling Ramsey quantifiers proceed by eliminating them.
In the early 1980s, \citet{schmerlsimpson1982} showed that 
in LIA, Ramsey quantifiers $\ram \bm x,\bm y$ can be eliminated if $\bm x$ and 
$\bm y$ are single variables (hence, it is about cliques of numbers, not 
vectors). [Actually, their result concerns only undirected cliques, but the 
proof easily generalizes to directed cliques.] At the turn of the 21st centry,
\citet{DangI02} provided a procedure to decide whether a given 
relation $R$ described in LIRA admits an infinite directed clique. Their proof 
yields that general Ramsey quantifiers (i.e.\ about vectors) can be eliminated 
in LIRA: The procedure transforms the input formula into a LIRA formula that 
holds if and only if $R$ admits an infinite directed clique. However, the 
procedure of \citet{DangI02} (i)~requires the input formula to be
quantifier-free (this is also the case for \citet{schmerlsimpson1982}) and
(ii)~yields a formula with several quantifier alternations.  Because of~(ii),
the algorithm needs to then decide the truth of a LIRA formula with quantifier
alternations, for which \citet[p.~924]{DangI02} provide (based on
\citet{Weispfenning99}) a doubly exponential time bound of $2^{{L^{n^c}}}$ for
a constant $c$, where $L$ is the length of the input formula and $n$ is the
number of variables.
Because of (i), applying the algorithm to an existential LIRA formula $\varphi$
for $R$ necessitates an elimination of the existential quantifiers from
$\varphi$. If $\varphi$ is of length $\ell$ with $q$ quantified variables, then
according to~\citet[Theorem 5.1]{Weispfenning99}, this results in a
quantifier-free formula of size $2^{\ell^{q^d}}$ for some constant $d$.
Plugging this into the construction of \citet{DangI02} yields a triply exponential time
bound of $2^{2^{n^c\cdot \ell^{q^d}}}$.

More recent results \cite{lics22-ramsey,TL08} on eliminating Ramsey quantifiers 
over \emph{theories of string (resp. tree) automatic structures} are also worth
mentioning. These are rich classes of logical structures whose domains/relations can be 
encoded using string/tree automata \cite{BLSS03,BG00}, and subsume various
arithmetic theories including LIA and Skolem
Arithmetic (i.e. $\langle \Z; \times, \leq, 1, 0\rangle$). Among others, this gives
rise to a decision procedure for LIA with Ramsey quantifiers, which runs in
exponential time.
The main problem with the decision procedures given in \cite{lics22-ramsey,TL08}
is that
they cannot be implemented directly on top of an existing (and highly optimized)
SMT-solver, and their complexity is rather high. Secondly, it does not yield
algorithms for LRA and LIRA. In fact, the common extension of LIRA and automatic
structures are the so-called $\omega$-automatic structures, for which 
eliminability of Ramsey quantifiers is a long-standing open problem
\cite{K10}.

\subsubsection*{Monadic decomposability.} Another application of Ramsey quantifiers is the analysis of variable 
dependencies (i.e. monadic decomposability 
\cite{DBLP:journals/jacm/VeanesBNB17}) of formulas. Loosely speaking, a formula 
$\varphi(x_1,\ldots,x_n)$ is \emph{monadically decomposable} in the theory $\Theory$ if it is
equivalent (over $\Theory$) to a boolean combination of $\Theory$-formulas of
the form $\varphi(x_i)$, i.e., with at most one $x_i$ as a free variable.
This boolean combination of monadic formulas is a monadic decomposition of
$\varphi$.
For example, the formula $x_1+x_2 \geq 2 \wedge x_1 \geq 0 \wedge x_2 \geq 0$
is monadically decomposable in LIA as it is equivalent to
\[
    (x_1 \geq 2 \wedge x_2 \geq 0) \vee (x_1 \geq 1 \wedge x_2 \geq 1)
    \vee (x_1 \geq 0 \wedge x_2 \geq 0).
\]
Monadic decompositions have numerous applications in formal verification 
including string analysis \cite{DBLP:journals/jacm/VeanesBNB17,DBLP:conf/cade/HagueLRW20}
and query optimization in constraint databases \cite{GRS01,CDB-book}.
\citet{DBLP:journals/jacm/VeanesBNB17} gave a generic
semi-algorithm that is guaranteed to output a monadic decomposition of SMT
formulas, if such a decomposition exists. To make this semi-algorithm
terminating, one may incorporate a monadic decomposability check, which 
exists for numerous theories 
\cite{DBLP:journals/jacm/VeanesBNB17,DBLP:journals/tocl/Libkin03,BHLLN19,DBLP:conf/cade/HagueLRW20,lics22-ramsey}.
However, most of these algorithms have very high computational complexity, and 
for some theories the precise computational complexity is still an open problem.
Recently, \citet{DBLP:conf/cade/HagueLRW20} have shown
that monadic decomposability of quantifier-free LIA formulas is $\coNP$-complete,
in contrast to the previously known double exponential-time algorithm 
\cite{DBLP:journals/tocl/Libkin03}. 
In case of quantifier-free LRA and LIRA the precise complexity is still open. 
Although both of which can be shown to be decidable in $\PSPACE$ \cite{DBLP:journals/corr/abs-2304-11034}.

\subsubsection*{Contributions.} The main contribution of our paper is
new algorithms for eliminating Ramsey quantifiers 
for three common SMT theories: Linear Integer
Arithmetic (LIA), Linear Real Arithmetic (LRA), and Linear Integer Real 
Arithmetic (LIRA). If we restrict to existential fragments, the algorithms 
run in polynomial time and produce formulas of linear size.
[Here, in the definition of size we assume that every variable occurrence has length one.]
As a consequence,
SMT over these theories can be extended with Ramsey quantifiers 
only with a small overhead on SMT-solvers. Our results substantially improve 
the complexity of the elimination procedures of Ramsey quantifiers from
\cite{schmerlsimpson1982,DangI02}, which run in at least double exponential
time. This has direct applications in
proving program termination (especially, connected to well-foundedness checks)
and monadic decomposability (including, the precise complexity for 
LIA/LRA/LIRA). We detail our contributions below.

\paragraph{Key novel ingredients:} We circumvent the high complexities 
of \cite{schmerlsimpson1982,DangI02} as follows. The \emph{first key 
ingredient} is a procedure to eliminate existential quantifiers \emph{in the context of Ramsey quantifiers}: We prove that any formula
\begin{equation} \ram\bm{x},\bm{y}\colon\exists\bm{w}\colon\varphi(\bm{x},\bm{y},\bm{w}) \label{intro-before-qe}\end{equation}
with some quantifier-free $\varphi$ and quantifier block $\exists\bm{w}$ is equivalent to 
\begin{equation} \ram (\bm{x},\bm{r},\bm{s}),(\bm{y},\bm{t},\bm{u})\colon \varphi(\bm{x},\bm{y},\bm{s}+\bm{t}). \label{intro-after-qe}\end{equation}
Note that the formula \eqref{intro-before-qe} says that there exists a sequence
$\bm{a}_1,\bm{a}_2,\ldots$ of vectors such that for any $i<j$, there exists a
$\bm{b}_{i,j}$ with $\varphi(\bm{a}_i,\bm{a}_j,\bm{b}_{i,j})$. The equivalence
says that if such $\bm{b}_{i,j}$ exist, then there are $\bm{b}_i,\bm{b}'_i$
such that one can choose $\bm{b}_{i,j}:=\bm{b}'_i+\bm{b}_j$ to satisfy
$\varphi$. This comes as a surprise, because instead of needing to choose a
vector for each \emph{edge} of an infinite clique, it suffices to merely choose
an additional vector at each \emph{node}. This non-obvious structural result
about infinite cliques yields an algorithmically extremly simple elimination of
quantifiers, which just replaces \eqref{intro-before-qe} with
\eqref{intro-after-qe}.

Our \emph{second key ingredient} allows us to express the existence of an infinite directed clique in an \emph{existential} formula. Very roughly speaking, \citet{DangI02} express the existence of an infinite directed clique by saying that for each $k$, there exists a vector $\bm{a}_k$ for which some measure is at least $k$. By choosing appropriate measures, this ensures that the sequence $\bm{a}_1,\bm{a}_2,\ldots$ has certain unboundedness or convergence properties. In contrast, our formula expresses the existence of vectors $\bm{a}$, $\bm{d}_c$, and $\bm{d}_\infty$, such that the sequence 
\begin{equation} \bm{a}_k=\bm{a}-\tfrac{1}{k}\bm{d}_c+k\bm{d}_\infty \label{intro-general-form}\end{equation}
has an infinite clique as a subsequence. Here, the vector $\bm{d}_c$ is used to
enable convergence behavior (and is not needed in the case of LIA). 
Note that this is possible despite the fact that there are
formulas $\varphi$ for which no infinite $\varphi$-clique can be written in the
form above. However, one can always find an infinite $\varphi$-clique as a
subsequence of a sequence as in \eqref{intro-general-form}.

\paragraph{Proving termination/non-termination.}
Since well-foundedness of an inductive relation can be expressed by means of
an SMT formula with a Ramsey quantifier, our quantifier elimination procedure
yields a formula $\psi$ without Ramsey quantifiers, whose size is linear in the 
size of the original formula
$\varphi$ expressing the verification condition. This means that we can prove 
termination by simply checking satisfiability of $\psi$, which can be checked
easily by an SMT-solver. Similarly, if we provide an \emph{underapproximation}
$T \subseteq R^+$, we may use this to prove non-termination of $R$ by simply
checking satisfiability of
$
    T(\bm x,\bm x) \vee \ram \bm x,\bm y \colon T(\bm x,\bm y). 
$
By the same token, the Ramsey quantifier can be eliminated using our algorithm,
resulting in a formula of linear size that can be easily handled by SMT-solvers.

In fact, one can combine our results with various semi-algorithms for
computing approximations of reachability relations (e.g.
\cite{BFLP08,BFLS05,BH06,BLW03,TORMC}),
yielding a semi-algorithm for deciding
termination/non-termination with completeness guarantee for many classes of
infinite-state systems operating over integer and real variables. These include 
classes of hybrid systems and timed systems (e.g. timed pushdown systems), 
reversal-bounded counter systems,  
and continuous vector addition systems with states. For these, we also
obtain tight computational complexity for the problem.

\paragraph{Monadic decomposition.}
Our procedure reduces a monadic decomposability check for a LIA, LRA, or LIRA formula
to linearly many unsatisfiability queries over the same theory. As before, the
resulting formulas without Ramsey quantifiers is of linear size, which can be
handled easily by SMT-solvers. This reduction also shows that monadic
decomposability for LIA/LRA/LIRA is in $\coNP$, which can be shown to be the
precise complexity for the problems. The $\coNP$ complexity for LIA was shown
already by \citet{DBLP:conf/cade/HagueLRW20}, but with a completely different
reduction (and no experimental validation). 
The $\coNP$ complexity of monadic decomposability for LRA/LIRA is new and answers
the open questions posed by \citet{DBLP:journals/jacm/VeanesBNB17} and \citet{DBLP:journals/corr/abs-2304-11034}.

\paragraph{Implementation.}
We have implemented a prototype of our elimination algorithms for LIA, LRA, and LIRA
and tested it on two sets of micro-benchmarks.
The first benchmarks contain examples where a single Ramsey quantifier has to be eliminated.
Such formulas can for example be derived from program (non-)termination.
With the second benchmarks we use our algorithms to check monadic decomposability as described above.
Here, we compare our algorithm to the ones in \cite{DBLP:journals/jacm/VeanesBNB17} and \cite{MSL21}.
For both sets of benchmarks we obtain promising results.

\OMIT{
\paragraph{Organization.} We begin in Section with a few more detailed examples
illustrating the problems of termination/non-termination proving, and our
approach of removing Ramsey quantifiers.
}
 
\section{More Detailed Examples}\label{sec:examples}

In this section, we give concrete examples illustrating the problems of 
proving termination and non-termination, and how these give rise to 
verification conditions involving Ramsey quantifiers. We then discuss how 
our algorithms eliminate Ramsey quantifiers from these verification conditions.

\paragraph*{Proving termination.} We consider a simplified version of 
McCarthy 91 program \cite{MP70}. The program has two integer variables $n,m$
and applies the following rules till termination:
\begin{enumerate}
\item $n := n-1$ and $m := m-1$, if $n > 0$ and $m \geq 0$.
\item $n := n+1$ and $m := m+2$, if $n > 0$ and $m < 0$.
\end{enumerate}
The interesting case in this termination proof is when $-1 \leq m \leq 1$.
For simplicity, we will only deal with this. In the sequel, we will write $R
\subseteq \N^2 \times \N^2$ to denote the relation generated by the program
restricted to $-1 \leq m \leq 1$.

To prove termination, we will need to annotate the program with an inductive
relation $T \subseteq \N^2 \times \N^2$ that is well-founded. Define $T$
as the conjunction of $n > 0 \wedge -1 \leq m,m' \leq 1 \wedge n' \geq 0$ and 
disjunctions of 
relations $T_1,\ldots,T_6$ as specified below.
\begin{align*}
    T_1 & := m' = 0 \wedge n' = n \wedge m = -1 \wedge n = 1 &
    T_2 & := m' = 1 \wedge n' = n+1 \wedge m = -1 \wedge n \geq 1
        \\
    T_3 & := m' > m \wedge n' = n \wedge n \geq 2 &
    T_4 & := n' < n \wedge n' \geq 0 \wedge m \leq 0 \\
    T_5 & := n' < n \wedge n' \geq 0 \wedge m = 1 \wedge m' \geq 0 &
    T_6 & := n' < n-1 \wedge n' \geq 0 \wedge m = 1 \wedge m' = -1
\end{align*}
The condition that $T$ is inductive is easily phrased as satisfiability
of a quantifier-free LIA formula:
\[
    [R(n,m,n',m') \wedge \neg T(n,m,n',m')] \vee [T(n,m,n',m') \wedge
        R(n',m',n'',m'') \wedge \neg T(n,m,n'',m'')].
\]
We need to prove unsatisfiability of this formula, which can be easily checked 
using a LIA solver, which is supported by major
SMT-solvers (e.g. Z3 \cite{MB08}). To prove well-foundedness of $T$, we consider
two cases. The looping case also easily translates into a LIA formula:
\[
    T(n,m,n',m') \wedge T(n',m',n',m').
\]
Again, we need to prove that this is unsatisfiable. The non-looping case is 
the one that requires a Ramsey quantifier where we need to prove unsatisfiability of:
\[
    \ram (n,m),(n',m') \colon T(n,m,n',m').
\]

\paragraph{Proving non-termination.}

\begin{wrapfigure}[14]{r}{0.5\textwidth}
\centering
\begin{minipage}{0.3\textwidth}
\begin{algorithm}[H]\footnotesize
\Real $x_1 \gets$ input-real()\;
\Int $x_2 \gets$ input-int()\;
\Assert $x_1 > 0$\;
\While{$x_2 > 0$}{
	\Real $t_1 \gets$ input-real()\;
	\Assert $t_1 \ge 0.5 x_1 + 0.5$\;
	$x_1 \gets t_1$\;
	\Int $t_2 \gets$ input-int()\;
	\Assert $t_2 \ge 0$\;
	$x_2 \gets x_2 - \lfloor x_1 \rfloor - t_2$\;
}
\end{algorithm}
\end{minipage}
\caption{Example of a non-terminating program.}
\label{fig:program}
\end{wrapfigure}

Let us now present an example of a program (see \Cref{fig:program}) where Ramsey quantifiers can be used to prove non-termination.

The reachability relation $\to_1$ for $x_1$ and $x_2$ after the first iteration of the while-loop is
$(x_1,x_2) \to_1 (y_1,y_2)$ such that
$x_1 > 0 \wedge x_2 > 0 \wedge y_1 \ge 0.5x_1+0.5 \wedge y_2 \le x_2 - \lfloor x_1 \rfloor$.
It turns out that $\to_1$ is already an under-approximation of the reachability relation $\to^+$ after at least one iteration.
Thus, to show non-termination, it suffices to show that $\to_1$ and therefore $\to^+$ has an infinite clique.
For example we find the clique $(a_i,b_i)_{i \ge 1}$ with 
$a_1 = 0.5$, $a_{i+1} = 0.5 a_{i}+0.5$, and $b_i = 1$ for all $i \ge 1$
which corresponds to choosing $x_1 = 0.5$ and $x_2 = 1$ at the beginning and $t_1 = 0.5x_1+0.5$ and $t_2 = 0$ in each iteration.
Here we can see that $(a_i)_{i\ge1}$ converges against 1 but never reaches 1, which means that $\lfloor a_i \rfloor$ is always 0 for all $i \ge 1$.

\paragraph{Illustration of how to remove a Ramsey quantifier.}
Let us see an example of Ramsey quantifier elimination. Consider the following formula:
\[ \varphi(\bm x,\bm y,\bm z)\quad =\quad x_1 + \frac{z_1-x_1}{2} \le y_1 \le z_1\quad \wedge \quad x_2+\frac{z_2-x_2}{2}\le y_2\le z_2 \quad \wedge \quad y_2=\lfloor y_1\rfloor, \]
in which $\bm x=(x_1,x_2)$, $\bm y=(y_1,y_2)$, and $\bm z=(z_1,z_2)$. Now we claim that
$\ram \bm x,\bm y\colon \varphi(\bm x,\bm y,\bm z)$ is equivalent to
\[ z_2=\lfloor z_1\rfloor \quad \vee \quad (z_1=\lfloor z_1\rfloor \wedge z_2=z_1-1). \]
Note that $\ram\bm x,\bm y\colon\varphi(\bm x,\bm y,\bm z)$ expresses that
there exists a sequence $\bm{a}_1,\bm{a}_2,\ldots\in\R^2$ such that the first
components converge from below against $z_1$: The first conjunct in $\varphi$ requires the
first component of $\bm a_j$ to have at most half the distance to $z_1$ as the
first component of $\bm a_i$, for every $i<j$. Similarly, the second conjunct
forces the second components to converge against $z_2$.

Furthermore, the third conjunct requires the second components of $\bm a_i$ to
be the floor of the first component of $\bm a_i$, for every $i$. Thus, if $z_1$
is not an integer, then the first components of $\bm a_1,\bm a_2,\ldots$ will
eventually be between $\lfloor z_1\rfloor$ and $z_1$. Thus, the second
components must eventually be equal to $\lfloor z_1\rfloor$ and thus also their
limit: $z_2=\lfloor z_1\rfloor$. However, if $z_1$ is an integer, then there is
another option: The first components can all be strictly smaller than
$z_1=\lfloor z_1\rfloor$. But then the second components must eventually be
equal to $\lfloor z_1\rfloor-1$, and thus $z_2=z_1-1$.

\section{Preliminaries}\label{sec:preliminaries}

We denote a vector with components $(x_1,\dots,x_k)$ of dimension $k$ with a boldface letter $\bm x$ and
for numbers $n$ we write $\bm n$ for a vector $(n,\dots,n)$ of appropriate dimension.
On vectors $\bm x$ and $\bm y$ of dimension $k$ we define the usual pointwise partial order $\bm x \le \bm y$
such that $x_i \le y_i$ for all $1 \le i \le k$.
Moreover, we define $\bm x \ll \bm y$ if $x_i < y_i$ for \emph{every} $i$.

To reduce the usage of parentheses, we assume the binding strengths of logical operators to be
$\neg$, $\wedge$, $\vee$, $\to$
in decreasing order and quantifiers bind the weakest.

We define the \emph{size} of a formula by the length of its usual encoding 
where we assume that every variable occurrence has length one.
In the following we formally only define formulas with constants $0$ and $1$,
but we will also use arbitrary constants that, when encoded in binary,
can be eliminated with only a linear blow-up in the above size definition. 
Note that for implementations, it would also make sense to measure the length of writing the formula using a fixed alphabet, 
which would incur a logarithmic-length string per variable occurrence.

\subsection*{Linear Integer Arithmetic}
\emph{Linear Integer Arithmetic} (LIA) is defined as the first-order theory with the structure $\ZLin$.
LIA is also called \emph{Presburger arithmetic} and we will use these terms interchangeably.
We will only work on the existential fragment of LIA, i.e., formulas of the form $\exists \bm x \colon \varphi(\bm x,\bm z)$
where the variables in $\bm x$ are bound by the existence quantifier and $\bm z$ is a vector of free variables.
\begin{proposition}[\cite{borosh1976bounds}]\label{prop:sat-presburger}
Satisfiability of existential formulas in LIA is $\NP$-complete.
\end{proposition}
To admit quantifier elimination, one has to enrich the structure $\ZLin$ by modulo constraints.
A modulo constraint is a binary predicate $\equiv_e$ with $e > 0$ such that $s \equiv_e t$ is fulfilled if and only if $e | s-t$.
Note that modulo constraints are definable in $\ZLin$ using existence quantifiers,
which means that $\langle \Z;+,<,0,1,(\equiv_e)_{e > 0} \rangle$ is still a structure for LIA. The following was famously shown by \citet{presburger1929} (see \cite{weispfenning1997} for complexity considerations):
\begin{proposition}\label{prop:qe-presburger}
LIA with the structure $\langle \Z;+,<,0,1,(\equiv_e)_{e > 0} \rangle$ admits quantifier elimination.
\end{proposition}

\subsection*{Linear Real Arithmetic}
In addition to the integers, we also consider linear arithmetic over the reals.
\emph{Linear Real Arithmetic} (LRA) is defined as the first-order theory with the structure $\RLin$.
As for LIA, the satisfiability problem for the existential fragment of LRA is $\NP$-complete~\cite[Corollary~3.4]{sontagRealAdditionPolynomial1985}:
\begin{proposition}\label{prop:sat-reals}
Satisfiability of existential formulas in LRA is $\NP$-complete.
\end{proposition}
Moreover, quantifiers can already be eliminated over the structure $\RLin$. This goes back to \citet{fourier1826} and was rediscovered several times thereafter~\cite{williams1986fourier}:
\begin{proposition}\label{prop:qe-reals}
LRA with the structure $\RLin$ admits quantifier elimination.
\end{proposition}

\subsection*{Linear Integer Real Arithmetic}
We define \emph{Linear Integer Real Arithmetic} (LIRA) as the first-order theory with the structure
$\langle \R;\lfloor . \rfloor,+,<,0,1 \rangle$
where $\lfloor r \rfloor$ denotes the greatest integer smaller than or equal to $r \in \R$.
In terms of full first-order logic, this logic is equally expressive as the first-order logic over the structure $\langle\R;\Z,+,<,0,1\rangle$. Here, we focus on $\langle \R;\lfloor . \rfloor,+,<,0,1 \rangle$, because its existential fragment is expressively complete~\cite[Theorem~3.1]{Weispfenning99}.
Note that by using $x = \lfloor x \rfloor$, we can extend LIRA to allow two sorts of variables: real and integer variables.
For a vector $\bm x = (x_1,\dots,x_n)$ of variables let $\bm x^\mathrm{i/r}$ denote the vector $(\bm x^\mathrm{int},\bm x^\mathrm{real})$
where $\bm x^\mathrm{int} = (x_1^\mathrm{int},\dots,x_n^\mathrm{int})$ is a vector of integer variables and
$\bm x^\mathrm{real} = (x_1^\mathrm{real},\dots,x_n^\mathrm{real})$ is a vector of real variables.
Two vectors $\bm{x}$ and $\bm{y}$ of dimension $n$ are said to have the same \emph{type} 
if for all $i \in [1,n]$ we have that $x_i$ and $y_i$ are both real or integer variables. 
The \emph{separation} of an existential formula $\exists x_1,\dots,x_n \colon \varphi(x_1,\dots,x_n,z_1,\dots,z_m)$ in LIRA is defined as
\begin{align*}
\exists \bm x^\mathrm{i/r} \colon
\varphi(\bm x^\mathrm{int}+ \bm x^\mathrm{real}, \bm z^\mathrm{int} + \bm z^\mathrm{real}) \wedge 
0 \le \bm x^\mathrm{real} < 1 \wedge 0 \le \bm z^\mathrm{real} < 1
\end{align*}
where $x_i^\mathrm{int}$, $z_j^\mathrm{int}$ are fresh integer variables and $x_i^\mathrm{real}$, $z_j^\mathrm{real}$ are fresh real variables
that express the integer and real part of $x_i$ and $z_j$.
If $x_i$ (resp. $z_j$) is an integer variable, we add the constraint $x_i^\mathrm{real} = 0$ (resp. $z_j^\mathrm{real} = 0$) to the separation.
We say that an existential formula in LIRA is \emph{decomposable} if its separation can be written as an existentially quantified Boolean combination of Presburger and LRA formulas (called \emph{decomposition}).
\begin{lemma}\label{lem:decomposition}
Every existential formula in LIRA is decomposable. Moreover, its decomposition is of linear size and can be computed in polynomial time.
\end{lemma}
\begin{proof}
Let $\psi = \exists x_1,\dots,x_n \colon \varphi(x_1,\dots,x_n,z_1,\dots,z_m)$ be an existential formula in LIRA.
By introducing new existentially quantified variables, we can assume that every atom of $\varphi$ is of one of the following forms:
(i) $x = 0$, (ii) $x = 1$, (iii) $x+y=z$, (iv) $x < 0$, (v) $x = \lfloor y \rfloor$. 
Note that the size of the formula is still linear (even if the coefficients are given in binary).
Let $\varphi'(\bm x^\mathrm{i/r},\bm z^\mathrm{i/r})$ be the formula obtained from 
$\varphi(\bm x^\mathrm{int}+ \bm x^\mathrm{real}, \bm z^\mathrm{int} + \bm z^\mathrm{real})$ by replacing every
\begin{itemize}
\item $x^\mathrm{int} + x^\mathrm{real} = 0$ by $x^\mathrm{int} = 0 \wedge x^\mathrm{real} = 0$,
\item $x^\mathrm{int} + x^\mathrm{real} = 1$ by $x^\mathrm{int} = 1 \wedge x^\mathrm{real} = 0$,
\item $x^\mathrm{int}+x^\mathrm{real}+y^\mathrm{int}+y^\mathrm{real} = z^\mathrm{int}+z^\mathrm{real}$ by
\begin{align*}
& (x^\mathrm{real}+y^\mathrm{real}<1 \to x^\mathrm{int}+y^\mathrm{int} = z^\mathrm{int} \wedge x^\mathrm{real}+y^\mathrm{real}=z^\mathrm{real})\ \wedge\\
& (x^\mathrm{real}+y^\mathrm{real} \ge 1 \to x^\mathrm{int}+y^\mathrm{int}+1 = z^\mathrm{int} \wedge x^\mathrm{real}+y^\mathrm{real}-1=z^\mathrm{real}),
\end{align*}
\item $x^\mathrm{int}+x^\mathrm{real}<0$ by $x^\mathrm{int} < 0$, and
\item $x^\mathrm{int}+x^\mathrm{real} = \lfloor y^\mathrm{int}+y^\mathrm{real} \rfloor$ by $x^\mathrm{real}=0 \wedge x^\mathrm{int} = y^\mathrm{int}$.
\end{itemize}
Thus, $\varphi'$ is a Boolean combination of formulas that either only involve integer variables or real variables.
Now the separation of $\psi$ is equivalent to
\[
\exists \bm x^\mathrm{i/r} \colon \varphi'(\bm x^\mathrm{i/r},\bm z^\mathrm{i/r})
\wedge 0 \le \bm x^\mathrm{real} < 1 \wedge 0 \le \bm z^\mathrm{real} < 1
\]
where we add $x_i^\mathrm{real} = 0$ if $x_i$ is an integer variable and $z_j^\mathrm{real} = 0$ if $z_j$ is an integer variable,
which is a linear sized decomposition.
\end{proof}

\begin{proposition}\label{prop:sat-mixed}
Satisfiability of existential formulas in LIRA is $\NP$-complete.
\end{proposition}
\begin{proof}
The $\NP$ lower bound is inherited from the Presburger (\Cref{prop:sat-presburger}) and LRA (\Cref{prop:sat-reals}) case.
For the upper bound let $\varphi$ be an existential formula in LIRA.
We first apply \Cref{lem:decomposition} to compute a decomposition $\psi$ of $\varphi$ in polynomial time.
Then we guess truth values for the Presburger and LRA subformulas of $\psi$ and verify the guesses in $\NP$ using \Cref{prop:sat-presburger,prop:sat-reals}.
Since $\varphi$ and its decomposition $\psi$ are equisatisfiable, 
it remains to check whether the truth values satisfy $\psi$ in order to decide satisfiability of $\varphi$.
\end{proof}

\subsection*{Ramsey quantifier}
Let $\bm x$ and $\bm y$ be two vectors of variables of the same type.
For a formula $\varphi$ in LIRA we define $\ram \bm{x},\bm{y} \colon \varphi(\bm{x},\bm{y},\bm{z})$ 
as the formula that is satisfied by a valuation $\bm c$ of $\bm z$ if and only if
there exists a sequence $(\bm a_i)_{i \ge 1}$ of pairwise distinct valuations of $\bm x$ (and $\bm y$) such that
$\varphi(\bm a_i, \bm a_j, \bm c)$ holds for all $i < j$.
The sequence $(\bm a_i)_{i \ge 1}$ with the above properties is called an \emph{infinite clique} of $\varphi$ with respect to $\bm c$.
The \emph{infinite clique problem} asks given a formula $\varphi(\bm x, \bm y)$, where $\bm x$ and $\bm y$ have the same type,
whether $\varphi$ has an infinite clique.

The infinite version of Ramsey's theorem can be formulated over graphs as follows:
\begin{theorem}[\cite{Ramsey30}]
Any complete infinite graph whose edges are colored with finitely many colors contains an infinite monochromatic clique.
\end{theorem}
We will often use the fact that by Ramsey's theorem $\ram \bm{x},\bm{y} \colon \varphi(\bm{x},\bm{y},\bm{z}) \vee \psi(\bm{x},\bm{y},\bm{z})$ is equivalent to
$(\ram \bm{x},\bm{y} \colon \varphi(\bm{x},\bm{y},\bm{z})) \vee (\ram \bm{x},\bm{y} \colon \psi(\bm{x},\bm{y},\bm{z}))$.

\section{Eliminating existential quantifiers}\label{sec:qe}
The first step in our elimination of the Ramsey quantifier in $\ram \bm x,\bm
y\colon \psi(\bm x, \bm y, \bm z)$ is to reduce to the case where $\psi$ is
quantifier-free. In LIA and LRA, there are procedures to convert $\psi$ into a
quantifier-free equivalent (\cref{prop:qe-presburger,prop:qe-reals}), but these
incur a doubly exponential blow-up~\cite{weispfenning1997}.  Instead, we will
show the following (perhaps surprising) equivalence:
\begin{theorem}\label{thm:qe-mixed}
Let $\varphi$ be an existential formula in LIRA. Then the formulas 
\begin{align*}
\ram \bm{x},\bm{y}\colon\exists \bm{w}\colon \varphi(\bm{x},\bm{y},\bm{w},\bm{z}) 
&& \text{and} &&
\ram (\bm{x},\bm{v}_1,\bm{v}_2),(\bm{y},\bm{w}_1,\bm{w}_2)\colon \varphi(\bm{x},\bm{y},\bm{v}_1+\bm{w}_2,\bm{z}) \wedge \bm x \ne \bm y
\end{align*}
are equivalent where $\bm v_1, \bm v_2, \bm w_1, \bm w_2$ have the same type as $\bm w$.
\end{theorem}
Thus, if we write $\psi(\bm x,\bm y, \bm z)$ in prenex form as $\exists \bm
w\colon \varphi(\bm x,\bm y,\bm w,\bm z)$ with a quantifier-free $\varphi$,
then \Cref{thm:qe-mixed} allows us to eliminate the block $\exists\bm{w}$ of
quantifiers by moving $\bm{w}$ under the Ramsey quantifier.  Note that both
formulas express the existence of an infinite clique.  The left says that for
every edge $\bm{x}\to\bm{y}$ in the clique, we can choose a vector $\bm{w}$
such that $\varphi(\bm{x},\bm{y},\bm{w},\bm{z})$ is satisfied.  The right
formula says that $\bm{w}$ can be chosen in a specific way: It says that for
each node, one can choose two vectors ($\bm{w}_1$,$\bm{w}_2$) such that for
each edge $\bm{x}\to\bm{y}$, the vector $\bm{w}$ can be the sum of $\bm{w}_1$
for $\bm{x}$ and $\bm{w}_2$ for $\bm{y}$. Thus,  the right-hand formula clearly
implies the left-hand formula. The challenging direction is to show that the
left-hand formula implies the right-hand formula.

The rest of this section is devoted to proving \cref{thm:qe-mixed}.
\subsection*{Presburger arithmetic}
We start with LIA, a.k.a.\ Presburger arithmetic.
\begin{theorem}\label{thm:qe-presburger}
Let $\varphi$ be an existential formula in Presburger arithmetic. Then the formulas 
\begin{align*} &\ram \bm{x},\bm{y}\colon\exists \bm{w}\colon \varphi(\bm{x},\bm{y},\bm{w},\bm{z}) &&\text{and}~~~
&&\ram (\bm{x},\bm{v}_1,\bm{v}_2),(\bm{y},\bm{w}_1,\bm{w}_2)\colon \varphi(\bm{x},\bm{y},\bm{v}_1+\bm{w}_2,\bm{z})\wedge \bm{x}\ne\bm{y}\end{align*}
are equivalent.
\end{theorem}
To prove \cref{thm:qe-presburger}, it suffices to prove it in case $\bm{w}$ consists of just one variable $w$: 
Then, \cref{thm:qe-presburger} follows by induction.
\begin{lemma}\label{lem:qe-presburger-single}
Let $\varphi$ be an existential formula in Presburger arithmetic. Then the formulas 
\begin{align}\label{eq-qe-single}
\ram \bm{x},\bm{y}\colon\exists w \colon \varphi(\bm{x},\bm{y},w,\bm{z}) 
&&\text{and}&&
\ram (\bm{x},v_1,v_2),(\bm{y},w_1,w_2)\colon \varphi(\bm{x},\bm{y},v_1+w_2,\bm{z})\wedge \bm{x}\ne\bm{y}
\end{align}
are equivalent.
\end{lemma}

\subsection*{Presburger: simple formulas}
Let $\bm{u}$ be a vector of variables and $w$ be a variable. We say that a
Presburger formula $\varphi(\bm{u},w)$ is \emph{$w$-simple} if it is a
Boolean combination of formulas of the form $\bm{r}^\top\bm{u}+c<w$, $w<\bm{r}^\top
\bm{u}+c$, and modulo constraints over $\bm{u}$ and $w$, where $\bm{r}$ is a
vector over $\Z$, and $c\in\Z$. 
\begin{lemma}\label{lem:qe-simple}
Let $\varphi(\bm{x},\bm{y},w,\bm{z})$ be $w$-simple. Then the
formulas 
\begin{align*} 
\ram \bm{x},\bm{y}\colon\exists w\colon \varphi(\bm{x},\bm{y},w,\bm{z}) 
&& \text{and}&&
\ram (\bm{x},v_1,v_2),(\bm{y},w_1,w_2)\colon \varphi(\bm{x},\bm{y},v_1+w_2,\bm{z})\wedge \bm{x}\ne\bm{y}
\end{align*}
are equivalent.
\end{lemma}
\begin{proof}
We first move all negations in $\varphi$ inwards to the atoms and possibly negate them (which for modulo constraints introduces disjunctions).
We then bring $\varphi$ into disjunctive normal form
and move the Ramsey quantifier and existence quantifier into the disjunction.
Since $\varphi$ is simple, we can assume that it is a conjunction of formulas
\[ \alpha_i(\bm{x})+\beta_i(\bm{y})+\gamma_i(\bm{z})+h_i < w \]
for $i=1,\ldots,n$ and
\[ w < \alpha'_j(\bm{x})+\beta'_j(\bm{y})+\gamma'_j(\bm{z})+h'_j \]
for $j=1,\ldots,m$, and modulo constraints
\[ \delta_i(\bm{x},\bm{y},w,\bm{z})\equiv_{e_i} d_i\]
for $i=1,\ldots,k$. Here,
$\alpha_i,\beta_i,\gamma_i,\alpha'_j,\beta'_j,\gamma'_j,\delta_i$ are linear
functions.
Let $f_i(\bm{x},\bm{y},\bm{z}) := \alpha_i(\bm{x})+\beta_i(\bm{y})+\gamma_i(\bm{z})+h_i$ for all $i \in [1,n]$ and
$f'_j(\bm{x},\bm{y},\bm{z}) := \alpha'_j(\bm{x})+\beta'_j(\bm{y})+\gamma'_j(\bm{z})+h'_j$ for all $j \in [1,m]$.
In the following we assume that $n,m > 0$; the other cases are simpler and can be handled similarly.

Assume $\bm{c} \in \Z^{|\bm{z}|}$ satisfies the left-hand formula, i.e.,
there is an infinite sequence $(\bm{a}_i)_{i \ge 1}$ of pairwise distinct vectors over $\Z$ such that
for all $i < j$ there exists $b_{i,j} \in \Z$ such that $\varphi(\bm{a}_i,\bm{a}_j,b_{i,j},\bm{c})$ holds.
By Ramsey's theorem we can take an infinite subsequence such that we can assume that
$f_1(\bm{a}_i,\bm{a}_j,\bm{c}) \le \dots \le f_n(\bm{a}_i,\bm{a}_j,\bm{c})$ and
$f'_1(\bm{a}_i,\bm{a}_j,\bm{c}) \le \dots \le f'_m(\bm{a}_i,\bm{a}_j,\bm{c})$
for all $i<j$.
Thus, it suffices to consider the greatest lower bound $f_n(\bm{a}_i,\bm{a}_j,\bm{c})$ and
the smallest upper bound $f'_1(\bm{a}_i,\bm{a}_j,\bm{c})$ on $w$.
Let $f := f_n, \alpha := \alpha_n, \beta := \beta_n, \gamma := \gamma_n$, and $h := h_n$.
Let $N := e_1 \cdots e_k$ be the product of all moduli where we set $N := 1$ if $k = 0$.
First observe that $b_{i,j}$ can always be chosen from the interval $[f(\bm{a}_i,\bm{a}_j,\bm{c})+1,f(\bm{a}_i,\bm{a}_j,\bm{c})+N]$ for all $i<j$.
Since this interval has fixed length $N$, by Ramsey's theorem we can restrict to an infinite subsequence such that
there is a constant $r \in [1,N]$ such that $f(\bm{a}_i,\bm{a}_j,\bm{c}) + r = b_{i,j}$ for all $i<j$.
Now if we set $b_i^1 := \alpha(\bm{a}_i)$ and $b_i^2 := \beta(\bm{a}_i) + \gamma(\bm{c}) + h + r$ for $i \ge 1$, the infinite sequence
$(\bm{a}_i,b_i^1,b_i^2)_{i \ge 1}$ satisfies $\varphi(\bm{a}_i,\bm{a}_j,b_i^1 + b_j^2,\bm{c})$ for all $i < j$ as desired.
\end{proof}

\subsection*{Presburger: general formulas}
Let us now prove \cref{lem:qe-presburger-single} for general existential
Presburger formulas.  Observe that if we can show equivalence of the formulas
in \cref{eq-qe-single} for quantifier-free $\varphi$ with modulo constraints, then the same follows for
general $\varphi$: Since for each $\varphi$, there exists an equivalent
quantifier-free $\varphi'$ with modulo constraints, we can apply the equivalence in
\cref{lem:qe-presburger-single} to $\varphi'$, which implies the same for
$\varphi$ itself. Therefore, we may assume that $\varphi$ is quantifier-free,
but contains modulo constraints.

We now modify $\varphi$ as in the standard
quantifier elimination procedure for Presburger arithmetic. To this end, we define the ``$w$-simplification'' of a quantifier-free formula $\theta(\bm{u},w)$ with modulo constraints that has free variables $\bm{u}$ and $w$. 
This means, $\theta$ is a Boolean combination of inequalities of the form
$\bm{r}^\top \bm{u} + c \sim s w$, where $\mathord{\sim} \in \{\mathord{<},\mathord{>}\}$, and modulo
constraints $\bm{r}^\top \bm{u} + s w \equiv_e c$ for some vector
$\bm{r}$ and $c,s\in\Z$. 
(Note that in Presburger arithmetic equality can be expressed by a conjunction of two strict inequalities.)
Let $N$ be the least common multiple of all coefficients $s$ of $w$ in these constraints. We obtain $\theta'$ from $\theta$ by replacing each inequality $\bm{r}^\top\bm{u}+c\sim sw$ with $\tfrac{N}{s}\bm{r}^\top\bm{u}+\tfrac{N}{s}c\sim w$
and replacing each modulo constraint $\bm{r}^\top \bm{u} + s w \equiv_e c$ with $\tfrac{N}{s}\bm{r}^\top\bm{u}+w\equiv_{\tfrac{N}{s}e} \tfrac{N}{s}c$. Now the \emph{$w$-simplification} of $\varphi$ is the pair $(\psi,N)$, where $\psi(\bm{u},w)=\varphi'(\bm{u},w)\wedge w\equiv_N 0$. Then clearly, $\psi$ is $w$-simple and for every integer vector $\bm{a}$ and $b\in\Z$, we have
\[ \theta(\bm{a},b)~~\text{if and only if}~~\psi(\bm{a},Nb) \]
and moreover, $\psi(\bm{a},b)$ implies that $b$ is a multiple of $N$.

Now suppose $\varphi(\bm{x},\bm{y},w,\bm{z})$ is quantifier-free, but contains modulo constraints. Moreover, let $\psi(\bm{x},\bm{y},w,\bm{z})$ and $N$ be the $w$-simplification of $\varphi$. To show \cref{lem:qe-presburger-single}, let us assume the left-hand formula in \cref{eq-qe-single} is satisfied for some integer vector $\bm{c}$. Then $\ram \bm{x},\bm{y}\colon \exists w\colon \psi(\bm{x},\bm{y},w,\bm{c})$ holds, because we can multiply the witness values by $N$.
By \cref{lem:qe-simple}, this implies that 
\[ \ram (\bm{x},v_1,v_2),(\bm{y},w_1,w_2)\colon \psi(\bm{x},\bm{y},v_1+w_2,\bm{c}) \]
is satisfied, meaning 
there exists a sequence $(\bm{a}_i,b_i,b'_i)_{i\ge 1}$ where $\bm{a}_1,\bm{a}_2,\ldots$ are pairwise distinct and where $\psi(\bm{a}_i,\bm{a}_j,b_i+b'_j,\bm{c})$ for every $i<j$.
By construction of $\psi$, this implies that $b_i+b'_j$ is a multiple of $N$ for every $i<j$ and therefore all the numbers $b_1,b_2,\ldots$ must have the same remainder modulo $N$, say $r\in[0,N-1]$, and all the numbers $b'_1,b'_2,\ldots$ must be congruent to $-r$ modulo $N$.
This means, the numbers $\bar{b}_i=(b_i-r)/N$ and
$\bar{b}'_i=(b'_i+r)/N$ must be integers. Then for every $i<j$, we have
$\psi(\bm{a}_i,\bm{a}_j,N(\bar{b}_i+\bar{b}'_j),\bm{c})$ and hence $\varphi(\bm{a}_i,\bm{a}_j,\bar{b}_i+\bar{b}'_j,\bm{c})$. Thus, the sequence $(\bm{a}_i,\bar{b}_i,\bar{b}'_i)_{i\ge 1}$ shows that
$\ram (\bm{x},v_1,v_2),(\bm{y},w_1,w_2)\colon \varphi(\bm{x},\bm{y},v_1+w_2,\bm{c})\wedge \bm{x}\ne\bm{y}$ is satisfied.

\subsection*{Linear Real Arithmetic}
We now turn to the case where $\varphi$ is a formula in LRA.
\begin{theorem}\label{thm:qe-real}
Let $\varphi$ be an existential formula in LRA. Then the formulas 
\begin{align*} &\ram \bm{x},\bm{y}\colon\exists \bm{w}\colon \varphi(\bm{x},\bm{y},\bm{w},\bm{z}) &&\text{and}~~~
&&\ram (\bm{x},\bm{v}_1,\bm{v}_2),(\bm{y},\bm{w}_1,\bm{w}_2)\colon \varphi(\bm{x},\bm{y},\bm{v}_1+\bm{w}_2,\bm{z})\wedge \bm{x}\ne\bm{y}\end{align*}
are equivalent.
\end{theorem}
We may assume that $\bm{w}$ consists of just one variable $w$: 
Then \cref{thm:qe-real} follows by induction.
\begin{lemma}\label{lem:qe-real-single}
Let $\varphi$ be an existential formula in LRA. Then the formulas 
\begin{align*}
\ram \bm{x},\bm{y}\colon\exists w \colon \varphi(\bm{x},\bm{y},w,\bm{z}) 
&&\text{and}&&
\ram (\bm{x},v_1,v_2),(\bm{y},w_1,w_2)\colon \varphi(\bm{x},\bm{y},v_1+w_2,\bm{z})\wedge \bm{x}\ne\bm{y}
\end{align*}
are equivalent.
\end{lemma}
\begin{proof}
Again by eliminating quantifiers in $\varphi$, bringing it into disjunctive normal form, and moving the quantifiers into the disjunction, 
we assume that $\varphi$ is a conjunction of formulas
\[ \alpha_i(\bm{x})+\beta_i(\bm{y})+\gamma_i(\bm{z})+h_i < w \]
for $i=1,\ldots,n$ and
\[ w < \alpha'_j(\bm{x})+\beta'_j(\bm{y})+\gamma'_j(\bm{z})+h'_j \]
for $j=1,\ldots,m$, and equality constraints
\[ w = \delta_i(\bm{x}) + \kappa_i(\bm{y}) + \lambda_i(\bm{z}) + d_i \]
for $i=1,\ldots,k$.
Here, $\alpha_i,\beta_i,\gamma_i,\alpha'_j,\beta'_j,\gamma'_j,\delta_i,\kappa_i,\lambda_i$ are linear functions with rational coefficients
and $h_i,h'_j,d_i \in \Q$ are constants.

Assume $\bm{c} \in \R^{|\bm{z}|}$ satisfies the left-hand formula, i.e.,
there is an infinite sequence $(\bm{a}_i)_{i \ge 1}$ of pairwise distinct vectors over $\R$ such that
for all $i < j$ there exists $b_{i,j} \in \R$ such that $\varphi(\bm{a}_i,\bm{a}_j,b_{i,j},\bm{c})$ holds.
Clearly, if $k > 0$, we can eliminate $w$ by replacing it by $\delta_1(\bm{x})+\kappa_1(\bm{y})+\lambda_1(\bm{z}) + d_1$.
Thus, setting $b_i := \delta_1(\bm{a}_i)$ and $b'_i := \kappa_1(\bm{a}_i)+\lambda_1(\bm{c}) + d_1$ for $i \ge 1$,
the sequence $(\bm{a}_i,b_i,b'_i)_{i\ge1}$ satisfies $\varphi(\bm{a}_i,\bm{a}_j,b_i+b'_j,\bm{c})$ for all $i<j$.
So assume $k=0$, i.e., $\varphi$ only contains lower and upper bounds on $w$.
We further assume that $n,m > 0$ since the other cases are obvious.
As in the Presburger case we can apply Ramsey's theorem so that we only have to consider the greatest lower bound
$\alpha(\bm{x})+\beta(\bm{y})+\gamma(\bm{z})+h$
and the smallest upper bound
$\alpha'(\bm{x})+\beta'(\bm{y})+\gamma'(\bm{z})+h'$
on $w$.
This means that $w$ can always be chosen to be the midpoint of this interval.
Therefore, if we set $b_i := (\alpha(\bm{a}_i) + \alpha'(\bm{a}_i)) / 2$ and 
$b'_i := (\beta(\bm{a}_i)+\gamma(\bm{c})+h+\beta'(\bm{a}_i)+\gamma'(\bm{c})+h')/2$ for $i \ge 1$,
the sequence $(\bm{a}_i,b_i,b'_i)_{i\ge1}$ satisfies $\varphi(\bm{a}_i,\bm{a}_j,b_i+b'_j,\bm{c})$ for all $i<j$.
\end{proof}

\subsection*{Linear Integer Real Arithmetic}
We are now ready to prove \Cref{thm:qe-mixed}.
As mentioned above, it suffices to prove the ``only if'' direction.
Let $\psi(\bm z)$ be the left-hand formula and $\bm c$ be a valuation of $\bm z$ that satisfies $\psi$.
By \Cref{lem:decomposition} there exists a decomposition 
$\varphi'(\bm x^\mathrm{i/r},\bm y^\mathrm{i/r},\bm w^\mathrm{i/r},\bm z^\mathrm{i/r})$ 
of $\varphi(\bm{x},\bm{y},\bm{w},\bm{z})$. 
We define the formula $\psi'(\bm z^\mathrm{i/r}) := \ram \bm x^\mathrm{i/r},\bm y^\mathrm{i/r} \colon \exists \bm w^\mathrm{i/r} \colon 
\varphi'(\bm x^\mathrm{i/r},\bm y^\mathrm{i/r},\bm w^\mathrm{i/r},\bm z^\mathrm{i/r})$.
By definition of a decomposition, there is a valuation $\bm c^\mathrm{i/r}$ of $\bm z^\mathrm{i/r}$ with $\bm c^\mathrm{int}+\bm c^\mathrm{real} = \bm c$
that satisfies $\psi'$.
We bring $\varphi'$ into disjunctive normal form
\[\bigvee_{i=1}^n \alpha_i(\bm x^\mathrm{int},\bm y^\mathrm{int},\bm w^\mathrm{int},\bm z^\mathrm{int})
\wedge \beta_i(\bm x^\mathrm{real},\bm y^\mathrm{real},\bm w^\mathrm{real},\bm z^\mathrm{real})\]
where $\alpha_i$ is an existential Presburger formula and $\beta_i$ is an existential formula in LRA.
By Ramsey's theorem there exists $1 \le i \le n$ such that
\[\ram \bm x^\mathrm{i/r}, \bm y^\mathrm{i/r} \colon 
\exists \bm w^\mathrm{int} \colon \alpha_i(\bm x^\mathrm{int},\bm y^\mathrm{int},\bm w^\mathrm{int},\bm c^\mathrm{int})
\wedge \exists \bm w^\mathrm{real} \colon \beta_i(\bm x^\mathrm{real},\bm y^\mathrm{real},\bm w^\mathrm{real},\bm c^\mathrm{real}).\]
Note that the existentially quantified variables can be split at the conjunction into the real and integer part.
To perform a similar splitting for the variables bound by the Ramsey quantifier, we have to distinct the two cases
whether the vectors of the clique are pairwise distinct in the real components or in the integer components.
We only show the case where the vectors of the clique are pairwise distinct in both the real and integer components.
The other cases are similar by allowing that either the integer or real components do not change throughout the clique,
i.e., either $\exists \bm x^\mathrm{int}, \bm w^\mathrm{int} \colon \alpha_i(\bm x^\mathrm{int},\bm x^\mathrm{int},\bm w^\mathrm{int},\bm c^\mathrm{int})$ or
$\exists \bm x^\mathrm{real}, \bm w^\mathrm{real} \colon \beta_i(\bm x^\mathrm{real},\bm x^\mathrm{real},\bm w^\mathrm{real},\bm c^\mathrm{real})$ holds.
So we assume that
\[\ram \bm x^\mathrm{int}, \bm y^\mathrm{int} \colon \exists \bm w^\mathrm{int} \colon 
\alpha_i(\bm x^\mathrm{int},\bm y^\mathrm{int},\bm w^\mathrm{int},\bm c^\mathrm{int}) \wedge
\ram \bm x^\mathrm{real}, \bm y^\mathrm{real} \colon \exists \bm w^\mathrm{real} \colon 
\beta_i(\bm x^\mathrm{real},\bm y^\mathrm{real},\bm w^\mathrm{real},\bm c^\mathrm{real}).\]
By applying \Cref{thm:qe-presburger} to the first conjunct and \Cref{thm:qe-real} to the second conjunct, we get
\begin{align*}
&\ram (\bm{x}^\mathrm{int},\bm{v}_1^\mathrm{int},\bm{v}_2^\mathrm{int}),(\bm{y}^\mathrm{int},\bm{w}_1^\mathrm{int},\bm{w}_2^\mathrm{int}) \colon
\alpha_i(\bm x^\mathrm{int},\bm y^\mathrm{int},\bm{v}_1^\mathrm{int}+\bm{w}_2^\mathrm{int},\bm c^\mathrm{int}) 
\wedge \bm x^\mathrm{int} \ne \bm y^\mathrm{int} \ \wedge\\
&\ram (\bm{x}^\mathrm{real},\bm{v}_1^\mathrm{real},\bm{v}_2^\mathrm{real}),(\bm{y}^\mathrm{real},\bm{w}_1^\mathrm{real},\bm{w}_2^\mathrm{real})
\colon \beta_i(\bm x^\mathrm{real},\bm y^\mathrm{real},\bm v_1^\mathrm{real}+\bm w_2^\mathrm{real},\bm c^\mathrm{real}) 
\wedge \bm x^\mathrm{real} \ne \bm y^\mathrm{real}.
\end{align*}
This implies that $\bm c$ satisfies the right-hand formula of the theorem by adding the two cliques componentwise.
 
\section{Ramsey quantifiers in Presburger arithmetic}\label{sec:presburger}
In this section, we describe our procedure to eliminate the Ramsey quantifier if applied to an existential Presburger formula.

It is not difficult to construct Presburger-definable relations that have
infinite cliques, but none that are definable in Presburger arithmetic.  For
example, consider the relation $R=\{(x,y)\in\N\times\N \mid y\ge 2x\}$. Then
every infinite clique $A = \{a_0,a_1,\ldots,\}$ with $a_0\le a_1\le a_2\le\cdots$
must satisfy $a_i\ge 2^{i}\cdot a_0$ for every $i\ge 1$ and thus cannot be
ultimately periodic (i.e., there is no $n,k \in \N$ such that for all $a \ge n$, we have $a \in A$ if and only if $a+k \in A$).
Since a subset of $\N$ is Presburger-definable if and only if it is ultimately periodic, 
it follows that $A$ is not Presburger-definable.
Nevertheless, we show the following: 
\begin{theorem}\label{main-presburger}
Given an existential Presburger formula $\varphi(\bm x,\bm y,\bm z)$, 
we can construct in polynomial time an existential
Presburger formula of linear size that is equivalent to
$\ram \bm{x},\bm{y}\colon \varphi(\bm{x},\bm{y},\bm{z})$.
\end{theorem}

We first assume that $\varphi$ is a conjunction of the form
\begin{equation} \bigwedge_{i=1}^n \bm{r}_i^\top \bm{x} < \bm{s}_i^\top\bm{y} + \bm{t}_i^\top\bm{z}+h_i ~~\wedge~~ \bigwedge_{j=1}^m \bm{u}_j^\top \bm{x}\approx_{e_j}^j \bm{v}_j^\top\bm{y}+\bm{w}_j^\top\bm{z}+d_j \label{presburger-variable-free}\end{equation}
where $\approx_{e_j}^j \in \{\equiv_{e_j},\not\equiv_{e_j}\}$. It should be noted that since \cref{thm:qe-mixed} allows us to eliminate any existential quantifier under the Ramsey quantifier without introducing modulo constraints, it would even suffice to treat the case where $\varphi$ has no modulo constraints. However, in practice it might be useful to be able to treat modulo constraints without first trading them in for existential quantifiers. For this reason, we describe the translation in the presence of modulo constraints.

\subsection*{Cliques in terms of profiles}
Our goal is to construct an existential Presburger formula $\varphi'(\bm{z})$
so that $\varphi'(\bm{c})$ holds if and only if there exists an infinite
sequence $\bm{a}_1,\bm{a}_2,\ldots$ of pairwise distinct vectors for which
$\varphi(\bm{a}_i,\bm{a}_j,\bm{c})$ for every $i<j$. As mentioned above, it is
possible that such a sequence exists, but that none of them is definable in
Presburger arithmetic. Therefore, our first step is to modify the condition
``$\varphi(\bm{a}_i,\bm{a}_j,\bm{c})$ for $i<j$'' into a different condition
such that (i)~the new condition is equivalent in terms of existence of a
sequence and (ii)~the new condition can always be satisfied by an arithmetic
progression. 

To illustrate the idea, suppose $\varphi(\bm{x},\bm{y},\bm{z})$ says that
$y_1>2\cdot x_1 \wedge \psi(\bm{x})$ for some Presburger formula $\psi$.
As mentioned above, any directed clique for $\varphi$ must grow exponentially
in the first component.  However, such a directed clique exists if and only if
there exists a sequence $\bm{a}_1,\bm{a}_2,\ldots$ such that
$\psi(\bm{a}_1),\psi(\bm{a}_2),\ldots$ and the sequence of numbers
$a_1,a_2,\ldots$ in the first components of $\bm a_1,\bm a_2,\dots$ grows unboundedly: Clearly, any directed
clique for $\varphi$ must satisfy this. Conversely, a sequence satisfying the
unboundedness condition must have a subsequence with $a_j>2\cdot a_i$ for $i<j$.

These modified conditions on sequences are based on the notion of profiles.
Essentially, a profile captures how in a sequence $\bm{a}_1,\bm{a}_2,\ldots$ the values $\bm{r}_i^\top\bm{a}_k$ and $\bm{s}_i^\top\bm{a}_k+\bm{t}_i^\top\bm{c}+h_i$ evolve.
A \emph{profile (for $\varphi$)} is a vector in $\Z_\omega^{2n}$ where $\Z_\omega := \Z \cup \{\omega\}$. 
Suppose $\bm{p}=(p_1,\ldots,p_{2n})$. Then value $p_{2i-1}$ being an integer means that $\bm{r}_i^\top\bm{a}_1,\bm{r}_i^\top\bm{a}_2,\ldots$ is bounded from above by $p_{2i-1}$. If $p_{2i-1}$ is $\omega$, then the sequence $\bm{r}_i^\top\bm{a}_1,\bm{r}_i^\top\bm{a}_2,\ldots$ tends to infinity. Similarly, even-indexed entries $p_{2i}$ describe the evolution of the sequence $\bm{s}_i^\top\bm{a}_1+\bm{t}_i^\top\bm{c}+h_i,\bm{s}_i^\top\bm{a}_2+\bm{t}_i^\top\bm{c}+h_i,\ldots$.

Let us make this precise. If $\bm{p}$ is a profile and $\bm{c}$ is a vector over $\Z$, then a sequence
$\bm{a}_1,\bm{a}_2,\ldots$ of pairwise distinct vectors over $\Z$ is \emph{compatible with $\bm{p}$ for $\bm{c}$} if for every $k<\ell$, we have $\bm{u}_j^\top\bm{a}_k\approx_{e_j}^j \bm{v}_j^\top\bm{a}_\ell+\bm{w}_j^\top\bm{c}+d_j$ and 
\begin{align} \sup\{\bm{r}_i^\top\bm{a}_k \mid k=1,2,\ldots \}\le p_{2i-1}, && p_{2i}\le\liminf \{\bm{s}_i^\top\bm{a}_k+\bm{t}_i^\top\bm{c}+h_i\mid k=1,2,\ldots\} \label{presburger-compatibility}.\end{align}
A profile
$\bm{p}=(p_1,\ldots,p_{2n})$ is \emph{admissible} if for every $i\in[1,n]$, we
have $p_{2i-1}<p_{2i}$ or $p_{2i}=\omega$. 
\begin{lemma}\label{lem:admissible-compatible}
Let $\bm{c}$ be a vector over $\Z$. Then $\ram
\bm{x},\bm{y}\colon\varphi(\bm{x},\bm{y},\bm{c})$ if and only if there exists
an admissible profile $\bm{p}\in\Z^{2n}_\omega$ such that there is a sequence
compatible with $\bm{p}$ for $\bm{c}$.
\end{lemma}
\begin{proof}
We begin with the ``only if'' direction. 
To ease notation, we write $f_i(\bm{x})=\bm{r}_i^\top\bm{x}$
and $g_i(\bm{x})=\bm{s}_i^\top\bm{x}+\bm{t}_i^\top\bm{c}+h_i$ for $i\in[1,n]$.
Suppose $\bm{a}_1,\bm{a}_2,\ldots$ is
a directed clique witnessing $\ram
\bm{x},\bm{y}\colon\varphi(\bm{x},\bm{y},\bm{c})$. First, we may assume that if
for some $i\in[1,n]$, the sequence $\{f_i(\bm{a}_k) \mid
k=1,2,\ldots\}$ is bounded from above, then for its maximum $M$, we have
$M<g_i(\bm{a}_k)$ for every $k\ge 1$. If this is not
the case, we can achieve it by removing an initial segment of
$\bm{a}_1,\bm{a}_2,\ldots$.
Now we define the profile $\bm{p}=(p_1,p_2,\ldots,p_{2n-1},p_{2n})$ as
\begin{align*} p_{2i-1} = \sup\{f_i(\bm{a}_k) \mid k=1,2,\ldots \}, && p_{2i}=\liminf \{g_i(\bm{a}_k)\mid k=1,2,\ldots\}. \end{align*}
Observe that $p_{2i}$ cannot be $-\omega$ and thus  belongs to $\Z_{\omega}$:
This is because the set
$\{g_i(\bm{a}_k)\mid k\ge 1\}$ is
bounded from below (by
$\min\{f_i(\bm{a}_1),g_i(\bm{a}_1)\}$).
Then $\bm{p}$ is admissible: Otherwise, we would have $p_{2i-1}\ge p_{2i}$ and $p_{2i}\in\Z$, implying that
there are $k<\ell$ with $f_i(\bm{a}_k)\ge
p_{2i}=g_i(\bm{a}_\ell)$, which contradicts the
fact that $\bm{a}_1,\bm{a}_2,\ldots$ witnesses $\ram
\bm{x},\bm{y}\colon\varphi(\bm{x},\bm{y},\bm{c})$. Moreover, by definition of
$\bm{p}$, the sequence $\bm{a}_1,\bm{a}_2,\ldots$ is clearly compatible with
$\bm{p}$ for $\bm{c}$.

Let us now prove the ``if'' direction. Let $\bm{p}\in\Z_\omega^{2n}$ be an
admissible profile and $\bm{a}_1,\bm{a}_2,\ldots$ be a sequence compatible with
$\bm{p}$ for $\bm{c}$. Then, we know that for any $k<\ell$, we have
$\bm{u}_j^\top\bm{a}_k\approx_{e_j}^j\bm{v}_j^\top\bm{a}_\ell+\bm{w}_j\bm{c}+d_j$. 
We claim that we can select a subsequence of
$\bm{a}_1,\bm{a}_2,\ldots$ such that, for every $i\in[1,n]$, we have $f_i(\bm{a}_k)<g_i(\bm{a}_\ell)$ for every $k<\ell$.
It suffices to do this for each $i=1,\ldots,n$ individually, because if for
some $i\in[1,n]$, we have $f_i(\bm{a}_k)<g_i(\bm{a}_\ell)$ for every $k<\ell$,
then this is still the case for any subsequence. Likewise, picking an infinite
subsequence does not spoil the property of being compatible with $\bm{p}$ for
$\bm{c}$.

Consider some $i\in[1,n]$. We distinguish two cases, namely whether
$p_{2i}\in\Z$ or $p_{2i}=\omega$. First, suppose $p_{2i}\in\Z$. Then, since
$\bm{p}$ is admissible, we have $p_{2i-1}<p_{2i}$. Now compatibility implies
that $p_{2i-1}<p_{2i}\le g_i(\bm{a}_\ell)$ for almost all $\ell$. Hence, by
removing some initial segment of our sequence, we can ensure that
$f_i(\bm{a}_k)<g_i(\bm{a}_\ell)$ for every $k<\ell$.

Now suppose $p_{2i}=\omega$. We successively choose the elements of a subsequence
$\bm{a}'_1,\bm{a}'_2,\ldots$ of $\bm{a}_1,\bm{a}_2,\ldots$ such that
$f_i(\bm{a}'_k)<g_i(\bm{a}'_\ell)$ for any $k<\ell$. Suppose we have already
chosen $\bm{a}'_1,\ldots,\bm{a}'_h$ for some $h\ge 1$.  Then the set
$\{f_i(\bm{a}'_k)\mid k\in[1,h]\}$ is finite and thus bounded by some $M\in\N$.
By compatibility of $\bm{a}_1,\bm{a}_2,\ldots$, there exist infinitely many
$\ell\in\N$ with $M<g_i(\bm{a}_\ell)$. This allows us to choose $\bm{a}'_{h+1}$
to extend our sequence. 

This completes the construction of our clique witnessing $\ram
\bm{x},\bm{y}\colon \varphi(\bm{x},\bm{y},\bm{c})$.
\end{proof}

\subsection*{Compatibility in terms of matrices}
Our next step is to express the existence of a sequence compatible with
$\bm{p}$ for $\bm{c}$ in terms of certain inequalities. To this end, we
define two matrices $\bm{A}_{\bm{p},\bm{c}}$ and $\bm{B}_{\bm{p}}$ and a vector
$\bm{b}_{\bm{p},\bm{c}}$. Here, $\bm{A}_{\bm{p},\bm{c}}\bm{x}\ge \bm{b}_{\bm{p},\bm{c}}$ will
express the compatibility conditions involving $p_{2i-1}$ and $p_{2i}$ that
are integers. 
Thus, we define $\bm{A}_{\bm{p},\bm{c}}$ and $\bm{b}_{\bm{p},\bm{c}}$ by
describing the system of inequality
$\bm{A}_{\bm{p},\bm{c}}\bm{x}\ge\bm{b}_{\bm{p},\bm{c}}$.  For every $i\in[1,n]$
with $p_{2i-1}\in\Z$, we add the inequality $\bm{r}_i^\top\bm{x}\le p_{2i-1}$.
Moreover, for every $i\in[1,n]$ with $p_{2i}\in\Z$, we add the inequality
$p_{2i}\le \bm{s}_i^\top\bm{x}+\bm{t}_i^\top\bm{c}+h_i$.

Moreover, $\bm{B}_{\bm{p}}$ will be used to express the unboundedness condition
on the right side of \cref{presburger-compatibility} if $p_{2i}=\omega$. Thus,
for every $i\in[1,n]$ with $p_{2i}=\omega$, we add the row $\bm{s}_i^\top$ to
$\bm{B}_{\bm{p}}$. We say that a function $f\colon X\to\Z^\ell$ is
\emph{simultaneously unbounded} on a sequence $x_1,x_2,\ldots\in X$ if for
every $k\in\N$, we have $f(x_j)\ge (k,\ldots,k)$ for almost all $j$. Now
observe the following:
\begin{lemma}\label{presburger-compatibility-matrices}
	Let $\bm{c}$ be a vector and $\bm{p}\in\Z^{2n}_\omega$ be a profile.
	Then there exists a sequence that is compatible with $\bm{p}$
	for $\bm{c}$ if and only if there exists a sequence $\bm{a}_1,\bm{a}_2,\ldots$ of pairwise distinct vectors such that
	(i)~$\bm{u}_j^\top\bm{a}_k\approx_{e_j}^j\bm{v}_j^\top\bm{a}_\ell+\bm{w}_j^\top\bm{c}+d_j$
	for every $j\in[1,m]$ and $k < \ell$ and (ii)
	~$\bm{A}_{\bm{p},\bm{c}}\bm{a}_k\ge\bm{b}_{\bm{p},\bm{c}}$ for every
	$k\ge 1$ and (iii)~$\bm{B}_{\bm{p}}$ is simultaneously unbounded on
	$\bm{a}_1,\bm{a}_2,\ldots$.
\end{lemma}

\subsection*{Arithmetic progressions}
The last key step is to show that there exists a sequence compatible with
$\bm{p}$ if and only if there exists such a sequence of the form
$\bm{a}_0+\bm{a},\bm{a}_0+2\bm{a},\bm{a}_0+3\bm{a},\ldots$. This will allow us
to express existence of a sequence by the existence of suitable vectors
$\bm{a}_0$ and $\bm{a}$.
\begin{lemma}\label{lem:compatible}
	Let $\bm{c}$ be a vector and $\bm{p}\in\Z^{2n}_\omega$ be a profile.
	There exists a sequence compatible with $\bm{p}$ for $\bm{c}$ if and
	only if there are vectors $\bm{a}_0,\bm{a}$ over $\Z$ with $\bm{a} \ne \bm{0}$ such that for all $j \in [1,m]$,
	\begin{align*}
		\bm{A}_{\bm{p},\bm{c}}\bm{a}_0\ge\bm{b}_{\bm{p},\bm{c}},~\bm{A}_{\bm{p},\bm{c}}\bm{a}\ge \bm{0},  && \bm{B}_{\bm{p}}\bm{a}\gg\bm{0}, \\
	\bm{u}_j^\top\bm{a}_0\approx_{e_j}^j\bm{v}_j^\top(\bm{a}_0+\bm{a})+\bm{w}_j^\top\bm{c}+d_j, && \bm{u}_j^\top\bm{a}\equiv_{e_j} \bm{v}_j^\top\bm{a}\equiv_{e_j} 0. 
\end{align*}
\end{lemma}
\begin{proof}
	We begin with the ``if'' direction. Suppose there are vectors
	$\bm{a}_0$ and $\bm{a}$ as described. Then we claim that the sequence
	$\bm{a}_1,\bm{a}_2,\ldots$ with $\bm{a}_k=\bm{a}_0+k\cdot \bm{a}$ is
	compatible with $\bm{p}$ for $\bm{c}$. We use
	\cref{presburger-compatibility-matrices} to show this. 
	First note that since $\bm{a} \ne \bm{0}$, the $\bm{a}_k$ are pairwise distinct.
	It is clear that
	the sequence satisfies conditions (i) and (ii) of
	\cref{presburger-compatibility-matrices}. Condition (iii) holds as
	well, because in the vector $\bm{B}_{\bm{p}}(k\cdot\bm{a})$, every entry
	is at least $k$. Thus, $\bm{B}_{\bm{p}}$ is simultaneously unbounded on
	$\bm{a}_1,\bm{a}_2,\ldots$.
	
	For the ``only if'' direction, suppose $\bm{a}_1,\bm{a}_2,\ldots$ is a
	sequence of pairwise distinct vectors that satisfies the conditions in
	\cref{presburger-compatibility-matrices}. Since there are only finitely many
	possible remainders modulo $e_j$ of the expressions $\bm{u}_j^\top\bm{a}_k$ and
	$\bm{v}_j^\top\bm{a}_k$, we can pick a subsequence such that for each
	$j\in[1,m]$, the maps $k\mapsto\bm{u}_j^\top \bm{a}_k$ and $k\mapsto
	\bm{v}_j^\top\bm{a}_k$ are constant modulo $e_j$.  In the second step, we
	notice that since $\bm{A_{p,c}}\bm{a}_k\ge\bm{b_{p,c}}$ for each $k\ge 1$, the
	sequence $\bm{A_{p,c}}\bm{a}_1,\bm{A_{p,c}}\bm{a}_2,\ldots$ cannot contain an infinite
	strictly descending chain in any component. Thus, by Ramsey's theorem, we may pick a subsequence
	so that $\bm{A_{p,c}}\bm{a}_1\le\bm{A_{p,c}}\bm{a}_2\le\cdots$. Note that
	passing to subsequences does not spoil the conditions of
	\cref{presburger-compatibility-matrices}.  Thus, $\bm{B}_{\bm{p}}$ is still
	simultaneously unbounded on $\bm{a}_1,\bm{a}_2,\ldots$. This allows us to pick
	a subsequence so that also
	$\bm{B}_{\bm{p}}\bm{a}_1\ll\bm{B}_{\bm{p}}\bm{a}_2\ll\cdots$.  Therefore, if we
	set $\bm{a}_0:=\bm{a}_1$ and $\bm{a}:=\bm{a}_2-\bm{a}_1$, then $\bm{a}_0$ and
	$\bm{a}$ are as desired.
\end{proof}

\subsection*{Construction of the formula}
We are now ready to prove \Cref{main-presburger} in the general case, i.e., $\varphi(\bm{x},\bm{y},\bm{z})$ is an arbitrary existential Presburger formula.
By \Cref{thm:qe-mixed} we can assume that $\varphi$ is quantifier-free.
By moving all negations inwards to the atoms and possibly negating those, we may further assume that $\varphi$ is a positive Boolean combination of inequality atoms 
$\alpha_i := \bm{r}_i^\top \bm{x} < \bm{s}_i^\top\bm{y} + \bm{t}_i^\top\bm{z}+h_i$ for $i \in [1,n]$ and
modulo constraint atoms $\beta_j := \bm{u}_j^\top \bm{x}\approx_{e_j}^j \bm{v}_j^\top\bm{y}+\bm{w}_j^\top\bm{z}+d_j$ 
with $\approx_{e_j}^j \in \{\equiv_{e_j}, \not\equiv_{e_j}\}$ for $j \in [1,m]$.
(Note that in Presburger arithmetic equality can be expressed by a conjunction of two strict inequalities.)

The key idea is to guess (using existentially quantified variables) a
subset of the atoms in $\varphi$ and check that 
(i)~satisfying those atoms makes $\varphi$ true and 
(ii)~for the conjunction of those atoms, there exists a directed
clique. 
Note that there are only finitely many conjunctions of atoms (from $\varphi$)
and $\varphi$ is equivalent to the disjunction over these 
conjunctions. Thus, by Ramsey's theorem, there exists a directed clique
for $\varphi$ if and only if there exists one for some conjunction of atoms.
Condition~(i) is easy to state.
To check~(ii), we then require the conditions of \cref{lem:compatible} to be satisfied for our conjunction of atoms.

For each atom $\alpha_i$, we introduce a variable $q^<_i$, and for each atom
$\beta_j$, we introduce a variable $q^\approx_j$. To check that (i) holds, we
use the formula $\varphi'$, that is obtained from $\varphi$ by replacing each
$\alpha_i$ with $q_i^< = 1$ and each $\beta_j$ with $q_j^\approx = 1$, and
adding the restrictions $q_i^< = 0 \vee q_i^< = 1$ and $q_j^\approx = 0 \vee
q_j^\approx = 1$.  Now, $\varphi$ is equivalent to
\[\psi := \exists \bm{q}^<, \bm{q}^\approx \colon \varphi' \wedge \bigwedge_{i=1}^n (q_i^< = 1 \to \alpha_i) \wedge 
\bigwedge_{j=1}^m (q_j^\approx = 1 \to \beta_j).\]

Let us now construct the formula for condition~(ii) above. To this end, we build a formula $\gamma_i$ that states all conditions of \cref{lem:compatible} that stem from the atom $\alpha_i$.
For $i \in [1,n]$ and fresh variables $p_{2i-1}, p_{2i}, \bm{x}_0, \bm{x}$ let
\begin{align*}
\gamma_i :=\ & (p_{2i-1} < \omega \to (\bm{r}_i^\top \bm{x}_0 \le p_{2i-1} \wedge \bm{r}_i^\top \bm{x} \le 0)) \ \wedge \\
& (p_{2i} < \omega \to (p_{2i} \le \bm{s}_i^\top \bm{x}_0 + \bm{t}_i^\top \bm{z} + h_i \wedge \bm{s}_i^\top \bm{x} \ge 0)) \ \wedge \\
& (p_{2i} = \omega \to \bm{s}_i^\top \bm{x} > 0)
\end{align*}
and for all $j \in [1,m]$ let
\begin{align*}
\delta_j := \bm{u}_j^\top\bm{x}_0 \approx_{e_j}^j \bm{v}_j^\top(\bm{x}_0+\bm{x})+\bm{w}_j^\top\bm{z}+d_j \wedge
\bm{u}_j^\top\bm{x} \equiv_{e_j} 0 \wedge \bm{v}_j^\top\bm{x} \equiv_{e_j} 0.
\end{align*}
Here, $p_\ell < \omega$ and $p_\ell = \omega$ is shorthand notation for $\omega_\ell = 0$ and $\omega_\ell = 1$, respectively,
where $\omega_\ell$ is a fresh variable associated with $p_\ell$ that is restricted to values from $\{0,1\}$.
Thus, from now on we implicitly quantify $\bm{\omega}$ when $\bm{p}$ is quantified.
The following requires $\bm{p}$ to be an admissible profile:
\[\theta := \bigwedge_{i=1}^n (p_{2i-1} < \omega \wedge p_{2i-1} < p_{2i} \vee p_{2i} = \omega).\]
Then we claim that $\ram \bm{x},\bm{y} \colon \psi(\bm{x},\bm{y},\bm{z})$ is equivalent to
\[
\chi := \exists \bm{q}^<, \bm{q}^\approx, \bm{p}, \bm{x}_0, \bm{x} \colon \varphi' \wedge \theta \wedge \bm{x} \ne \bm{0} \wedge 
\bigwedge_{i=1}^n (q_i^< = 1 \to \gamma_i) \wedge \bigwedge_{j=1}^m (q_j^\approx = 1 \to \delta_j).
\]
We show that for any valuation $\bm{c} \in \Z^{|\bm{z}|}$ of $\bm{z}$ we have
$\ram \bm{x},\bm{y} \colon \psi(\bm{x},\bm{y},\bm{c})$ if and only if $\chi(\bm{c})$.
For an assignment $\nu$ of the $q_i^<,q_j^\approx$ to $\{0,1\}$ let 
$I_\nu := \{i \in [1,n] \mid \nu(q_i^<) = 1\}$ and $J_\nu := \{j \in [1,m] \mid \nu(q_j^\approx) = 1\}$.
By Ramsey's theorem we have $\ram \bm{x},\bm{y} \colon \psi(\bm{x},\bm{y},\bm{c})$ if and only if 
there is an assignment $\nu$ of the $q_i^<,q_j^\approx$ satisfying $\varphi'$ such that
$\ram \bm{x},\bm{y} \colon \bigwedge_{i \in I_\nu} \alpha_i(\bm{x},\bm{y},\bm{c}) \wedge \bigwedge_{j \in J_\nu} \beta_j(\bm{x},\bm{y},\bm{c})$.
By \Cref{lem:admissible-compatible,lem:compatible} this formula is equivalent to
$\exists \bm{p},\bm{x}_0,\bm{x} \colon \theta \wedge \bm{x} \ne \bm{0} \wedge \bigwedge_{i \in I_\nu} \gamma_i(\bm{p},\bm{x}_0,\bm{x},\bm{c}) \wedge 
\bigwedge_{j \in J_\nu} \delta_j(\bm{x}_0,\bm{x},\bm{c})$
which in turn holds for some assignment $\nu$ of the $q_i^<,q_j^\approx$ satisfying $\varphi'$ if and only if $\chi(\bm{c})$.
 
\section{Ramsey quantifiers in Linear Real Arithmetic}\label{sec:reals}
In this section, we describe our procedure to eliminate the Ramsey quantifier if applied to an existential LRA formula. At the end of the section, we mention a version of this result for the structure $\QLin$ (\cref{thm:rational-main}).
\begin{theorem}\label{thm:real-main}
Given an existential formula $\varphi(\bm x, \bm y, \bm z)$ in LRA,
we can construct in polynomial time an existential formula in LRA of linear size that is equivalent to $\ram \bm x, \bm y \colon \varphi(\bm x, \bm y, \bm z)$.
\end{theorem}
Similar to the integer case, we first assume that $\varphi$ is a conjunction of the following form:
\begin{equation} \bigwedge_{i=1}^n \bm{r}_i^\top \bm{x} < \bm{s}_i^\top\bm{y} + \bm{t}_i^\top\bm{z}+h_i ~~\wedge~~ 
	\bigwedge_{j=1}^m \bm{u}_j^\top \bm{x} = \bm{v}_j^\top\bm{y}+\bm{w}_j^\top\bm{z}+d_j\label{reals-conjunction}\end{equation}
where $\bm r_i, \bm s_i, \bm t_i, \bm u_j, \bm v_j, \bm w_j \in \Q^d$ for $d \ge 1$ and $h_i,d_j \in \Q$.

\subsection*{Cliques in terms of profiles}
We now define the notion of a profile with a similar purpose as in the integer
case.  In the real case, these carry more information: In the case of
Presburger arithmetic, it is enough to guess whether a particular function
grows or has a particular upper bound.  Here, it is possible that a function
grows strictly, but it still bounded, because it converges.  For example, if
$\varphi(x,y)$ says that $x<y$ and $x\le 1$, then a clique must be a strictly
ascending sequence of numbers $\le 1$.

A \emph{profile} (for $\varphi$) is a tuple $\bm p = (\rho,\sigma,t_\rho,t_\sigma)$ of functions where
$\rho,\sigma \colon \{1,\dots,n\} \to \R \cup \{-\omega,\omega\}$ and
$t_\rho,t_\sigma \colon \{1,\dots,n\} \to \{-\omega,-1,0,1,\omega\}$.
For a sequence $(\bm a_k)_{k\ge1}$ in $\R^d$ let $\bm\rho_i := (\bm r_i^\top \bm a_k)_{k\ge1}$ and $\bm\sigma_i := (\bm s_i^\top \bm a_k)_{k\ge1}$.
We say that a sequence $(\bm a_k)_{k\ge1}$ of pairwise distinct vectors is \emph{compatible} with $\bm p$ if
$\rho(i)$ and $\sigma(i)$ are the real values to which the sequences $\bm\rho_i$ and $\bm\sigma_i$ converge
or $\omega$ (resp. $-\omega$) if the corresponding sequence is strictly increasing (resp. decreasing) and diverges to $\infty$ (resp. $-\infty$) and
the functions $t_\rho$ and $t_\sigma$ describe the type of convergence where
type $0$ means that the corresponding sequence is constant,
type $1$ (resp. $-1$) means that it is strictly increasing (resp. decreasing) and converges from below (resp. above), and
the type is $\omega$ (resp. $-\omega$) in the divergent case.
A profile $\bm p$ is $\bm c$\emph{-admissible} for a vector $\bm c \in \R^d$ if for all $i \in \{1,\dots,n\}$ we have
\begin{itemize}
\item $\sigma(i) \ne -\omega$ and if $\rho(i) = \omega$, then $\sigma(i) = \omega$,
\item $\rho(i) < \sigma(i) + \bm t_i^\top \bm c + h_i$ 
if either $t_\rho(i) \in \{-1,0\}$ and $t_\sigma(i) \in \{0,1\}$ or $t_\rho(i) = -1$ and $t_\sigma(i) = -1$,
\item $\rho(i) \le \sigma(i) + \bm t_i^\top \bm c + h_i$
if either $t_\rho(i) = 0$ and $t_\sigma(i) = -1$ or $t_\rho(i) = 1$.
\end{itemize}
We say that a sequence $(\bm a_k)_{k\ge1}$ satisfies the equality constraints (of $\varphi$) for $\bm c \in \R^d$ if
$\bm u_j^\top \bm a_k = \bm{v}_j^\top\bm a_\ell+\bm{w}_j^\top\bm{c}+d_j$
for all $j \in \{1,\dots,m\}$ and $k < \ell$.

\begin{lemma}\label{lem:ramsey-compatible}
Let $\bm c \in \R^d$. Then $\ram \bm x,\bm y \colon \varphi(\bm x,\bm y, \bm c)$ if and only if
there exists a $\bm c$-admissible profile $\bm p$ such that there is a sequence compatible with $\bm p$ that satisfies the equality constraints for $\bm c$.
\end{lemma}
\begin{proof}
	We first show the ``only if'' direction.
	Let $(\bm a_k)_{k\ge1}$ be a clique witnessing $\ram \bm x,\bm y \colon \varphi(\bm x,\bm y, \bm c)$.
	For all $i \in \{1,\dots,n\}$ consider the sequence $\bm \rho_i$.
	By the Bolzano-Weierstrass theorem if $\bm \rho_i$ is bounded, we can replace $(\bm a_k)_{k\ge1}$ by an infinite subsequence such that
	$\bm \rho_i$ converges against a real value $r_i \in \R$.
	By restricting further to an infinite subsequence, we have that $\bm \rho_i$ is either constant, strictly increasing, or strictly decreasing.
	Thus, we set $\rho(i) := r_i$ and $t_\rho(i)$ to $0$, $1$, or $-1$ depending on whether $\bm \rho_i$ is constant, increasing, or decreasing.
	If $\bm \rho_i$ is unbounded, we replace $(\bm a_k)_{k\ge1}$ by an infinite subsequence such that
	$\bm \rho_i$ is strictly increasing if it is unbounded above and strictly decreasing if it is unbounded below.
	Then we set $\rho(i)$ and $t_\rho(i)$ to $\omega$ or $-\omega$ depending on whether $\bm \rho_i$ is increasing or decreasing.
	Similarly, we can define $\sigma(i)$ and $t_\sigma(i)$ by considering the sequence $\bm \sigma_i$.
	Thus, there is a sequence $(\bm a_k)_{k\ge1}$ that is compatible with the profile $\bm p := (\rho,\sigma,t_\rho,t_\sigma)$.
	Since $(\bm a_k)_{k\ge1}$ still satisfies the equality constraints for $\bm c$, it remains to show that $\bm p$ is $\bm c$-admissible.
	First observe that $\sigma(i) \ne -\omega$ since $\bm \sigma_i$ is bounded from below by $\min\{\bm r_i^\top \bm a_1, \bm s_i^\top \bm a_1\}$.
	If $\rho(i) = \omega$, then also $\sigma(i) = \omega$ since otherwise 
	there were $k < \ell$ such that $\bm r_i^\top \bm a_k \ge \bm s_i^\top \bm a_\ell +\bm t_i^\top \bm c+h_i$.
	With a similar reasoning we can show that
	if $t_\rho(i) \in \{-1,0\}$ and $t_\sigma(i) \in \{0,1\}$, then $\rho(i) < \sigma(i)+\bm t_i^\top \bm c+h_i$,
	if $t_\rho(i) = t_\sigma(i) = -1$, then $\rho(i) < \sigma(i)+\bm t_i^\top \bm c+h_i$, and
	if either $t_\rho(i) = 0$ and $t_\sigma(i) = -1$ or $t_\rho(i) = 1$, then $\rho(i) \le \sigma(i)+\bm t_i^\top \bm c+h_i$.
	
	We now turn to the ``if'' direction.
	Let $\bm p = (\rho,\sigma,t_\rho,t_\sigma)$ be a $\bm c$-admissible profile 
	and $(\bm a_k)_{k\ge1}$ be a sequence compatible with $\bm p$ that satisfies the equality constraints for $\bm c$.
	We successively restrict for each $i \in \{1,\dots,n\}$ to a subsequence such that 
	$\bm r_i^\top \bm a_k < \bm s_i^\top \bm a_\ell +\bm t_i^\top \bm c+h_i$ for all $k < \ell$.
	We construct the subsequence inductively.
	Suppose we already constructed the subsequence $\bm a_{i_1},\dots,\bm a_{i_h}$
	such that $\bm r_i^\top \bm a_{i_k} < \bm s_i^\top \bm a_{i_\ell} +\bm t_i^\top \bm c+h_i$ for all $1 \le k < \ell \le h$ and the set 
	$L_h := \{\ell > i_h \mid \forall 1 \le k \le h \colon \bm r_i^\top \bm a_{i_k} < \bm s_i^\top \bm a_\ell +\bm t_i^\top \bm c+h_i\}$ 
	is infinite.
	Let $i_{h+1} := \min(L_h)$.
	If $\rho(i) < \sigma(i)+\bm t_i^\top \bm c+h_i$, then clearly there is an infinite subset $L$ of $L_h$ such that
	$\bm r_i^\top \bm a_{i_{h+1}} < \bm s_i^\top \bm a_\ell +\bm t_i^\top \bm c+h_i$ for all $\ell \in L \setminus \{i_{h+1}\}$.
	If $\rho(i) = \sigma(i)+\bm t_i^\top \bm c+h_i$, then by definition of $\bm c$-admissibility we have that
	either $t_\rho(i) = 0$ and $t_\sigma(i) = -1$ or $t_\sigma(i) = 1$, or $\rho(i) = \sigma(i) = \omega$.
	In all of these cases we can find an infinite subset $L$ of $L_h$ such that
	$\bm r_i^\top \bm a_{i_{h+1}} < \bm s_i^\top \bm a_\ell +\bm t_i^\top \bm c+h_i$ for all $\ell \in L \setminus \{i_{h+1}\}$.
	Thus, we can extend the subsequence by $\bm a_{i_{h+1}}$
	where the set $L_{h+1}$ is infinite since it contains $L$.
	Finally, note that passing to subsequences does not spoil the satisfaction of the equality constraints for $\bm c$.
	Thus, the constructed subsequence is an infinite clique witnessing $\ram \bm x,\bm y \colon \varphi(\bm x,\bm y, \bm c)$.
\end{proof}

\subsection*{A general form of cliques}
In the case of Presburger arithmetic, a key insight was that if there exists a clique compatible with a profile, then there exists one of the form $\bm{a}_0,\bm{a}_0+\bm{a},\bm{a}_0+2\cdot\bm{a},\ldots$. The case of reals is more involved in this regard: There are profiles with which no arithmetic progression is compatible.

For example, consider the profile that specifies that in the first component,
the numbers must increase strictly in each step and tend to infinity. In the
second component, the numbers must also increase strictly, but are bounded from
above by $1$. A sequence compatible with this would be $(\tfrac{1}{2},1),
(\tfrac{3}{4},2), (\tfrac{4}{5},3), \ldots$. However, such a sequence cannot be
of the form $\bm{a}_0,\bm{a}_0+\bm{a},\bm{a}_0+2\cdot\bm{a},\ldots$: The entry
in the first component of $\bm{a}$ would have to be positive; but if it is, then the first component
also tends to infinity.
Instead, we look for cliques of the form $\bm{a}_1,\bm{a}_2,\ldots$ with
\begin{equation} \bm{a}_k=\bm{a}-\tfrac{1}{k}\bm{d}_c+k\bm{d}_\infty \label{general-form-reals}\end{equation}
for some vectors $\bm{a}$, $\bm{d}_c$, and $\bm{d}_\infty$. Here the vector
$\bm{d}_c$ realizes the convergence behavior: By subtracting smaller and
smaller fractions of it, the part $\bm{a}-\tfrac{1}{k}\bm{d}_c$ converges to
$\bm{a}$. Moreover, the vector $\bm{d}_\infty$ realizes divergence to
$\infty$ or $-\infty$: By adding larger and larger multiples of it, we can make
sure certain linear functions on $\bm{a}_k$ grow unboundedly.

We will later formulate necessary conditions on vectors $\bm{a}$, $\bm{d}_c$, and $\bm{d}_\infty$ such that the sequence \eqref{general-form-reals} is compatible with a profile and satisfies the equality constraints of $\varphi$. We will then show the converse in \cref{lem:compatible-vectors}: If there is a compatible sequence, then there is one of the form \eqref{general-form-reals}.

\subsection*{Extracting $\bm{a}$ and $\bm{d}_\infty$}
Before we formulate the necessary conditions, we present the key lemma that will yield the existence of $\bm{a}$ and $\bm{d}_\infty$. Suppose that we are looking for a sequence $\bm{a}_1,\bm{a}_2,\ldots$ in $\R^d$ where for some linear maps $\bm{A}\colon \R^d\to\R^m$ and $\bm{B}\colon\R^d\to\R^n$, the sequence $\bm{A}\bm{a}_1,\bm{A}\bm{a}_2,\ldots$ converges to some $\bm{v}\in\R^d$ and the sequence $\bm{B}\bm{a}_1,\bm{B}\bm{a}_2,\ldots$ is simultaneously unbounded. If we want to show that there exists such  a sequence of the form \cref{general-form-reals}, then we need $\bm{a}$ and $\bm{d}_\infty$ to satisfy (i)~$\bm{A}\bm{a}=\bm{v}$, (ii)~$\bm{A}\bm{d}_\infty=\bm{0}$ and (iii)~$\bm{B}\bm{d}_\infty\gg\bm{0}$. 
Indeed, we need $\bm{A}\bm{d}_\infty=\bm{0}$, because if $\bm{A}\bm{d}_\infty$ had a non-zero component, the sequence $k\mapsto \bm{A}(a-\tfrac{1}{k}\bm{d}_c+k\bm{d}_\infty)$ would diverge in that component. Moreover, if $\bm{A}\bm{d}_\infty=\bm{0}$, then $k\mapsto \bm{A}(\bm{a}-\tfrac{1}{k}\bm{d}_c+k\bm{d}_\infty)=\bm{A}(\bm{a}-\tfrac{1}{k}\bm{d}_c)$ converges to $\bm{A}\bm{a}$, meaning we need $\bm{A}\bm{a}=\bm{v}$. Finally, the map $\bm{B}$ is simultaneously unbounded on the sequence $k\mapsto \bm{a}-\tfrac{1}{k}\bm{d}_c+k\bm{d}_\infty$ if and only if $\bm{B}\bm{d}_\infty\gg\bm{0}$.

The following lemma yields vectors $\bm{a}$ and $\bm{d}_\infty$ that satisfy these conditions.
\begin{lemma}\label{existence-d-infty}
	Let $\bA\colon\R^d\to\R^m$ and $\bB\colon\R^d\to\R^n$ be linear maps.
	Let $\bm{a}_1,\bm{a}_2,\ldots\in\R^d$ be a sequence such that
	$\bA\ba_1,\bA\ba_2,\ldots$ converges against $\bv\in\R^m$ and
	$\bB$ is simultaneously unbounded on $\ba_1,\ba_2,\ldots$.  Then there
	exist 
(1)~$\ba\in\R^d$ with $\bA\ba=\bv$ and
(2)~$\bd_\infty\in\R^d$ with $\bA\bd_\infty=\bzero$ and
			$\bB\bd_\infty\gg\bzero$.
\end{lemma}
\begin{proof}
For a vector $\bv\in\R^m$, we write $\bv[j]$ for its $j$-th component.
For any $j\in\{0,\ldots,n\}$, let
\[ T_j=\{\bu\in\R^d \mid (\bA\bu)[i]=\bv[i]~\text{for all $1\le i\le j$}\}. \]
Note that $T_j$ is a \emph{coset}, meaning that if $\bu_1,\bu_2,\bu_3\in T_j$ and $\alpha\in\R$,
then $\bu_1+\alpha(\bu_2-\bu_3)\in T_j$.

We call a sequence $\ba_1,\ba_2,\ldots\in\R^d$ \emph{good} if
$\bA\ba_1,\bA\ba_2,\ldots$ converges against $\bv$ and $\bB$ is
simultaneously unbounded on $\ba_1,\ba_2,\ldots$.  We shall prove that for every $j\in\{0,\ldots,n-1\}$, 
if there is a good sequence contained in $T_j$, then there is a good
sequence contained in $T_{j+1}$. This implies that there is a good
sequence contained in $T_n$. This clearly yields the desired
$\ba$ as the first element of the sequence and $\bd_\infty$ as
the difference between two elements.

Thus, let $\ba_1,\ba_2,\ldots$ be a good sequence in $T_j$ for $j\in\{0,\ldots,n-1\}$. Without loss of generality, we
may assume that $(\bA\ba_1)[j+1], (\bA\ba_2)[j+1],\ldots$
converges against $\bv[j+1]$ from below and is strictly increasing: All other cases are symmetric or already yield a subsequence in $T_{j+1}$.
Moreover, let $N\in\N$ and $\varepsilon>0$. To construct the good sequence in $T_{j+1}$, we need to show that there
exists a vector $\ba'\in\R^d$ with $\bB\ba'\ge(N,\ldots,N)$ with
$\ba'\in T_{j+1}$ and $\|\bA\ba'-\bv\|<\varepsilon$. Since the sequence
$\ba_1,\ba_2,\ldots$ is good, we can choose $k,\ell\in\N$, $k<\ell$,
such that
\begin{enumerate}
	\item $\bB\ba_k\ge (N,\ldots,N)$
	\item $\bB\ba_k\ll\bB\ba_\ell$
	\item $\|\bA\ba_k-\bv\|,\|\bA\ba_\ell-\bv\|<\delta$
	\item $|(\bA\ba_\ell-\bv)[j+1]|<\tfrac{1}{2}|(\bA\bx_k-\bv)[j+1]|$,
\end{enumerate}
where $\delta>0$ will be chosen later. Since $|(\bA\ba_\ell-\bv)[j+1]|<\tfrac{1}{2}|(\bA\ba_k-\bv)[j+1]|$, there must be some $\alpha$ with $1\le\alpha\le 2$ such that
\[ \left(\bA\ba_k+\alpha(\bA\ba_\ell-\bA\ba_k)\right)[j+1]=\bv[j+1]. \]
We set $\ba':=\ba_k+\alpha(\ba_\ell-\ba_k)$. By the previous equation and the fact that $T_j$ is a coset, we have $\ba'\in T_{j+1}$. Moreover,
$\bB\ba_k\ge (N,\ldots,N)$ and $\bB\ba_\ell\gg\bB\ba_k$ also imply $\bB\ba'\ge (N,\ldots,N)$. Furthermore,
we have
\begin{align*}
	\|\bA\ba'-\bv\|&=\|\bA\ba_k+\alpha(\bA\ba_\ell-\bA\ba_k)-\bv\|\\
	&\le \|\bA\ba_k-\bv\|+\alpha\|\bA\ba_\ell-\bA\ba_k\|\le \delta+2\alpha\delta \le 5\delta.
\end{align*}
Thus, picking $\delta<\varepsilon/5$ yields the desired $\ba'$.
\end{proof}

\subsection*{Compatibility in terms of inequalities}
We are now ready to describe the necessary and sufficient conditions for the vectors $\bm{a}$, $\bm{d}_c$, and $\bm{d}_\infty$.
We define matrices and vectors to describe systems of linear (in)equalities that are needed to express the compatibility conditions.
Let $\bm p$ be a profile and define the following inequalities.
\begin{description}
	\item[Limit values] Let $\bm{L_p}$ be a matrix and $\bm{\ell_p}$ be a vector such that $\bm{L_p} \bm x = \bm{\ell_p}$ if and only if
		\begin{align*}
			\bm{r}_i^\top \bm{x}&=\rho(i) & & \text{for each $i$ with $t_{\rho}(i)\in\{-1,1\}$} \\
			\bm{s}_i^\top \bm{x}&=\sigma(i) & & \text{for each $i$ with $t_{\sigma}(i)\in\{-1,1\}$}\end{align*}
		Then, as discussed above, our vectors need to satisfy
		$\bm{L_p}\bm{a}=\bm{\ell_p}$ and $\bm{L_p}\bm{d}_\infty=\bm{0}$.

\item[Constant values] Let $\bm{C_p}$ be a matrix and $\bm{c_p}$ be a vector such that $\bm{C_p} \bm x = \bm{c_p}$ if and only if
		\begin{align*}
			\bm r_i^\top \bm x &= \rho(i) && \text{for each $i$ with $t_\rho(i) = 0$} \\
			\bm s_i^\top \bm x &= \sigma(i) && \text{for each $i$ with $t_\sigma(i) = 0$}
		\end{align*}
		Since components $i$ with $t_\rho(i)=0$ (resp.\
		$t_\sigma(i)=0$) are those where
		$\bm{r}_i^\top\bm{a}_1,\bm{r}^\top\bm{a}_2,\ldots$ (resp.\
		$\bm{s}_i^\top\bm{a}_1,\bm{s}^\top\bm{a}_2,\ldots$) is
		constant, our vectors clearly need to satisfy
		$\bm{C_p}\bm{d}_c=\bm{0}$ and $\bm{C_p}\bm{d}_\infty=\bm{0}$.

	\item[Convergence] Let $\bm{D_p}$ be a matrix such that $\bm{D_p}\bm{x}\gg \bm{0}$ if and only if
\begin{align*}
			\bm r_i^\top \bm x &> 0~\text{(resp.\ $< 0$)} & &\text{for each $i$ with $t_\rho(i) = 1$ (resp. $= -1$)}, \\
			\bm s_i^\top \bm x &> 0~\text{(resp.\ $< 0$)} & &\text{for each $i$ with $t_\sigma(i) = 1$ (resp. $= -1$)}
\end{align*}
Since the components $i$ with $t_\rho(i)=1$ (resp.\ $t_\rho(i)<0$) are those where $\bm{r}_i^\top\bm{a}_1,\bm{r}_i^\top\bm{a}_2,\ldots$ converges to a real number from below (resp.\ from above), and similarly for $\bm{s}_i^\top\bm{a}_1,\bm{s}_i^\top\bm{a}_2,\ldots$, we must have $\bm{D_p}\bm{d}_c\gg\bm{0}$.
\item[Unboundedness] Let $\bm{U_p}$ be a matrix such that $\bm{U_p} \bm x \gg \bm 0$ if and only if
	\begin{align*}
		\bm r_i^\top \bm x &> 0~\text{(resp. $< 0$)} &&\text{for each $i$ with $t_\rho(i) = \omega$ (resp. $= -\omega$)} \\
		\bm s_i^\top \bm x &> 0~\text{(resp. $< 0$)} && \text{for each $i$ with $t_\sigma(i) = \omega$ (resp. $= -\omega$)}
	\end{align*}
	Since the components $i$ with $t_{\rho}(i)=\omega$ are those where
	$\bm{r}_i^\top\bm{a}_1,\bm{r}_i^\top\bm{a}_2,\ldots$ diverges to
	$\infty$ (and analogous relationships hold for $t_\rho(i)=-\omega$ and
	for $t_\sigma(i)$), we must have $\bm{U_p}\bm{d}_\infty\gg\bm{0}$.
\end{description}

Let us now formally provide a list of necessary and sufficient conditions on
$\bm{a}$, $\bm{d}_c$, and $\bm{d}_\infty$ for the existence of a sequence
compatible with $\bm{p}$ that satisfies the equality constraints for $\bm{c}$.
\begin{lemma}\label{lem:compatible-vectors}
Let $\bm c \in \R^d$ and $\bm p$ be a profile.
Then there exists a sequence compatible with $\bm p$ that satisfies the equality constraints for $\bm c$ if and only if
there are vectors $\bm a, \bm d_c, \bm d_\infty \in \R^d$ with $\bm d_c \ne \bm 0$ with
\begin{enumerate}
\item $\bm{L_p} \bm a = \bm{\ell_p}$, $\bm{C_p} \bm a = \bm{c_p}$,
\item $\bm{D_p} \bm d_c \gg \bm 0$, $\bm{C_p} \bm d_c = \bm 0$,
\item $\bm{L_p} \bm d_\infty = \bm 0$, $\bm{C_p} \bm d_\infty = \bm 0$, $\bm{U_p} \bm d_\infty \gg \bm 0$, and
\item $\bm u_j^\top \bm d_c = \bm u_j^\top \bm d_\infty = 0$, $\bm v_j^\top \bm d_c = \bm v_j^\top \bm d_\infty = 0$,
$(\bm u_j^\top - \bm v_j^\top) \bm a = \bm w_j^\top \bm c + d_j$ for all $j \in \{1,\dots,m\}$.
\end{enumerate}
\end{lemma}
\begin{proof}
We start with the ``only if'' direction.
Let $(\bm a_k)_{k\ge1}$ be a sequence compatible with $\bm p$ that satisfies the equality constraints for $\bm c$.
First observe that the equality constraints imply that 
$\bm u_j^\top \bm a_k = \bm u_j^\top \bm a_\ell$,
$\bm v_j^\top \bm a_k = \bm v_j^\top \bm a_\ell$, and
$(\bm u_j^\top - \bm v_j^\top) \bm a_k = \bm w_j^\top \bm c + d_j$
for all $2 \le k < \ell$.
Thus, by removing the first vector of the sequence we can assume that $(\bm a_k)_{k\ge1}$ fulfills this property already for $1 \le k < \ell$.
For $\bm d_c$ we choose $\bm a_2 - \bm a_1$ where $\bm d_c \ne \bm 0$ since $\bm a_1 \ne \bm a_2$.
This fulfills (2) since the sequences $\bm \rho_i$ and $\bm \sigma_i$ are strictly increasing/decreasing if $t_\rho,t_\sigma \in \{-1,1\}$
and constant if $t_\rho = t_\sigma = 0$.
Moreover, $\bm d_c$ fulfills (4) since 
$\bm u_j^\top \bm d_c = \bm u_j^\top \bm a_2 - \bm u_j^\top \bm a_1 = 0$ and
$\bm v_j^\top \bm d_c = \bm v_j^\top \bm a_2 - \bm v_j^\top \bm a_1 = 0$.
Let $\bm A$ be the matrix obtained by concatenating $\bm{L_p}$ and $\bm{C_p}$ vertically and
adding the rows $\bm u_j^\top$, $\bm v_j^\top$, and $\bm u_j^\top - \bm v_j^\top$ for all $j \in \{1,\dots,m\}$.
In parallel, we define the vector $\bm v$ as the vertical concatenation of $\bm{\ell_p}$ and $\bm{c_p}$
extended by the entry $\bm u_j^\top \bm a_1$ in the row of $\bm u_j^\top$, 
the entry $\bm v_j^\top \bm a_1$ in the row of $\bm v_j^\top$, and
the entry $\bm w_j^\top \bm c + d_j$ in the row of $\bm u_j^\top - \bm v_j^\top$.
Then we have that the sequence $(\bm A \bm a_k)_{k\ge1}$ converges against $\bm v$.
We can now apply \Cref{existence-d-infty} for $\bm A$, $\bm v$, and $\bm B := \bm{U_p}$
to obtain vectors $\bm a$ and $\bm d_\infty$ with the desired properties.

For the ``if'' direction let $\bm a, \bm d_c, \bm d_\infty \in \R^d$ be as in the lemma.
We claim that the sequence with $\bm a_k := \bm a -\frac{1}{k}\bm d_c + k \bm d_\infty$ for all $k \ge k_0$ and sufficiently large $k_0$ is as desired.
For convenient notation, define the sequence $\bm \rho_i = \bm r_i^\top \bm a_k$ for each $i$.

We first show that the sequence $\bm{a}_1,\bm{a}_2,\ldots$ is compatible with $\bm p$.
Since $\bm d_c \ne \bm 0$, the $\bm a_k$ are pairwise distinct for all $k \ge k_0$ and sufficiently large $k_0$.
If $t_\rho(i) = 1$ (resp. $t_\rho(i) = -1$), then
$\bm r_i^\top \bm a_k = \bm r_i^\top \bm a -\frac{1}{k}\bm r_i^\top \bm d_c + k \bm r_i^\top \bm d_\infty$
and since $\bm r_i^\top \bm d_\infty = 0$, $\bm r_i^\top \bm a = \rho(i)$, and $\bm r_i^\top \bm d_c > 0$ (resp. $< 0$),
we have that the sequence $\bm \rho_i$ is strictly increasing (resp. decreasing) and converges against $\rho(i)$ from below (resp. above).
If $t_\rho(i) = 0$, then $\bm \rho_i$ is constantly $\rho(i)$ since
$\bm r_i^\top \bm d_\infty = 0$, $\bm r_i^\top \bm a = \rho(i)$, and $\bm r_i^\top \bm d_c = 0$.
Finally, if $t_\rho(i) = \omega$ (resp. $t_\rho(i) = -\omega$), then
$\bm r_i^\top \bm d_\infty > 0$ (resp. $< 0$) which means that for sufficiently large $k_0$, 
the sequence $(\bm \rho_i)_{k\ge k_0} = (\bm r_i^\top \bm a_k)_{k \ge k_0}$ is strictly increasing (resp.\ decreasing) and diverges to $\infty$ (resp. $-\infty$).
The statement can be shown analogously for $\bm \sigma_i$.

It remains to show that the sequence satisfies the equality constraints for $\bm c$.
By (4) we have that
$\bm u_j^\top (\bm a -\frac{1}{k}\bm d_c + k \bm d_\infty) = \bm v_j^\top (\bm a -\frac{1}{\ell}\bm d_c + \ell \bm d_\infty) + \bm w_j^\top \bm c + d_j$
for $k < \ell$ if and only if
$\bm u_j^\top \bm a  = \bm v_j^\top \bm a + \bm w_j^\top \bm c + d_j$
which holds if and only if
$(\bm u_j^\top - \bm v_j^\top) \bm a = \bm w_j^\top \bm c + d_j$
which is fulfilled by (4).
\end{proof}

\subsection*{Constructing the formula}
We now prove \Cref{thm:real-main} in the general case, i.e., $\varphi(\bm{x},\bm{y},\bm{z})$ is an arbitrary existential LRA formula.
If $\varphi$ is a conjunction of inequalities, \Cref{lem:compatible-vectors} essentially tells us how to construct an existential formula for $\ram\bm{x},\bm{y}\colon\varphi(\bm{x},\bm{y},\bm{z})$. Moreover, by \cref{thm:qe-mixed}, we may assume $\varphi$ to be quantifier-free.
Thus, it remains to treat the case that $\varphi$ is a Boolean combination of constraints as in \eqref{reals-conjunction}. 

We first move all negations inward and, if necessary, negate atoms, so that we are left with a positive Boolean combination of atoms.
Let $\alpha_i := \bm{r}_i^\top \bm{x} < \bm{s}_i^\top\bm{y} + \bm{t}_i^\top\bm{z}+h_i$ for $i \in [1,n]$ be the inequality atoms and
$\beta_j := \bm{u}_j^\top \bm{x} = \bm{v}_j^\top\bm{y}+\bm{w}_j^\top\bm{z}+d_j$ for $j \in [1,m]$ be the equality atoms in $\varphi$.

As in the Presburger case, we now guess a subset of the atoms and then assert that (i)~satisfying all these atoms makes $\varphi$ true and (ii)~there exists a clique satisfying the conjunction of these atoms.

Let $\varphi'$ be the formula obtained from $\varphi$ by replacing each $\alpha_i$ by $q_i^< = 1$ for a fresh variable $q_i^<$, for all $i \in [1,n]$, and
each $\beta_j$ by $q_j^= = 1$ for a fresh variable $q_j^=$, for all $j \in [1,m]$,
and add the restrictions $q_i^< = 0 \vee q_i^< = 1$ and $q_j^= = 0 \vee q_j^= = 1$.
Now, $\varphi$ is equivalent to
\[\psi := \exists \bm{q}^<, \bm{q}^= \colon \varphi' \wedge \bigwedge_{i=1}^n (q_i^< = 1 \to \alpha_i) \wedge 
\bigwedge_{j=1}^m (q_j^= = 1 \to \beta_j).\]
We represent a profile $\bm p$ by the variables $\rho_i$, $\sigma_i$, $t_{\rho,i}$, and $t_{\sigma,i}$ for all $i \in [1,n]$
where $\rho_i,\sigma_i$ range over $\R$ and $t_{\rho,i}, t_{\sigma,i}$ range over $\{-2,-1,0,1,2\}$.
Here, $-2$ and $2$ represent $-\omega$ and $\omega$, respectively.

We now define formulas for the inequalities and equality constraints from \Cref{lem:compatible-vectors}.
For $i \in [1,n]$, let $\rho_i, \sigma_i, t_{\rho,i}, t_{\sigma,i}, \bm x, \bm x_c, \bm x_\infty$ be fresh variables. Our first formula $\lambda_i$ contains all the constraints from $\bm{L}_p\bm{x}=\bm{\ell_p}$ and $\bm{L_p}\bm{x}_\infty=\bm{0}$ that stem from the atom $\alpha_i$:
\begin{align*}
\lambda_i :=\ & ((t_{\rho,i} = -1 \vee t_{\rho,i} = 1) \to \bm r_i^\top \bm x = \rho_i \wedge \bm r_i^\top \bm x_\infty = 0) \ \wedge \\
& ((t_{\sigma,i} = -1 \vee t_{\sigma,i} = 1) \to \bm s_i^\top \bm x = \sigma_i \wedge \bm s_i^\top \bm x_\infty = 0)
\end{align*}
Next, $\chi_i$ states the constraints about constant values---meaning: those from $\bm{C_p}\bm{x}=\bm{c_p}$ and $\bm{C_p}\bm{x}_c=\bm{C_p}\bm{x}_\infty=\bm{0}$---that stem from $\alpha_i$:
\begin{align*}
\chi_i :=\ & (t_{\rho,i} = 0 \to \bm r_i^\top \bm x = \rho_i \wedge \bm r_i^\top \bm x_c = 0 \wedge \bm r_i^\top \bm x_\infty = 0) \ \wedge \\
& (t_{\sigma,i} = 0 \to \bm s_i^\top \bm x = \sigma_i \wedge \bm s_i^\top \bm x_c = 0 \wedge \bm s_i^\top \bm x_\infty = 0)
\end{align*}
With $\delta_i$, we express the convergence constraints from $\bm{D_p}\bm{x}_c\gg \bm{0}$ required by $\alpha_i$:
\begin{align*}
\delta_i :=\ & (t_{\rho,i} = -1 \to \bm r_i^\top \bm x_c < 0) \wedge
(t_{\rho,i} = 1 \to \bm r_i^\top \bm x_c > 0) \ \wedge \\
& (t_{\sigma,i} = -1 \to \bm s_i^\top \bm x_c < 0) \wedge
(t_{\sigma,i} = 1 \to \bm s_i^\top \bm x_c > 0)
\end{align*}
Furthermore, $\mu_i$ states the unboundedness condition $\bm{U}_p\bm{x}_\infty\gg\bm{0}$:
\begin{align*}
\mu_i :=\ & (t_{\rho,i} = -2 \to \bm r_i^\top \bm x_\infty < 0) \wedge
(t_{\rho,i} = 2 \to \bm r_i^\top \bm x_\infty > 0) \ \wedge \\
& (t_{\sigma,i} = -2 \to \bm s_i^\top \bm x_\infty < 0) \wedge
(t_{\sigma,i} = 2 \to \bm s_i^\top \bm x_\infty > 0)
\end{align*}
Finally, $\varepsilon_j$ expresses the equality constraints~(5) in \cref{lem:compatible-vectors}: For all $j \in [1,m]$ let
\begin{align*}
\varepsilon_j :=\ & \bm u_j^\top \bm x_c = 0 \wedge \bm u_j^\top \bm x_\infty = 0 \wedge \bm v_j^\top \bm x_c = 0 \wedge \bm v_j^\top \bm x_\infty = 0 \wedge
(\bm u_j^\top - \bm v_j^\top) \bm x = \bm w_j^\top \bm z + x_j.
\end{align*}
To check if $\bm{p}$ is a $\bm z$-admissible profile, we define the formula
\begin{align*}
\theta :=\ & \bigwedge_{i=1}^n t_{\rho,i} \in \{-2,-1,0,1,2\} \wedge t_{\sigma,i} \in \{-1,0,1,2\} \wedge (t_{\rho,i} = 2 \to t_{\sigma,i} = 2) \ \wedge \\
& [(t_{\rho,i} \in \{-1, 0\} \wedge t_{\sigma,i} \in \{0, 1\} \vee t_{\rho,i} = -1 \wedge t_{\sigma,i} = -1)
\to \rho_i < \sigma_i + \bm t_i^\top \bm z + h_i] \ \wedge \\
& [(t_{\rho,i} = 0 \wedge t_{\sigma,i} = -1 \vee t_{\rho,i} = 1) \to \rho_i \le \sigma_i + \bm t_i^\top \bm z + h_i]
\end{align*}
where we use set notation as a shorthand.
Then we claim that $\ram \bm{x},\bm{y} \colon \psi(\bm{x},\bm{y},\bm{z})$ is equivalent to
\begin{align*}
\gamma :=\ & \exists \bm{q}^<, \bm{q}^=, \bm{p}, \bm{x}, \bm x_c, \bm x_\infty \colon \varphi' \wedge \theta \wedge \bm x_c \ne \bm{0} \ \wedge \\
&\bigwedge_{i=1}^n (q_i^< = 1 \to \lambda_i \wedge \chi_i \wedge \delta_i \wedge \mu_i) \wedge 
\bigwedge_{j=1}^m (q_j^= = 1 \to \varepsilon_j).
\end{align*}
We show that for any valuation $\bm{c} \in \R^d$ of $\bm{z}$ we have
$\ram \bm{x},\bm{y} \colon \psi(\bm{x},\bm{y},\bm{c})$ if and only if $\gamma(\bm{c})$.
For an assignment $\nu$ of the $q_i^<,q_j^=$ to $\{0,1\}$ let 
$I_\nu := \{i \in [1,n] \mid \nu(q_i^<) = 1\}$ and $J_\nu := \{j \in [1,m] \mid \nu(q_j^=) = 1\}$.
By Ramsey's theorem, $\ram \bm{x},\bm{y} \colon \psi(\bm{x},\bm{y},\bm{c})$ holds if and only if there 
there is an assignment $\nu$ of the $q_i^<,q_j^=$ satisfying $\varphi'$ such that
$\ram \bm{x},\bm{y} \colon \bigwedge_{i \in I_\nu} \alpha_i(\bm{x},\bm{y},\bm{c}) \wedge \bigwedge_{j \in J_\nu} \beta_j(\bm{x},\bm{y},\bm{c})$.
By \Cref{lem:ramsey-compatible,lem:compatible-vectors}, this is equivalent to
$\exists \bm{p},\bm{x},\bm x_c, \bm x_\infty \colon \theta(\bm p,\bm c) \wedge \bm x_c \ne \bm{0} \wedge 
\bigwedge_{i \in I_\nu} \lambda_i \wedge \chi_i \wedge \delta_i \wedge \mu_i \wedge 
\bigwedge_{j \in J_\nu} \varepsilon_j(\bm{x}, \bm x_c, \bm x_\infty, \bm{c})$.
This holds for some assignment $\nu$ of the $q_i^<,q_j^=$ satisfying $\varphi'$ if and only if $\gamma(\bm{c})$.

Using standard arguments, one can observe that \cref{thm:real-main} has an analogue over the rationals:
\begin{theorem}\label{thm:rational-main}
Given an existential formula $\varphi(\bm x, \bm y, \bm z)$
over $\QLin$, we can construct in polynomial time an
existential formula over $\QLin$ of linear size that is
equivalent to $\ram \bm x, \bm y \colon \varphi(\bm x, \bm
y, \bm z)$.
\end{theorem}
\begin{proof}
	We can view $\varphi(\bm x, \bm y, \bm z)$ as a
	formula over LRA and thus obtain a formula
	$\psi(\bm z)$ that is equivalent to $\ram \bm x,\bm
	y\colon \varphi(\bm x,\bm y,\bm z)$ over LRA. We
	claim that $\psi(\bm z)$ is equivalent to $\ram \bm
	x,\bm y\colon \varphi(\bm x,\bm y,\bm z)$ over
	$\QLin$ as well.
	
	Suppose $\ram \bm x,\bm y\colon \varphi(\bm x,\bm
	y, \bm{c})$ holds in $\QLin$ for some vector
	$\bm c\in\Q^{|\bm z|}$. Then in particular, $\ram
	\bm x,\bm y\colon \varphi(\bm x,\bm y, \bm{c})$
	holds in LRA and thus by construction,
	$\psi(\bm{c})$  holds in LRA. Since LRA and $\QLin$
	are elementarily equivalent, this implies that
	$\psi(\bm{c})$ holds in $\QLin$.
	
	Conversely, suppose $\psi(\bm c)$ holds for some
	$\bm{c}\in\Q^{|\bm z|}$. Recall that $\psi$ is
	constructed by first writing $\varphi$ equivalently
	as a disjunction $\bigvee_{i=1}^n \varphi_i$, where
	each $\varphi_i$ is a conjunction as in
	\cref{reals-conjunction}. Then, to each
	$\varphi_i$, we associate the system $S_i$ of
	linear equalities described in
	\cref{lem:compatible-vectors}, where one can view
	the entries of $\bm{p}$ as variables. Since
	$\psi(\bm c)$ holds in $\QLin$, some system $S_i$
	has a rational solution. Moreover, since we treat
	the entries of $\bm p$ as variables, $\bm p$ is
	rational as well. Moreover, since it stems from a
	satisfying assignment of $\psi(\bm c)$, $\bm{p}$ is
	also $\bm c$-admissible. Now an inspection of the
	proof of \cref{lem:compatible-vectors} yields that
	the clique constructed based on a rational solution
	of $S_i$ will be rational as well. Therefore, there
	exists a rational clique that witnesses $\ram \bm
	x,\bm y\colon\varphi(\bm x,\bm y,\bm c)$.
\end{proof}
 
\section{Ramsey quantifiers in Linear Integer Real Arithmetic}\label{sec:mixed}
We now show elimination of the Ramsey quantifier in LIRA. At the end of the section, we mention a version of this result for the structure $\QfLin$ (\cref{thm:mixed-rational-main}).
\begin{theorem}\label{thm:mixed-main}
Given an existential formula $\varphi(\bm x, \bm y, \bm z)$ in LIRA,
we can construct in polynomial time an existential formula in LIRA of linear size that is equivalent to $\ram \bm x, \bm y \colon \varphi(\bm x, \bm y, \bm z)$.
\end{theorem}
\begin{proof}
It suffices to show the theorem for the decomposition of $\varphi$:
Given $\varphi$, we first compute its decomposition $\varphi'(\bm x^\mathrm{i/r},\bm y^\mathrm{i/r},\bm z^\mathrm{i/r})$ using \Cref{lem:decomposition}.
We then show how to compute a formula $\psi'(\bm z^\mathrm{i/r})$ in LIRA that is equivalent to 
$\ram \bm x^\mathrm{i/r},\bm y^\mathrm{i/r} \colon \varphi'(\bm x^\mathrm{i/r},\bm y^\mathrm{i/r},\bm z^\mathrm{i/r})$.
Let $\psi(\bm z)$ be the formula obtained from $\psi'$ by replacing every 
$z_i^\mathrm{int}$ by $\lfloor z_i \rfloor$ and every $z_i^\mathrm{real}$ by $z_i - \lfloor z_i \rfloor$.
Now $\psi$ is equivalent to $\ram \bm x, \bm y \colon \varphi(\bm x,\bm y,\bm z)$ since 
$\ram \bm x, \bm y \colon \varphi(\bm x,\bm y,\bm c)$ if and only if $\bm c = \bm c^\mathrm{int}+ \bm c^\mathrm{real}$ and
$\ram \bm x^\mathrm{i/r},\bm y^\mathrm{i/r} \colon \varphi'(\bm x^\mathrm{i/r},\bm y^\mathrm{i/r},\bm c^\mathrm{i/r})$.

Thus, we now assume that $\varphi'(\bm x^\mathrm{i/r},\bm y^\mathrm{i/r},\bm z^\mathrm{i/r})$ is a decomposition of $\varphi$.
By \Cref{thm:qe-mixed} we can assume that $\varphi'$ is quantifier-free.
We further assume that all negations are moved directly into the atoms. Let
$\alpha_1,\dots,\alpha_n$ be the atoms in LRA and $\beta_1,\dots,\beta_m$ be the Presburger atoms of $\varphi'$.
For fresh real variables $p_1,\dots,p_n$ and integer variables $q_1,\dots,q_m$ let $\sigma$ be the formula obtained from $\varphi'$ by replacing
every $\alpha_i$ by $p_i = 1$ and every $\beta_j$ by $q_j = 1$ and adding the constraints $p_i = 0 \vee p_i = 1$ and $q_j = 0 \vee q_j = 1$.
Then $\varphi'$ is equivalent to
\[\delta := \exists \bm p, \bm q \colon \sigma \wedge \bigwedge_{i=1}^n (p_i = 1 \to \alpha_i) \wedge \bigwedge_{j=1}^m (q_j = 1 \to \beta_j).\]
Since each $p_i$ and $q_j$ only has finitely many (in fact two) possible valuations, Ramsey's theorem implies that
$\ram \bm x^\mathrm{i/r}, \bm y^\mathrm{i/r} \colon \delta(\bm x^\mathrm{i/r},\bm y^\mathrm{i/r},\bm z^\mathrm{i/r})$
is equivalent to
\[\exists \bm p,\bm q \colon \sigma \wedge \ram \bm x^\mathrm{i/r},\bm y^\mathrm{i/r} \colon \bigwedge_{i=1}^n (p_i = 1 \to \alpha_i) \wedge
\bigwedge_{j=1}^m (q_j = 1 \to \beta_j).\]
Let $\alpha := \bigwedge_{i=1}^n (p_i = 1 \to \alpha_i)$ and $\beta := \bigwedge_{j=1}^m (q_j = 1 \to \beta_j)$.
To split the vectors of the Ramsey quantified variables into real and integer components, 
we have to allow that in the infinite clique either the real components or the integer components do not change throughout the clique.
To this end, we introduce a fresh variable $r$ that is either 0 or 1 and get the equivalent formula
\begin{align*}
\exists \bm p,\bm q, r \colon & \sigma \wedge (r = 0 \vee r = 1) \ \wedge\\
& \big[\big(\ram \bm x^\mathrm{real},\bm y^\mathrm{real} \colon \alpha(\bm p,\bm x^\mathrm{real},\bm y^\mathrm{real},\bm z^\mathrm{real})\big)
\vee r = 0 \wedge \exists \bm x^\mathrm{real} \colon \alpha(\bm p,\bm x^\mathrm{real},\bm x^\mathrm{real},\bm z^\mathrm{real})\big] \ \wedge\\
& \big[\big(\ram \bm x^\mathrm{int},\bm y^\mathrm{int} \colon \beta(\bm q,\bm x^\mathrm{int},\bm y^\mathrm{int},\bm z^\mathrm{int})\big)
\vee r = 1 \wedge \exists \bm x^\mathrm{int} \colon \beta(\bm q,\bm x^\mathrm{int},\bm x^\mathrm{int},\bm z^\mathrm{int})\big]
\end{align*}
where the Ramsey quantifiers can be eliminated by \Cref{thm:real-main,main-presburger}. 
\end{proof}

Let us mention a simple consequence.
\begin{corollary}\label{lira-infinite-clique-np}
The infinite clique problem for existential formulas in LIRA is $\NP$-complete.
\end{corollary}
\begin{proof}
	The $\NP$ lower bound already holds for existential formulas in Presburger arithmetic and LRA by a reduction from the respective satisfiability problem
	which by \Cref{prop:sat-presburger,prop:sat-reals} is $\NP$-complete:
	Let $\varphi(\bm x)$ be an existential formula in Presburger or LRA.
	Then $\varphi$ is satisfiable if and only if
	$\ram (\bm x,v),(\bm y,w) \colon \varphi(\bm x)$.
For the upper bound let $\varphi(\bm x,\bm y)$ be an existential formula in LIRA where $\bm x$ and $\bm y$ have the same dimension.
	By \Cref{thm:mixed-main} one can compute an existential formula $\psi$ in LIRA that is equivalent to $\ram \bm x, \bm y \colon \varphi(\bm x,\bm y)$ in polynomial time.
	Then, the $\NP$ upper bound follows from \Cref{prop:sat-mixed}.
\end{proof}
 
\begin{theorem}\label{thm:mixed-rational-main}
Given an existential formula $\varphi(\bm x, \bm y, \bm z)$ over $\QfLin$,
we can construct in polynomial time an existential formula of linear size that is equivalent to $\ram \bm x, \bm y \colon \varphi(\bm x, \bm y, \bm z)$ over $\QfLin$.
\end{theorem}
\begin{proof}
	We use almost the same construction as for
	\cref{thm:mixed-main}. The only difference is that
	we use \cref{thm:rational-main} in place of
	\cref{thm:real-main}.
\end{proof}
 
\section{Applications}\label{sec:applications}
In this section, we present further applications of our results.

\subsection*{Monadic decomposability}
A formula is called \emph{monadic} if every atom contains at most one variable. As mentioned above, monadic formulas play an important role
in constraint databases\cite{GRS01,CDB-book}. Partly motivated by this, \citet{DBLP:journals/jacm/VeanesBNB17} recently raised the question of how to decide whether a given formula $\varphi$ is equivalent to a monadic formula. In this case, $\varphi$ is called \emph{monadically decomposable}.
For LIA, monadic decomposability was shown decidable (under slightly
different terms) by \citet[Corollary, p.~1048]{ginsburg1966bounded} and a more general
result by \citet[Theorem~3]{DBLP:journals/tocl/Libkin03} establishes dedicability
for LIA, LRA, and other logics~\cite[Corollaries~7,8]{DBLP:journals/tocl/Libkin03}. In terms of complexity, given a quantifier-free LIA formula, monadic decomposability was shown $\coNP$-complete by \citet{DBLP:conf/cade/HagueLRW20}. However, it remained open what the complexity is in the case of LRA and LIRA. From \cref{main-presburger,thm:real-main,thm:mixed-main}, we can conclude the following:
\begin{corollary}\label{corollary-mon-dec}
Given a quantifier-free formula in LIA, LRA, or LIRA, deciding monadic decomposability is $\coNP$-complete.
\end{corollary}
As mentioned above, the result about LIA was also shown by \citet{DBLP:conf/cade/HagueLRW20}.
The $\coNP$-hardness in \cref{corollary-mon-dec} uses the same idea as \cite[Lemma 6.4]{lics22-ramsey}. 
We show that monadic decomposability for formulas in LRA is $\coNP$-hard:  We
reduce from the unsatisfiability problem for quantifier-free formulas in LRA.
Suppose $\psi(\bm{x})$ is a quantifier-free formula in LRA. Now consider the
formula
\[ \varphi(\bm{x},y,z): = \neg\psi(\bm{x})\vee y=z, \]
where $y$ and $z$ are (fresh) single variables. We claim that $\varphi$ is
\emph{not} monadically decomposable iff $\psi$ is satisfiable.

Let $\psi$ be satisfiable with some $\bm c\in\R^{|\bm x|}$. Without loss of
generality, we may assume $\bm c\in\Q^{|\bm x|}$. Towards a contradiction,
suppose $\varphi$ is monadically decomposable. Then the formula
$\varepsilon(y,z):=\varphi(\bm{c},y,z)$ must also be monadically decomposable.
However, we have $\varepsilon(y,z)$ if and only if $y=z$, which is not
monadically decomposable.

Conversely, suppose $\psi$ is unsatisfiable. Then $\varphi(\bm x,y,z)$ is
satisfied for all $\bm x, y, z$, which means $\varphi$ is clearly monadically
decomposable.

The $\coNP$ upper bound in \cref{corollary-mon-dec} follows from \cref{main-presburger,thm:real-main,thm:mixed-main} as follows. 
 Suppose $\varphi(x,\bm{y})$ is a formula in LIA, LRA, or LIRA with some free variable $x$ and a further vector $\bm{y}$ of free variables. We define the equivalence $\sim_{\varphi,x}$ on the domain $D$ (i.e.\ $\R$ or $\Z$) by
\[ a\sim_{\varphi,x} b \iff \text{for all $\bm{c}\in D^{|\bm{y}|}$, we have $\varphi(a,\bm{c})$ iff $\varphi(b,\bm{c})$}. \]
Note that if $\varphi$ is quantifier-free, we can easily construct a linear-size existential formula for the negation of $\sim_{\varphi,x}$ by setting
$ \delta_{\varphi,x}(x,x'):=\exists \bm{y}\colon \neg(\varphi(x,\bm{y})\leftrightarrow\varphi(x',\bm{y}))$.
For LIA, LRA, and LIRA, the formula $\varphi(x_1,\ldots,x_n)$ is monadically decomposable if and only if for each $i\in\{1,\ldots,n\}$, the equivalence $\sim_{\varphi,x_i}$ has only finitely many equivalence classes. This is shown in \cite[Lemma~4]{DBLP:journals/tocl/Libkin03} for LIA and LRA and in \cite[Lemma 10]{DBLP:journals/corr/abs-2304-11034} for LIRA.
Thus, the formula $\varphi(x_1,\ldots,x_n)$ is \emph{not} monadically decomposable if and only if $\mu_x:=\ram (x,x')\colon \delta_{\varphi,x_i}$
holds for some $i\in\{1,\ldots,n\}$. Thus, we can decide monadic non-decomposability by deciding in $\NP$ each of the $n$ formulas $\mu_x$ by applying \cref{main-presburger,thm:real-main,thm:mixed-main}.

We should mention that although a $\coNP$ algorithm was known for LIA, our new procedure to decide monadic decomposability is asymptotically much more efficient than the one by \citet{DBLP:conf/cade/HagueLRW20}: They construct for each variable $x$ a formula $\nu_x$ that contains an exponential constant $B$~\cite[p.~128]{DBLP:conf/cade/HagueLRW20}. They choose $B=2^{dmn+3}$~\cite[p.~132]{DBLP:conf/cade/HagueLRW20}, where (i)~$d$ is the number of bits needed to encode constants in $\varphi$, (ii)~$n$ is the number of linear inequalities in any disjunct in $\varphi$, and (iii)~$m$ is the number of variables in $\varphi$. Thus, this constant requires $dmn+3$ bits, meaning $\nu_x$ is of length $O(dmn)$, which is cubic in the size of the input formula. In contrast, each of our formulas $\mu_x$ (and thus the result after eliminating $\ram$) is of linear size.

\subsection*{Linear liveness for systems with counters and clocks}
As already observed by \citet{lics22-ramsey}, the Ramsey quantifier can be used to
check liveness properties of formal systems, provided that the reachability
relation is expressible in the respective logic. This yields several
applications for systems that involve counters and/or clocks.

Specifically, there is a rich variety of models where a configuration is an
element of $C=Q\times\Z^k\times\D^\ell$, where $Q$ is a finite set of control
states, and $\D$ is either $\R$ or $\Q$, with a step relation
$\mathord{\to}\subseteq C\times C$, and for $p,q\in Q$, one can effectively
construct an existential first-order formula $\varphi_{p,q}(\bm{x},\bm{y})$ for the reachability relation: This means, 
$(p,\bm{x})\xrightarrow{*}(q,\bm{y})$ if and only if
$\varphi_{p,q}(\bm{x},\bm{y})$. Here the components $\Z^k$ and $\D^\ell$ hold counter or clock values.
We will see some concrete examples below.

For systems of this type, we can consider the \emph{linear liveness problem}:
\begin{description}
\item[Given] A description of a system, a formula $\psi(\bm{x},\bm{y},\bm{z})$, and a state $q$.
\item[Question] Is there an infinite run $(q_1,\bm{u}_1)\to(q_2,\bm{u}_2)\to\cdots$ and a vector $\bm{v}$ such that for some infinite set $I\subseteq \N$, we have $q_i=q$ for every $i\in I$ and $\psi(\bm{u}_i,\bm{u}_j,\bm{v})$ for any $i,j\in I$ with $i<j$.
\end{description}
Here, a simple case is that $\psi$ simply states a linear condition on each
configuration (thus, $\psi(\bm{x},\bm{y},\bm{z})$ would just depend on
$\bm{x}$). But one can also require that between $(q_i,\bm{u}_i)$ and
$(q_j,\bm{u}_j)$, the values in $\bm{u}_i$ and $\bm{u}_j$ have increased by at
least some positive value in $\bm{v}$. With this, one can express, e.g.\ that
clock values grow unboundedly (rather than converging).

If we have reachability formulas $\varphi_{p,q}$ as above, linear liveness can easily be decided using the Ramsey quantifier: Note that there is a run as above if and only if for some state $q\in Q$, we have
\begin{equation} \exists \bm{z}\colon \ram (\bm{x},\bm{y})\colon \varphi_{q,q}(\bm{x},\bm{y})\wedge \psi(\bm{x},\bm{y},\bm{z}). \label{linear-liveness-formula} \end{equation}
Let us see some applications of this.

\paragraph{Timed (pushdown) automata}
	In a \emph{timed automaton}~\cite{alur1994theory}, configurations are elements of $Q\times\R^\ell$ and the real numbers are clock values. In each step, some time can elapse or, depending on satisfaction of guards, some counters can be reset; see \cite{alur1994theory} for details. It was shown by \citet[Theorem~5]{DBLP:conf/concur/ComonJ99} that the reachability relation in timed automata is effectively definable in $\langle\R;+,<,1,0\rangle$. Using a conceptually simpler construction, \citet[Theorem~10]{QuaasSW17} construct an exponential-size existential formula in $\langle\R;+,<,0,1\rangle$ for the reachability relation. Using the formula \eqref{linear-liveness-formula} and our results, we can thus decide the linear liveness problem for timed automata in $\NEXPTIME$. Recall that liveness in timed automata is $\PSPACE$-complete~\cite[Theorem 4.17]{alur1994theory}. The difference to linear liveness is that in the latter, one can express arbitrary LIRA constraints (even between configurations).

In order to model timed behavior of recursive programs, timed automata have
been extended by stacks, where each stack either
has~\cite{DBLP:conf/lics/AbdullaAS12} or does not
have~\cite{DBLP:conf/hybrid/BouajjaniER94} its own clock value. These two
versions are semantically equivalent and have been extended to \emph{timed
pushdown automata}~\cite{DBLP:conf/icalp/ClementeL18}, a strict extension that
allows additional counter constraints.
\citet[Theorem~5]{DBLP:conf/icalp/ClementeL18} show that the reachability
relation, between two configurations with empty stack, is definable by a
doubly-exponential existential formula over $\langle\Q;+,<,0,1\rangle$ (for the more restricted model of \citet{DBLP:conf/lics/AbdullaAS12}, existence of such a formula had been shown by \citet{DBLP:journals/tcs/Dang03}, but without complexity bounds).  Based on
this, our results allow us to decide the linear liveness problem for timed
pushdown automata in $\TwoNEXPTIME$, if we view each run from empty stack
to empty stack as a single step of the system.

\paragraph{Continuous vector addition systems with states}
	Vector addition systems with states (VASS; a.k.a.\
	Petri nets) are arguably the most popular formal
	model for concurrent systems. They consist of a
	control state and some counters that assume natural
	numbers. Since the reachability problem is
	Ackermann-complete~\cite{DBLP:conf/focs/Leroux21,DBLP:conf/focs/CzerwinskiO21,DBLP:conf/lics/LerouxS19}
	and the coverability
	problem is 
	$\EXPSPACE$-complete~\cite{Rackoff78,lipton1976reachability},
	there has been substantial interest in 
	finding overapproximations where these problems
	become easier.

	A particularly successful overapproximation is the
	\emph{continuous semantics}, where each added
	vector is non-deterministically multiplied by some
	$0<\alpha<1$. This has been used
	to speed up the backward search
	procedure by pruning configurations that
	cannot cover the target
	continuously~\cite{DBLP:conf/tacas/BlondinFHH16}.
	Thus, in continuous semantics, the configurations
	belong to $Q\times\Q^\ell$.
	
	As shown by \citet[Theorem
	4.9]{DBLP:conf/lics/BlondinH17}, the reachability
	relation under continuous semantics can be
	described by a polynomial-size existential
	formula over
	$\langle\Q;+,<,0,1\rangle$. Thus, our results imply
	$\NP$-completeness of the linear liveness problem 
	for continuous VASS ($\NP$-hardness easily
	follows from $\NP$-hardness of reachability~\cite[Theorem
	4.14]{DBLP:conf/lics/BlondinH17}).

\paragraph{Systems with discrete counters} There are several
	prominent types of counter systems for which one
	can compute an existential Presburger formula for the reachability relation.
	The most well-known example are
	\emph{reversal-bounded counter
	machines} (RBCM)~\cite{DBLP:journals/jacm/Ibarra78}. These
	admit an existential formula for the reachability
	relation, even if one of the counters has no
	reversal bound~\cite[Theorem~12]{DBLP:conf/mfcs/IbarraSDBK00}, even with a polynomial-time construction~\cite{DBLP:conf/cav/HagueL11}.

	Closely related to RBCM are \emph{Parikh
	automata} (PA)~\cite{DBLP:conf/icalp/KlaedtkeR03}, for
	which one can also compute an existential Presburger
	formula for the reachability relation in polynomial
	time. 

	Another example is the class of \emph{succinct
	one-counter systems}, which have one unrestricted
	counter with binary-encoded updates. Based on~\cite{DBLP:conf/concur/HaaseKOW09},
	\citet{LiCWX20} have shown that one can construct
	in polynomial time an existential Presburger formula for the reachability
    relations. In fact, using the proof techniques in \cite{HL12,To9}, this 
    result can be extended to multithreaded
    programs with $k$ threads --- each represented as a succinct
    one-counter system --- where inter-thread communication is limited, e.g.,
    when the number of context switches is also fixed (in the style of 
    \cite{QR05}). 

	Thus, our results imply that for these models, the
	linear liveness problem is $\NP$-complete. (Again,
	$\NP$-hardness follows using a simple reduction
	from the reachability problem.) In the case of
	PA, this strengthens recent results on
	PA over infinite
	words~\cite{DBLP:journals/corr/abs-2301-08969,DBLP:conf/fsttcs/GuhaJL022}.

\subsection*{Deciding whether a relation is a WQO} The
concept of well-structured transition
systems~(WSTS)~\cite{DBLP:conf/lics/AbdullaCJT96,DBLP:journals/tcs/FinkelS01}
is a cornerstone of the verification of infinite-state
systems. Here, the key idea is to order the configurations
of a system by a well-quasi-ordering (WQO). This recently
led~\citet{DBLP:conf/fsttcs/FinkelG19} to consider the
problem of automatically establishing that a given counter
machine is well-structured. In particular, they raised the
problem of deciding whether a relation, specified by a
formula in Presburger arithmetic, is a well-quasi ordering.
\citet{DBLP:journals/corr/abs-1910-02736} show that this is decidable using Ramsey quantifiers in automatic structures, which leads to high complexity:
For quantifier-free formulas this results in a $\PSPACE$ procedure by constructing an NFA for the negation and 
then evaluating a Ramsey quantifier using \cite{lics22-ramsey}.

Our results settle the complexity, if the relation is given by a quantifier-free formula $\varphi(\bm{x},\bm{y})$: Deciding whether $\varphi$ defines a WQO is $\coNP$-complete.
Suppose $\bm{x}$ and $\bm{y}$ are vectors of $k$ variables and thus $\varphi$ defines a relation on $\Z^k$.
Recall that a relation $R\subseteq\Z^k$ is a WQO iff it is reflexive and transitive and for every infinite sequence $\bm{a}_1,\bm{a}_2,\ldots$, there are $i<j$ with $(\bm{a}_i,\bm{a}_j) \in R$. Thus, $\varphi$ violates the conditions of a WQO if and only if
(1)~$\exists \bm{x}\colon \neg\varphi(\bm{x},\bm{x})$ (reflexivity violation) or (2)~$\exists \bm{x},\bm{y},\bm{z}\colon \varphi(\bm{x},\bm{y})\wedge \varphi(\bm{y},\bm{z})\wedge\neg\varphi(\bm{x},\bm{z})$ (transitivity violation) or (3)~$\ram \bm{x},\bm{y}\colon \neg\varphi(\bm{x},\bm{y})$ (violation of the sequence condition). Thus, we obtain an $\NP$ procedure using \cref{main-presburger} and \cref{prop:sat-presburger}.

Here, $\coNP$-hardness can be shown using a simple ad-hoc proof:
We show that it is $\coNP$-hard to decide whether a given quantifier-free formula $\varphi(\bm{x},\bm{y})$ defines a WQO. In fact, we show that it is $\NP$-hard to test whether $\varphi$ does \emph{not} define a WQO. To this end, we reduce from the problem of satisfiability of quantifier-free Presburger formulas.
Suppose we are given a quantifier-free formula $\psi(\bm{x})$ where $\bm{x}$ is a vector of $k$ variables that range over $\Z$. Consider the following formula $\varphi((x,\bm{x}),(y,\bm{y}))$:
\begin{align*}
	(x=y=0) \vee (x<0 \wedge y<0) \vee (x>0 \wedge y>0) \vee (x<0 \wedge y=0) \vee (x=0 \wedge y>0 \wedge \psi(\bm{y})).
\end{align*}
We claim that $\varphi$ defines a WQO if and only if $\psi$
has no solution. Suppose $\psi$ has no solution. Then we
have $\varphi((x,\bm{x}),(y,\bm{y}))$ if and only if $x$
and $y$ have the same sign (or are both zero) or $x<0$ and
$y=0$. Let $\Z^{k+1}=A\cup B\cup C$, where $A$ are the
elements with negative first component, $B$ are the
elements with first component zero, and $C$ are the
elements with positive first component. Within each set
$A,B,C$, $\varphi$ relates all vectors. Moreover, the
vectors in $A$ are smaller than those in $B$. And those in
$C$ are unrelated to those in $A\cup B$. This is clearly a
WQO.

If $\psi$ has a solution, say $\psi(\bm{a})$, then the
relation defined by $\varphi$ is not transitive: We have
$\varphi((-1,\bm{0}),(0,\bm{0}))$ and
$\varphi((0,\bm{0}),(1,\bm{a}))$, but not
$\varphi((-1,\bm{0}),(1,\bm{a}))$.

\section{Experiments}\label{sec:experiments}
We have implemented a prototype (which can be found at \cite{implementation}) of our Ramsey quantifier elimination algorithms for LIA, LRA, and LIRA 
in Python using the Z3 \cite{MB08} interface Z3Py. 
We have tested it against two sets of micro-benchmarks. 
The first benchmarks contain the following examples, where the dimension $d$ of $\bm x$ and $\bm y$ is a parameter:
\begin{enumerate}[(a)]
\item
$\varphi_\mathrm{half} := \ram \bm x,\bm y \colon 2 \bm y \le \bm x \wedge \bm x \ge \bm t$ for parameter $t \in \Z$
\item
$\varphi_\mathrm{eq\_ex} := \ram \bm x,\bm y \colon \exists \bm z \colon \bm x \ll \bm y \wedge \bm x = \bm z$
\item
$\varphi_\mathrm{eq\_free} := \ram \bm x,\bm y \colon \bm x \ll \bm y \wedge \bm x = \bm z$
\item
$\varphi_\mathrm{dickson} := \ram \bm x,\bm y \colon \bm x \ge \bm 0 \wedge (\bm x > \bm y \vee \bm x \not\le \bm y \wedge \bm y \not\le \bm x)$
where unsatisfiability over $\Z$ proves Dickson's lemma
\item
$\varphi_\mathrm{program} := \ram (\bm x_1,\bm x_2),(\bm y_1, \bm y_2) \colon \bm x_1 \gg \bm 0 \wedge \bm x_2 \gg \bm 0 \wedge
\bm y_1 \ge \bm{0.5} \bm x_1 + \bm{0.5} \wedge \bm y_2 \le \bm x_2 - \lfloor \bm x_1 \rfloor$
that describes an under-approximation of the non-terminating program in \Cref{sec:examples},
where $\bm x_1, \bm y_1$ are vectors of real variables and $\bm x_2, \bm y_2$ are vectors of integer variables. 
\end{enumerate}
Here, $\lfloor \bm v \rfloor$ for a vector $\bm v$ denotes the vector $(\lfloor v_1 \rfloor,\dots,\lfloor v_d \rfloor)$.
Moreover, recall that for numbers $n$ we write $\bm n$ for the vector $(n,\dots,n)$ of appropriate dimension.

The experiments were conducted on an Intel(R) Core(TM) i7-10510U CPU with 16GB of RAM running on Windows~10.
The results are summarized in \Cref{fig:ram-experiments}.
We observe that the number of output variables and atoms linearly depends on the number of input variables and atoms.
In the first three cases, the output formula has ca.\ 5 times as many variables as the input has variables and atoms.
The choice of parameter $t \in \Z$ has no notable effect on the size of the output formula or the running time 
since it only changes constants.
For $\varphi_\mathrm{program}$ our prototype implementation assumes the formula to be decomposed into a Boolean combination of LIA and LRA formulas
whose size is given in the input column of \Cref{fig:ram-experiments}.
Then the running time is dominated by the Z3 satisfiability check due to the large number of variables and atoms in the output.

\begin{table}
\resizebox{.8\columnwidth}{!}{\begin{tabular}{|c|c|c|c c|c c|c|c|c|c|c|}
\hline
\multirow{2}{3em}{formula} & \multirow{2}{2em}{dom} & \multirow{2}{2em}{sat} & \multicolumn{2}{|c|}{input} & \multicolumn{7}{|c|}{output}\\
\cline{4-12}
& & & \#vars & \#atoms & \#vars & \#atoms & $d=1$ & $d=10$ & $d=20$ & $d=50$ & $d=100$\\
\hline\hline
\multirow{2}{3em}{$\varphi_\mathrm{half}$} & $\Z$ & no & \multirow{2}{1em}{$2d$} & \multirow{2}{1em}{$2d$} & $22d$ & $130d$ & $0.04$s & $0.33$s & $0.84$s & $3.35$s & $11.00$s\\
& $\R$ & $t \le 0$ & & & $25d$ & $284d$ & $0.06$s & $0.75$s & $2.10$s & $10.31$s & $38.48$s\\
\hline
\multirow{2}{3em}{$\varphi_\mathrm{eq\_ex}$} & $\Z$ & yes & \multirow{2}{1em}{$3d$} & \multirow{2}{1em}{$2d$} & $31d$ & $166d$ & $0.05$s & $0.67$s & $2.01$s & $10.23$s & $39.44$s\\
& $\R$ & yes & & & $25d$ & $213d$ & $0.05$s & $1.03$s & $3.53$s & $20.01$s & $82.46$s\\
\hline
\multirow{2}{3em}{$\varphi_\mathrm{eq\_free}$} & $\Z$ & no & \multirow{2}{1em}{$3d$} & \multirow{2}{1em}{$2d$} & $28d$ & $162d$ & $0.05$s & $0.44$s & $1.12$s & $4.73$s & $16.11$s\\
& $\R$ & no & & & $20d$ & $209d$ & $0.08$s & $0.54$s & $1.64$s & $8.18$s & $31.03$s\\
\hline
\multirow{2}{3em}{$\varphi_\mathrm{dickson}$} & $\Z$ & no & \multirow{2}{1em}{$2d$} & \multirow{2}{1em}{$5d$} & $37d$ & $226d$ & $0.06$s & $0.58$s & $1.52$s & $6.33$s & $21.60$s\\
& $\R$ & yes & & & $40d$ & $482d$ & $0.08$s & $1.17$s & $4.48$s & $17.18$s & $66.46$s\\
\hline
\multirow{1}{3em}{$\varphi_\mathrm{program}$} & $\R,\Z$ & yes & \multirow{1}{1em}{$6d$} & \multirow{1}{1em}{$14d$} & $426d+1$ & $3858d+4$ & $0.84$s & $68.28$s & $445.89$s & $>500$s & $>500$s\\
\hline
\end{tabular}}
\caption{Experiments for the elimination of the Ramsey quantifier with a 500 seconds timeout.}
\label{fig:ram-experiments}
\end{table}

For the second benchmarks we used our elimination procedure 
to check monadic decomposability, as described in \Cref{sec:applications}, of the following formulas:
\begin{enumerate}[(a)]
\item $\varphi_\mathrm{imp} := \bigwedge_{i=1}^d x_i \ge 0 \to x_i + y_i \ge k \wedge y_i \ge 0$ for parameter $k \in \N$
\item $\varphi_\mathrm{diagonal} := \bm 0 \le \bm x \le \bm k \wedge x_1 = \dots = x_d$ for parameter $k \in \N$
\item $\varphi_\mathrm{cubes2d} := x_1 + x_2 \le k \wedge \bigwedge_{i=1}^k i \le x_1 \le i+2 \wedge i \le x_2 \le i+2$ 
where parameter $k \in \N$ is the number of cubes
\item $\varphi_\mathrm{cubes10} := \bigwedge_{i=1}^{10} \bm i \le \bm x \le \bm{i+2}$
\item $\varphi_\mathrm{mixed} := \bm x = \lfloor \bm y \rfloor \wedge \bm 0 \le \bm y \le \bm k$
over LIRA with parameter $k \in \N$ 
\end{enumerate}
The results are shown in \Cref{fig:mondec-experiments} where either the dimension $d$ or parameter $k$ is varied.
The size of the input refers to the formula $\delta_{\varphi,(x_1,\dots,x_{i-1},x_{i+1},\dots,x_d)}$ for $i = 1$
that is defined similarly to $\delta_{\varphi,x}$ in \Cref{sec:applications} but uses only one existentially quantified variable.
This has the advantage that the algorithm only has to eliminate one existential variable before eliminating the Ramsey quantifier.
For the output we measure the size of the first formula given to Z3, i.e., 
$\delta_{\varphi,(x_2,\dots,x_d)}$ after elimination of the Ramsey quantifier.
We observe that if $n$ is the number of input variables plus atoms, on these instances the number of output variables can be estimated by 
$5 \cdot n$ over $\Z$ and $10 \cdot n$ over $\R$.
Note that not only is the formula $\delta_{\varphi,(x_2,\dots,x_d)}$ (the input to the elimination procedure) larger than $\varphi$,
where monadic decomposability is checked on, we also have to consider all of the $\delta_{\varphi,(x_1,\dots,x_{i-1},x_{i+1},\dots,x_d)}$ in case $\varphi$
is monadically decomposable.
This explains the slowdown compared to \Cref{fig:ram-experiments}.

For $\varphi_\mathrm{imp}$ and $\varphi_\mathrm{diagonal}$ we observe that, despite the larger output formula, 
over $\R$ the algorithm terminates significantly faster than over $\Z$
since it only needs to construct $\delta_{\varphi,(x_2,\dots,x_d)}$ to detect that $\varphi$ is not monadically decomposable.
The first four examples are taken from \cite{MSL21} where the authors compare their tool to $\mathrm{mondec}_1$ from \cite{DBLP:journals/jacm/VeanesBNB17} 
that \emph{computes} a monadic decomposition if one exists.
We observe that on these instances our \emph{decision} algorithm is significantly faster than $\mathrm{mondec_1}$, 
especially for $\varphi_\mathrm{imp}$ and $\varphi_\mathrm{diagonal}$ when only the parameter $k$ is varied
(and $d = 1$ resp. $d = 2$ as in \cite{MSL21}).
The reason for this is that $\mathrm{mondec_1}$ computes the monadic decomposition whose size grows exponentially in the encoding of $k$,
whereas in our approach, where we only decide if a decomposition exists, $k$ only changes a constant in the formulas where the Ramsey quantifier is eliminated.
Therefore, changing $k$ in the two examples (and also in $\varphi_\mathrm{mixed}$) does not have any notable effect on the running time in \Cref{fig:mondec-experiments}.
In this case, our algorithm is also faster than the one developed in \cite{MSL21} that outputs the decomposition in form of cubes.
Since both algorithms in \cite{DBLP:journals/jacm/VeanesBNB17} and \cite{MSL21} only terminate if the input formula is monadically decomposable,
our algorithm is the only one that terminates on $\varphi_\mathrm{imp}$, $\varphi_\mathrm{diagonal}$, and $\varphi_\mathrm{cubes2d}$ over $\R$
and can therefore be used as a termination check in the other algorithms.
Finally, note that the increase of the running time for $\varphi_\mathrm{cubes2d}$, $\varphi_\mathrm{cubes10}$ over $\R$ and $\varphi_\mathrm{mixed}$
is due to the large number of atoms in the output, which is problematic not only for the elimination procedure but especially for the satisfiability check with Z3.
We observe that for large instances, the running time is dominated by the satisfiability check.

\begin{table}
\resizebox{.8\columnwidth}{!}{\begin{tabular}{|c|c|c|c c|c c|c|c|c|c|}
\hline
\multirow{2}{3em}{formula} & \multirow{2}{2em}{dom} & \multirow{2}{3em}{mondec} & \multicolumn{2}{|c|}{input} & \multicolumn{6}{|c|}{output}\\
\cline{4-11}
& & & \#vars & \#atoms & \#vars & \#atoms & \multicolumn{4}{|c|}{}\\
\hline\hline
\multirow{3}{3em}{$\varphi_\mathrm{imp}$} & & & & & & & $d=1$ & $d=5$ & $d=10$ & $d=20$\\
\cline{8-11}
& $\Z$ & yes & \multirow{2}{3em}{$4d-1$} & \multirow{2}{3em}{$12d$} & $84d-8$ & $516d-62$ & $0.12$s & $6.19$s & $34.11$s & $224.53$s\\
& $\R$ & no & & & $136d-7$ & $1678d-130$ & $0.32$s & $2.27$s & $6.69$s & $19.01$s\\
\hline
\multirow{3}{3em}{$\varphi_\mathrm{diagonal}$} & & & & & & & $d=2$ & $d=10$ & $d=20$ & $d=30$\\
\cline{8-11}
& $\Z$ & yes & \multirow{2}{3em}{$2d-1$} & \multirow{2}{3em}{$4d+4$} & $52d-8$ & $322d-62$ & $0.15$s & $8.20$s & $46.48$s & $151.42$s\\
& $\R$ & no & & & $35d+59$ & $416d+716$ & $0.32$s & $1.33$s & $3.82$s & $7.37$s\\
\hline
\multirow{3}{3em}{$\varphi_\mathrm{cubes2d}$} & & & & & & & $k=50$ & $k=100$ & $k=150$ & $k=250$\\
\cline{8-11}
& $\Z$ & yes & \multirow{2}{3em}{$3$} & \multirow{2}{3em}{$16k+4$} & $80k+36$ & $512k+198$ & $7.74$s & $18.77$s & $37.73$s & $109.17$s\\
& $\R$ & no & & & $176k+63$ & $2256k+702$ & $210.12$s & $>500$s & $>500$s & $>500$s\\
\hline
\multirow{3}{3em}{$\varphi_\mathrm{cubes10}$} & & & & & & & $d=2$ & $d=10$ & $d=15$ & $d=20$\\
\cline{8-11}
& $\Z$ & yes & \multirow{2}{3em}{$2d-1$} & \multirow{2}{3em}{$80d$} & $412d-8$ & $2626d-62$ & $1.18$s & $66.40$s & $231.67$s & $482.39$s\\
& $\R$ & yes & & & $893d-7$ & $11414d-130$ & $5.83$s & $>500$s & $>500$s & $>500$s\\
\hline
\multirow{2}{3em}{$\varphi_\mathrm{mixed}$} & & & & & & & $d=1$ & $d=2$ & $d=3$ & $d=4$\\
\cline{8-11}
& $\R,\Z$ & yes & \multirow{1}{3em}{$6d-1$} & \multirow{1}{3em}{$28d$} & $842d-28$ & $7710d-198$ & $3.67$s & $42.84$s & $192.76$s & $>500$s\\
\hline
\end{tabular}}
\caption{Experiments for monadic decomposability with a 500 seconds timeout.}
\label{fig:mondec-experiments}
\end{table}
 
\section{Conclusion and Future Work}
\label{sec:conc}

We have given efficient algorithms for removing Ramsey quantifiers from the
theories of Linear Integer Arithmetic (LIA), Linear Real Arithmetic (LRA), and
Linear Integer Real Arithmetic (LIRA). The algorithm runs in polynomial time and
is guaranteed to produce formulas of linear size. We have shown that this leads
to applications in proving termination/non-termination of programs, as well as
checking variable dependencies (a.k.a. monadic decomposability) in a given
formula. 

We mention several future research avenues. First, combined with existing 
results on computation of reachability relations 
\cite{BFLP08,BFLS05,BH06,BLW03,TORMC}, we obtain fully-automatic methods for
proving termination/non-termination. Recent software verification frameworks,
however, rely on 
\emph{Constraint Horn Clauses} (CHC), which extend SMT with recursive predicate, e.g., see \cite{BMR12,BGMR15}. 
To extend the framework for proving termination, one typically extends CHC with
ad-hoc well-foundedness conditions \cite{BPR13}. Our results suggest that we
can instead extend CHC with Ramsey quantifiers, and develop synthesis algorithms
for the framework. We leave this for future work. Second, we also mention that
eliminability of Ramsey quantifiers from other theories (e.g. non-linear real
arithmetics and EUF) remains open, which we also leave for future work.
 
\begin{acks}
    We thank anonymous reviewers, Arie Gurfinkel, and Philipp R\"{u}mmer
    for their helpful comments. Moreover, we are grateful to Christoph Haase for discussions on existing quantifier elimination techniques.
Funded by the European Union (\grantsponsor{501100000781}{ERC}{http://dx.doi.org/10.13039/501100000781}, LASD, \grantnum[https://doi.org/10.3030/101089343]{501100000781}{101089343}, and FINABIS, \grantnum[https://doi.org/10.3030/101077902]{501100000781}{101077902})\marginpar{\includegraphics[width=1.3cm]{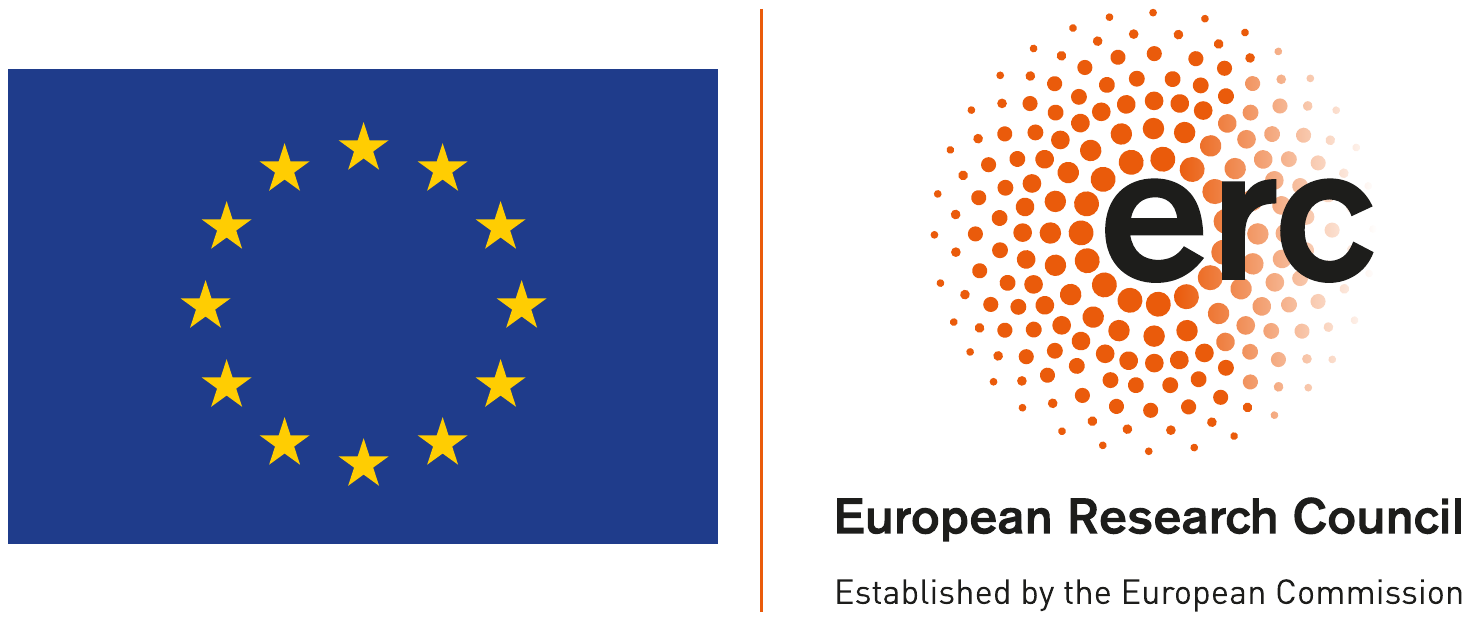}}. Views and opinions expressed are however those of the authors only and do not necessarily reflect those of the European Union or the European Research Council Executive Agency. Neither the European Union nor the granting authority can be held responsible for them.
\end{acks}

\bibliographystyle{ACM-Reference-Format}
\bibliography{bib}


\begin{thebibliography}{76}


\ifx \showCODEN    \undefined \def \showCODEN     #1{\unskip}     \fi
\ifx \showDOI      \undefined \def \showDOI       #1{#1}\fi
\ifx \showISBNx    \undefined \def \showISBNx     #1{\unskip}     \fi
\ifx \showISBNxiii \undefined \def \showISBNxiii  #1{\unskip}     \fi
\ifx \showISSN     \undefined \def \showISSN      #1{\unskip}     \fi
\ifx \showLCCN     \undefined \def \showLCCN      #1{\unskip}     \fi
\ifx \shownote     \undefined \def \shownote      #1{#1}          \fi
\ifx \showarticletitle \undefined \def \showarticletitle #1{#1}   \fi
\ifx \showURL      \undefined \def \showURL       {\relax}        \fi
\providecommand\bibfield[2]{#2}
\providecommand\bibinfo[2]{#2}
\providecommand\natexlab[1]{#1}
\providecommand\showeprint[2][]{arXiv:#2}

\bibitem[Abdulla et~al\mbox{.}(2012)]%
        {DBLP:conf/lics/AbdullaAS12}
\bibfield{author}{\bibinfo{person}{Parosh~Aziz Abdulla},
  \bibinfo{person}{Mohamed~Faouzi Atig}, {and} \bibinfo{person}{Jari Stenman}.}
  \bibinfo{year}{2012}\natexlab{}.
\newblock \showarticletitle{Dense-Timed Pushdown Automata}. In
  \bibinfo{booktitle}{\emph{Proceedings of the 27th Annual {IEEE} Symposium on
  Logic in Computer Science, {LICS} 2012, Dubrovnik, Croatia, June 25-28,
  2012}}. \bibinfo{publisher}{{IEEE} Computer Society},
  \bibinfo{pages}{35--44}.
\newblock
\urldef\tempurl%
\url{https://doi.org/10.1109/LICS.2012.15}
\showDOI{\tempurl}


\bibitem[Abdulla et~al\mbox{.}(1996)]%
        {DBLP:conf/lics/AbdullaCJT96}
\bibfield{author}{\bibinfo{person}{Parosh~Aziz Abdulla},
  \bibinfo{person}{Karlis Cerans}, \bibinfo{person}{Bengt Jonsson}, {and}
  \bibinfo{person}{Yih{-}Kuen Tsay}.} \bibinfo{year}{1996}\natexlab{}.
\newblock \showarticletitle{General Decidability Theorems for Infinite-State
  Systems}. In \bibinfo{booktitle}{\emph{Proceedings, 11th Annual {IEEE}
  Symposium on Logic in Computer Science, New Brunswick, New Jersey, USA, July
  27-30, 1996}}. \bibinfo{publisher}{{IEEE} Computer Society},
  \bibinfo{pages}{313--321}.
\newblock
\urldef\tempurl%
\url{https://doi.org/10.1109/LICS.1996.561359}
\showDOI{\tempurl}


\bibitem[Alur and Dill(1994)]%
        {alur1994theory}
\bibfield{author}{\bibinfo{person}{Rajeev Alur} {and} \bibinfo{person}{David~L
  Dill}.} \bibinfo{year}{1994}\natexlab{}.
\newblock \showarticletitle{A theory of timed automata}.
\newblock \bibinfo{journal}{\emph{Theoretical computer science}}
  \bibinfo{volume}{126}, \bibinfo{number}{2} (\bibinfo{year}{1994}),
  \bibinfo{pages}{183--235}.
\newblock


\bibitem[Barcel{\'{o}} et~al\mbox{.}(2019)]%
        {BHLLN19}
\bibfield{author}{\bibinfo{person}{Pablo Barcel{\'{o}}},
  \bibinfo{person}{Chih{-}Duo Hong}, \bibinfo{person}{Xuan~Bach Le},
  \bibinfo{person}{Anthony~W. Lin}, {and} \bibinfo{person}{Reino Niskanen}.}
  \bibinfo{year}{2019}\natexlab{}.
\newblock \showarticletitle{Monadic Decomposability of Regular Relations}. In
  \bibinfo{booktitle}{\emph{46th International Colloquium on Automata,
  Languages, and Programming, {ICALP} 2019, July 9-12, 2019, Patras, Greece}}
  \emph{(\bibinfo{series}{LIPIcs}, Vol.~\bibinfo{volume}{132})},
  \bibfield{editor}{\bibinfo{person}{Christel Baier}, \bibinfo{person}{Ioannis
  Chatzigiannakis}, \bibinfo{person}{Paola Flocchini}, {and}
  \bibinfo{person}{Stefano Leonardi}} (Eds.). \bibinfo{publisher}{Schloss
  Dagstuhl - Leibniz-Zentrum f{\"{u}}r Informatik},
  \bibinfo{pages}{103:1--103:14}.
\newblock
\urldef\tempurl%
\url{https://doi.org/10.4230/LIPIcs.ICALP.2019.103}
\showDOI{\tempurl}


\bibitem[Bardin et~al\mbox{.}(2008)]%
        {BFLP08}
\bibfield{author}{\bibinfo{person}{S{\'{e}}bastien Bardin},
  \bibinfo{person}{Alain Finkel}, \bibinfo{person}{J{\'{e}}r{\^{o}}me Leroux},
  {and} \bibinfo{person}{Laure Petrucci}.} \bibinfo{year}{2008}\natexlab{}.
\newblock \showarticletitle{{FAST:} acceleration from theory to practice}.
\newblock \bibinfo{journal}{\emph{Int. J. Softw. Tools Technol. Transf.}}
  \bibinfo{volume}{10}, \bibinfo{number}{5} (\bibinfo{year}{2008}),
  \bibinfo{pages}{401--424}.
\newblock
\urldef\tempurl%
\url{https://doi.org/10.1007/s10009-008-0064-3}
\showDOI{\tempurl}


\bibitem[Bardin et~al\mbox{.}(2005)]%
        {BFLS05}
\bibfield{author}{\bibinfo{person}{S{\'{e}}bastien Bardin},
  \bibinfo{person}{Alain Finkel}, \bibinfo{person}{J{\'{e}}r{\^{o}}me Leroux},
  {and} \bibinfo{person}{Philippe Schnoebelen}.}
  \bibinfo{year}{2005}\natexlab{}.
\newblock \showarticletitle{Flat Acceleration in Symbolic Model Checking}. In
  \bibinfo{booktitle}{\emph{Automated Technology for Verification and Analysis,
  Third International Symposium, {ATVA} 2005, Taipei, Taiwan, October 4-7,
  2005, Proceedings}} \emph{(\bibinfo{series}{Lecture Notes in Computer
  Science}, Vol.~\bibinfo{volume}{3707})},
  \bibfield{editor}{\bibinfo{person}{Doron~A. Peled} {and}
  \bibinfo{person}{Yih{-}Kuen Tsay}} (Eds.). \bibinfo{publisher}{Springer},
  \bibinfo{pages}{474--488}.
\newblock
\urldef\tempurl%
\url{https://doi.org/10.1007/11562948\_35}
\showDOI{\tempurl}


\bibitem[Barwise and Feferman(1985)]%
        {BF85}
\bibfield{editor}{\bibinfo{person}{J. Barwise} {and} \bibinfo{person}{S.
  Feferman}} (Eds.). \bibinfo{year}{1985}\natexlab{}.
\newblock \bibinfo{booktitle}{\emph{Model-Theoretic Logics}}.
  \bibinfo{series}{Perspectives in Logic}, Vol.~\bibinfo{volume}{8}.
  \bibinfo{publisher}{Association for Symbolic Logic}.
\newblock


\bibitem[Benedikt et~al\mbox{.}(2003)]%
        {BLSS03}
\bibfield{author}{\bibinfo{person}{Michael Benedikt}, \bibinfo{person}{Leonid
  Libkin}, \bibinfo{person}{Thomas Schwentick}, {and} \bibinfo{person}{Luc
  Segoufin}.} \bibinfo{year}{2003}\natexlab{}.
\newblock \showarticletitle{Definable relations and first-order query languages
  over strings}.
\newblock \bibinfo{journal}{\emph{J. {ACM}}} \bibinfo{volume}{50},
  \bibinfo{number}{5} (\bibinfo{year}{2003}), \bibinfo{pages}{694--751}.
\newblock
\urldef\tempurl%
\url{https://doi.org/10.1145/876638.876642}
\showDOI{\tempurl}


\bibitem[Bergstr{\"{a}}{\ss}er and Ganardi(2023)]%
        {DBLP:journals/corr/abs-2304-11034}
\bibfield{author}{\bibinfo{person}{Pascal Bergstr{\"{a}}{\ss}er} {and}
  \bibinfo{person}{Moses Ganardi}.} \bibinfo{year}{2023}\natexlab{}.
\newblock \showarticletitle{Revisiting Membership Problems in Subclasses of
  Rational Relations}.
\newblock \bibinfo{journal}{\emph{CoRR}}  \bibinfo{volume}{abs/2304.11034}
  (\bibinfo{year}{2023}).
\newblock
\urldef\tempurl%
\url{https://doi.org/10.48550/arXiv.2304.11034}
\showDOI{\tempurl}
\showeprint[arXiv]{2304.11034}


\bibitem[Bergstr{\"{a}}{\ss}er et~al\mbox{.}(2022)]%
        {lics22-ramsey}
\bibfield{author}{\bibinfo{person}{Pascal Bergstr{\"{a}}{\ss}er},
  \bibinfo{person}{Moses Ganardi}, \bibinfo{person}{Anthony~W. Lin}, {and}
  \bibinfo{person}{Georg Zetzsche}.} \bibinfo{year}{2022}\natexlab{}.
\newblock \showarticletitle{Ramsey Quantifiers over Automatic Structures:
  Complexity and Applications to Verification}. In
  \bibinfo{booktitle}{\emph{{LICS} '22: 37th Annual {ACM/IEEE} Symposium on
  Logic in Computer Science, Haifa, Israel, August 2 - 5, 2022}}.
  \bibinfo{pages}{28:1--28:14}.
\newblock
\urldef\tempurl%
\url{https://doi.org/10.1145/3531130.3533346}
\showDOI{\tempurl}


\bibitem[Bergstr{\"{a}}{\ss}er et~al\mbox{.}(2023)]%
        {implementation}
\bibfield{author}{\bibinfo{person}{Pascal Bergstr{\"{a}}{\ss}er},
  \bibinfo{person}{Moses Ganardi}, \bibinfo{person}{Anthony~W. Lin}, {and}
  \bibinfo{person}{Georg Zetzsche}.} \bibinfo{year}{2023}\natexlab{}.
\newblock \bibinfo{title}{Ramsey Quantifiers in Linear Arithmetics - Artifact}.
\newblock \bibinfo{howpublished}{\url{https://doi.org/10.5281/zenodo.8422415}}.
\newblock


\bibitem[Beyene et~al\mbox{.}(2013)]%
        {BPR13}
\bibfield{author}{\bibinfo{person}{Tewodros~A. Beyene},
  \bibinfo{person}{Corneliu Popeea}, {and} \bibinfo{person}{Andrey
  Rybalchenko}.} \bibinfo{year}{2013}\natexlab{}.
\newblock \showarticletitle{Solving Existentially Quantified Horn Clauses}. In
  \bibinfo{booktitle}{\emph{Computer Aided Verification - 25th International
  Conference, {CAV} 2013, Saint Petersburg, Russia, July 13-19, 2013.
  Proceedings}} \emph{(\bibinfo{series}{Lecture Notes in Computer Science},
  Vol.~\bibinfo{volume}{8044})}, \bibfield{editor}{\bibinfo{person}{Natasha
  Sharygina} {and} \bibinfo{person}{Helmut Veith}} (Eds.).
  \bibinfo{publisher}{Springer}, \bibinfo{pages}{869--882}.
\newblock
\urldef\tempurl%
\url{https://doi.org/10.1007/978-3-642-39799-8\_61}
\showDOI{\tempurl}


\bibitem[Bj{\o}rner et~al\mbox{.}(2015)]%
        {BGMR15}
\bibfield{author}{\bibinfo{person}{Nikolaj~S. Bj{\o}rner},
  \bibinfo{person}{Arie Gurfinkel}, \bibinfo{person}{Kenneth~L. McMillan},
  {and} \bibinfo{person}{Andrey Rybalchenko}.} \bibinfo{year}{2015}\natexlab{}.
\newblock \showarticletitle{Horn Clause Solvers for Program Verification}. In
  \bibinfo{booktitle}{\emph{Fields of Logic and Computation {II} - Essays
  Dedicated to Yuri Gurevich on the Occasion of His 75th Birthday}}
  \emph{(\bibinfo{series}{Lecture Notes in Computer Science},
  Vol.~\bibinfo{volume}{9300})}, \bibfield{editor}{\bibinfo{person}{Lev~D.
  Beklemishev}, \bibinfo{person}{Andreas Blass}, \bibinfo{person}{Nachum
  Dershowitz}, \bibinfo{person}{Bernd Finkbeiner}, {and}
  \bibinfo{person}{Wolfram Schulte}} (Eds.). \bibinfo{publisher}{Springer},
  \bibinfo{pages}{24--51}.
\newblock
\urldef\tempurl%
\url{https://doi.org/10.1007/978-3-319-23534-9\_2}
\showDOI{\tempurl}


\bibitem[Bj{\o}rner et~al\mbox{.}(2012)]%
        {BMR12}
\bibfield{author}{\bibinfo{person}{Nikolaj~S. Bj{\o}rner},
  \bibinfo{person}{Kenneth~L. McMillan}, {and} \bibinfo{person}{Andrey
  Rybalchenko}.} \bibinfo{year}{2012}\natexlab{}.
\newblock \showarticletitle{Program Verification as Satisfiability Modulo
  Theories}. In \bibinfo{booktitle}{\emph{10th International Workshop on
  Satisfiability Modulo Theories, {SMT} 2012, Manchester, UK, June 30 - July 1,
  2012}} \emph{(\bibinfo{series}{EPiC Series in Computing},
  Vol.~\bibinfo{volume}{20})}, \bibfield{editor}{\bibinfo{person}{Pascal
  Fontaine} {and} \bibinfo{person}{Amit Goel}} (Eds.).
  \bibinfo{publisher}{EasyChair}, \bibinfo{pages}{3--11}.
\newblock
\urldef\tempurl%
\url{https://doi.org/10.29007/1l7f}
\showDOI{\tempurl}


\bibitem[Blondin et~al\mbox{.}(2016)]%
        {DBLP:conf/tacas/BlondinFHH16}
\bibfield{author}{\bibinfo{person}{Michael Blondin}, \bibinfo{person}{Alain
  Finkel}, \bibinfo{person}{Christoph Haase}, {and} \bibinfo{person}{Serge
  Haddad}.} \bibinfo{year}{2016}\natexlab{}.
\newblock \showarticletitle{Approaching the Coverability Problem Continuously}.
  In \bibinfo{booktitle}{\emph{Proc. of the 22nd International Conference on
  Tools and Algorithms for the Construction and Analysis of Systems (TACAS
  2016)}} \emph{(\bibinfo{series}{LNCS}, Vol.~\bibinfo{volume}{9636})}.
  \bibinfo{publisher}{Springer}, \bibinfo{pages}{480--496}.
\newblock
\urldef\tempurl%
\url{https://doi.org/10.1007/978-3-662-49674-9\_28}
\showDOI{\tempurl}


\bibitem[Blondin and Haase(2017)]%
        {DBLP:conf/lics/BlondinH17}
\bibfield{author}{\bibinfo{person}{Michael Blondin} {and}
  \bibinfo{person}{Christoph Haase}.} \bibinfo{year}{2017}\natexlab{}.
\newblock \showarticletitle{Logics for continuous reachability in Petri nets
  and vector addition systems with states}. In \bibinfo{booktitle}{\emph{32nd
  Annual {ACM/IEEE} Symposium on Logic in Computer Science, {LICS} 2017,
  Reykjavik, Iceland, June 20-23, 2017}}. \bibinfo{publisher}{{IEEE} Computer
  Society}, \bibinfo{pages}{1--12}.
\newblock
\urldef\tempurl%
\url{https://doi.org/10.1109/LICS.2017.8005068}
\showDOI{\tempurl}


\bibitem[Blumensath and Gr{\"{a}}del(2000)]%
        {BG00}
\bibfield{author}{\bibinfo{person}{Achim Blumensath} {and}
  \bibinfo{person}{Erich Gr{\"{a}}del}.} \bibinfo{year}{2000}\natexlab{}.
\newblock \showarticletitle{Automatic Structures}. In
  \bibinfo{booktitle}{\emph{15th Annual {IEEE} Symposium on Logic in Computer
  Science, Santa Barbara, California, USA, June 26-29, 2000}}.
  \bibinfo{publisher}{{IEEE} Computer Society}, \bibinfo{pages}{51--62}.
\newblock
\urldef\tempurl%
\url{https://doi.org/10.1109/LICS.2000.855755}
\showDOI{\tempurl}


\bibitem[Boigelot and Herbreteau(2006)]%
        {BH06}
\bibfield{author}{\bibinfo{person}{Bernard Boigelot} {and}
  \bibinfo{person}{Fr{\'{e}}d{\'{e}}ric Herbreteau}.}
  \bibinfo{year}{2006}\natexlab{}.
\newblock \showarticletitle{The Power of Hybrid Acceleration}. In
  \bibinfo{booktitle}{\emph{Computer Aided Verification, 18th International
  Conference, {CAV} 2006, Seattle, WA, USA, August 17-20, 2006, Proceedings}}
  \emph{(\bibinfo{series}{Lecture Notes in Computer Science},
  Vol.~\bibinfo{volume}{4144})}, \bibfield{editor}{\bibinfo{person}{Thomas
  Ball} {and} \bibinfo{person}{Robert~B. Jones}} (Eds.).
  \bibinfo{publisher}{Springer}, \bibinfo{pages}{438--451}.
\newblock
\urldef\tempurl%
\url{https://doi.org/10.1007/11817963\_40}
\showDOI{\tempurl}


\bibitem[Boigelot et~al\mbox{.}(2003)]%
        {BLW03}
\bibfield{author}{\bibinfo{person}{Bernard Boigelot}, \bibinfo{person}{Axel
  Legay}, {and} \bibinfo{person}{Pierre Wolper}.}
  \bibinfo{year}{2003}\natexlab{}.
\newblock \showarticletitle{Iterating Transducers in the Large (Extended
  Abstract)}. In \bibinfo{booktitle}{\emph{Computer Aided Verification, 15th
  International Conference, {CAV} 2003, Boulder, CO, USA, July 8-12, 2003,
  Proceedings}} \emph{(\bibinfo{series}{Lecture Notes in Computer Science},
  Vol.~\bibinfo{volume}{2725})}, \bibfield{editor}{\bibinfo{person}{Warren
  A.~Hunt Jr.} {and} \bibinfo{person}{Fabio Somenzi}} (Eds.).
  \bibinfo{publisher}{Springer}, \bibinfo{pages}{223--235}.
\newblock
\urldef\tempurl%
\url{https://doi.org/10.1007/978-3-540-45069-6\_24}
\showDOI{\tempurl}


\bibitem[Borosh and Treybig(1976)]%
        {borosh1976bounds}
\bibfield{author}{\bibinfo{person}{Itshak Borosh} {and}
  \bibinfo{person}{Leon~Bruce Treybig}.} \bibinfo{year}{1976}\natexlab{}.
\newblock \showarticletitle{Bounds on positive integral solutions of linear
  Diophantine equations}.
\newblock \bibinfo{journal}{\emph{Proc. Amer. Math. Soc.}}
  \bibinfo{volume}{55}, \bibinfo{number}{2} (\bibinfo{year}{1976}),
  \bibinfo{pages}{299--304}.
\newblock


\bibitem[Bouajjani et~al\mbox{.}(2011)]%
        {lists-counters}
\bibfield{author}{\bibinfo{person}{Ahmed Bouajjani}, \bibinfo{person}{Marius
  Bozga}, \bibinfo{person}{Peter Habermehl}, \bibinfo{person}{Radu Iosif},
  \bibinfo{person}{Pierre Moro}, {and} \bibinfo{person}{Tom{\'{a}}s Vojnar}.}
  \bibinfo{year}{2011}\natexlab{}.
\newblock \showarticletitle{Programs with lists are counter automata}.
\newblock \bibinfo{journal}{\emph{Formal Methods Syst. Des.}}
  \bibinfo{volume}{38}, \bibinfo{number}{2} (\bibinfo{year}{2011}),
  \bibinfo{pages}{158--192}.
\newblock
\urldef\tempurl%
\url{https://doi.org/10.1007/s10703-011-0111-7}
\showDOI{\tempurl}


\bibitem[Bouajjani et~al\mbox{.}(1994)]%
        {DBLP:conf/hybrid/BouajjaniER94}
\bibfield{author}{\bibinfo{person}{Ahmed Bouajjani}, \bibinfo{person}{Rachid
  Echahed}, {and} \bibinfo{person}{Riadh Robbana}.}
  \bibinfo{year}{1994}\natexlab{}.
\newblock \showarticletitle{On the Automatic Verification of Systems with
  Continuous Variables and Unbounded Discrete Data Structures}. In
  \bibinfo{booktitle}{\emph{Hybrid Systems II, Proceedings of the Third
  International Workshop on Hybrid Systems, Ithaca, NY, USA, October 1994}}
  \emph{(\bibinfo{series}{Lecture Notes in Computer Science},
  Vol.~\bibinfo{volume}{999})}, \bibfield{editor}{\bibinfo{person}{Panos~J.
  Antsaklis}, \bibinfo{person}{Wolf Kohn}, \bibinfo{person}{Anil Nerode}, {and}
  \bibinfo{person}{Shankar Sastry}} (Eds.). \bibinfo{publisher}{Springer},
  \bibinfo{pages}{64--85}.
\newblock
\urldef\tempurl%
\url{https://doi.org/10.1007/3-540-60472-3\_4}
\showDOI{\tempurl}


\bibitem[Chang and Keisler(1990)]%
        {model-theory-book}
\bibfield{author}{\bibinfo{person}{C.~C. Chang} {and} \bibinfo{person}{H.~J.
  Keisler}.} \bibinfo{year}{1990}\natexlab{}.
\newblock \bibinfo{booktitle}{\emph{Model Theory}}.
\newblock \bibinfo{publisher}{Elsevier}.
\newblock


\bibitem[Clemente and Lasota(2018)]%
        {DBLP:conf/icalp/ClementeL18}
\bibfield{author}{\bibinfo{person}{Lorenzo Clemente} {and}
  \bibinfo{person}{Slawomir Lasota}.} \bibinfo{year}{2018}\natexlab{}.
\newblock \showarticletitle{Binary Reachability of Timed Pushdown Automata via
  Quantifier Elimination and Cyclic Order Atoms}. In
  \bibinfo{booktitle}{\emph{45th International Colloquium on Automata,
  Languages, and Programming, {ICALP} 2018, July 9-13, 2018, Prague, Czech
  Republic}} \emph{(\bibinfo{series}{LIPIcs}, Vol.~\bibinfo{volume}{107})},
  \bibfield{editor}{\bibinfo{person}{Ioannis Chatzigiannakis},
  \bibinfo{person}{Christos Kaklamanis}, \bibinfo{person}{D{\'{a}}niel Marx},
  {and} \bibinfo{person}{Donald Sannella}} (Eds.). \bibinfo{publisher}{Schloss
  Dagstuhl - Leibniz-Zentrum f{\"{u}}r Informatik},
  \bibinfo{pages}{118:1--118:14}.
\newblock
\urldef\tempurl%
\url{https://doi.org/10.4230/LIPIcs.ICALP.2018.118}
\showDOI{\tempurl}


\bibitem[Comon and Jurski(1999)]%
        {DBLP:conf/concur/ComonJ99}
\bibfield{author}{\bibinfo{person}{Hubert Comon} {and} \bibinfo{person}{Yan
  Jurski}.} \bibinfo{year}{1999}\natexlab{}.
\newblock \showarticletitle{Timed Automata and the Theory of Real Numbers}. In
  \bibinfo{booktitle}{\emph{{CONCUR} '99: Concurrency Theory, 10th
  International Conference, Eindhoven, The Netherlands, August 24-27, 1999,
  Proceedings}} \emph{(\bibinfo{series}{Lecture Notes in Computer Science},
  Vol.~\bibinfo{volume}{1664})}, \bibfield{editor}{\bibinfo{person}{Jos C.~M.
  Baeten} {and} \bibinfo{person}{Sjouke Mauw}} (Eds.).
  \bibinfo{publisher}{Springer}, \bibinfo{pages}{242--257}.
\newblock
\urldef\tempurl%
\url{https://doi.org/10.1007/3-540-48320-9\_18}
\showDOI{\tempurl}


\bibitem[Cook et~al\mbox{.}(2011)]%
        {CPR11}
\bibfield{author}{\bibinfo{person}{Byron Cook}, \bibinfo{person}{Andreas
  Podelski}, {and} \bibinfo{person}{Andrey Rybalchenko}.}
  \bibinfo{year}{2011}\natexlab{}.
\newblock \showarticletitle{Proving program termination}.
\newblock \bibinfo{journal}{\emph{Commun. {ACM}}} \bibinfo{volume}{54},
  \bibinfo{number}{5} (\bibinfo{year}{2011}), \bibinfo{pages}{88--98}.
\newblock
\urldef\tempurl%
\url{https://doi.org/10.1145/1941487.1941509}
\showDOI{\tempurl}


\bibitem[Czerwinski and Orlikowski(2021)]%
        {DBLP:conf/focs/CzerwinskiO21}
\bibfield{author}{\bibinfo{person}{Wojciech Czerwinski} {and}
  \bibinfo{person}{Lukasz Orlikowski}.} \bibinfo{year}{2021}\natexlab{}.
\newblock \showarticletitle{Reachability in Vector Addition Systems is
  Ackermann-complete}. In \bibinfo{booktitle}{\emph{62nd {IEEE} Annual
  Symposium on Foundations of Computer Science, {FOCS} 2021, Denver, CO, USA,
  February 7-10, 2022}}. \bibinfo{publisher}{{IEEE}},
  \bibinfo{pages}{1229--1240}.
\newblock
\urldef\tempurl%
\url{https://doi.org/10.1109/FOCS52979.2021.00120}
\showDOI{\tempurl}


\bibitem[Dang(2001)]%
        {D01}
\bibfield{author}{\bibinfo{person}{Zhe Dang}.} \bibinfo{year}{2001}\natexlab{}.
\newblock \showarticletitle{Binary Reachability Analysis of Pushdown Timed
  Automata with Dense Clocks}. In \bibinfo{booktitle}{\emph{Computer Aided
  Verification, 13th International Conference, {CAV} 2001, Paris, France, July
  18-22, 2001, Proceedings}} \emph{(\bibinfo{series}{Lecture Notes in Computer
  Science}, Vol.~\bibinfo{volume}{2102})},
  \bibfield{editor}{\bibinfo{person}{G{\'{e}}rard Berry},
  \bibinfo{person}{Hubert Comon}, {and} \bibinfo{person}{Alain Finkel}} (Eds.).
  \bibinfo{publisher}{Springer}, \bibinfo{pages}{506--518}.
\newblock
\urldef\tempurl%
\url{https://doi.org/10.1007/3-540-44585-4\_48}
\showDOI{\tempurl}


\bibitem[Dang(2003)]%
        {DBLP:journals/tcs/Dang03}
\bibfield{author}{\bibinfo{person}{Zhe Dang}.} \bibinfo{year}{2003}\natexlab{}.
\newblock \showarticletitle{Pushdown timed automata: a binary reachability
  characterization and safety verification}.
\newblock \bibinfo{journal}{\emph{Theor. Comput. Sci.}} \bibinfo{volume}{302},
  \bibinfo{number}{1-3} (\bibinfo{year}{2003}), \bibinfo{pages}{93--121}.
\newblock
\urldef\tempurl%
\url{https://doi.org/10.1016/S0304-3975(02)00743-0}
\showDOI{\tempurl}


\bibitem[Dang and Ibarra(2002)]%
        {DangI02}
\bibfield{author}{\bibinfo{person}{Zhe Dang} {and} \bibinfo{person}{Oscar~H.
  Ibarra}.} \bibinfo{year}{2002}\natexlab{}.
\newblock \showarticletitle{The Existence of $\omega$-Chains for Transitive
  Mixed Linear Relations and Its Applications}.
\newblock \bibinfo{journal}{\emph{Int. J. Found. Comput. Sci.}}
  \bibinfo{volume}{13}, \bibinfo{number}{6} (\bibinfo{year}{2002}),
  \bibinfo{pages}{911--936}.
\newblock
\urldef\tempurl%
\url{https://doi.org/10.1142/S0129054102001539}
\showDOI{\tempurl}


\bibitem[Dang et~al\mbox{.}(2000)]%
        {DIBKS00}
\bibfield{author}{\bibinfo{person}{Zhe Dang}, \bibinfo{person}{Oscar~H.
  Ibarra}, \bibinfo{person}{Tevfik Bultan}, \bibinfo{person}{Richard~A.
  Kemmerer}, {and} \bibinfo{person}{Jianwen Su}.}
  \bibinfo{year}{2000}\natexlab{}.
\newblock \showarticletitle{Binary Reachability Analysis of Discrete Pushdown
  Timed Automata}. In \bibinfo{booktitle}{\emph{Computer Aided Verification,
  12th International Conference, {CAV} 2000, Chicago, IL, USA, July 15-19,
  2000, Proceedings}} \emph{(\bibinfo{series}{Lecture Notes in Computer
  Science}, Vol.~\bibinfo{volume}{1855})},
  \bibfield{editor}{\bibinfo{person}{E.~Allen Emerson} {and}
  \bibinfo{person}{A.~Prasad Sistla}} (Eds.). \bibinfo{publisher}{Springer},
  \bibinfo{pages}{69--84}.
\newblock
\urldef\tempurl%
\url{https://doi.org/10.1007/10722167\_9}
\showDOI{\tempurl}


\bibitem[Dang et~al\mbox{.}(2001)]%
        {DPK01}
\bibfield{author}{\bibinfo{person}{Zhe Dang}, \bibinfo{person}{Pierluigi {San
  Pietro}}, {and} \bibinfo{person}{Richard~A. Kemmerer}.}
  \bibinfo{year}{2001}\natexlab{}.
\newblock \showarticletitle{On Presburger Liveness of Discrete Timed Automata}.
  In \bibinfo{booktitle}{\emph{{STACS} 2001, 18th Annual Symposium on
  Theoretical Aspects of Computer Science, Dresden, Germany, February 15-17,
  2001, Proceedings}} \emph{(\bibinfo{series}{Lecture Notes in Computer
  Science}, Vol.~\bibinfo{volume}{2010})},
  \bibfield{editor}{\bibinfo{person}{Afonso Ferreira} {and}
  \bibinfo{person}{Horst Reichel}} (Eds.). \bibinfo{publisher}{Springer},
  \bibinfo{pages}{132--143}.
\newblock
\urldef\tempurl%
\url{https://doi.org/10.1007/3-540-44693-1\_12}
\showDOI{\tempurl}


\bibitem[de~Moura and Bj{\o}rner(2008)]%
        {MB08}
\bibfield{author}{\bibinfo{person}{Leonardo~Mendon{\c{c}}a de Moura} {and}
  \bibinfo{person}{Nikolaj~S. Bj{\o}rner}.} \bibinfo{year}{2008}\natexlab{}.
\newblock \showarticletitle{{Z3:} An Efficient {SMT} Solver}. In
  \bibinfo{booktitle}{\emph{Tools and Algorithms for the Construction and
  Analysis of Systems, 14th International Conference, {TACAS} 2008, Held as
  Part of the Joint European Conferences on Theory and Practice of Software,
  {ETAPS} 2008, Budapest, Hungary, March 29-April 6, 2008. Proceedings}}
  \emph{(\bibinfo{series}{Lecture Notes in Computer Science},
  Vol.~\bibinfo{volume}{4963})}, \bibfield{editor}{\bibinfo{person}{C.~R.
  Ramakrishnan} {and} \bibinfo{person}{Jakob Rehof}} (Eds.).
  \bibinfo{publisher}{Springer}, \bibinfo{pages}{337--340}.
\newblock
\urldef\tempurl%
\url{https://doi.org/10.1007/978-3-540-78800-3\_24}
\showDOI{\tempurl}


\bibitem[Finkel and Gupta(2019a)]%
        {DBLP:conf/fsttcs/FinkelG19}
\bibfield{author}{\bibinfo{person}{Alain Finkel} {and}
  \bibinfo{person}{Ekanshdeep Gupta}.} \bibinfo{year}{2019}\natexlab{a}.
\newblock \showarticletitle{The Well Structured Problem for {Presburger}
  Counter Machines}. In \bibinfo{booktitle}{\emph{39th {IARCS} Annual
  Conference on Foundations of Software Technology and Theoretical Computer
  Science, {FSTTCS} 2019, December 11-13, 2019, Bombay, India}}
  \emph{(\bibinfo{series}{LIPIcs}, Vol.~\bibinfo{volume}{150})},
  \bibfield{editor}{\bibinfo{person}{Arkadev Chattopadhyay} {and}
  \bibinfo{person}{Paul Gastin}} (Eds.). \bibinfo{publisher}{Schloss Dagstuhl -
  Leibniz-Zentrum f{\"{u}}r Informatik}, \bibinfo{pages}{41:1--41:15}.
\newblock
\urldef\tempurl%
\url{https://doi.org/10.4230/LIPIcs.FSTTCS.2019.41}
\showDOI{\tempurl}


\bibitem[Finkel and Gupta(2019b)]%
        {DBLP:journals/corr/abs-1910-02736}
\bibfield{author}{\bibinfo{person}{Alain Finkel} {and}
  \bibinfo{person}{Ekanshdeep Gupta}.} \bibinfo{year}{2019}\natexlab{b}.
\newblock \showarticletitle{The Well Structured Problem for {Presburger}
  Counter Machines}.
\newblock \bibinfo{journal}{\emph{CoRR}}  \bibinfo{volume}{abs/1910.02736}
  (\bibinfo{year}{2019}).
\newblock
\showeprint[arXiv]{1910.02736}


\bibitem[Finkel and Schnoebelen(2001)]%
        {DBLP:journals/tcs/FinkelS01}
\bibfield{author}{\bibinfo{person}{Alain Finkel} {and}
  \bibinfo{person}{Philippe Schnoebelen}.} \bibinfo{year}{2001}\natexlab{}.
\newblock \showarticletitle{Well-structured transition systems everywhere!}
\newblock \bibinfo{journal}{\emph{Theor. Comput. Sci.}} \bibinfo{volume}{256},
  \bibinfo{number}{1-2} (\bibinfo{year}{2001}), \bibinfo{pages}{63--92}.
\newblock
\urldef\tempurl%
\url{https://doi.org/10.1016/S0304-3975(00)00102-X}
\showDOI{\tempurl}


\bibitem[Fourier(1826)]%
        {fourier1826}
\bibfield{author}{\bibinfo{person}{Jean Baptiste~Joseph Fourier}.}
  \bibinfo{year}{1826}\natexlab{}.
\newblock \showarticletitle{Solution d'une question particuliere du calcul des
  in\'{e}galit\'{e}s}.
\newblock \bibinfo{journal}{\emph{Nouveau Bulletin des Sciences par la
  Soci\'{e}t\'{e} philomatique de Paris}}  \bibinfo{volume}{99}
  (\bibinfo{year}{1826}).
\newblock


\bibitem[Ginsburg and Spanier(1966)]%
        {ginsburg1966bounded}
\bibfield{author}{\bibinfo{person}{Seymour Ginsburg} {and}
  \bibinfo{person}{Edwin~H Spanier}.} \bibinfo{year}{1966}\natexlab{}.
\newblock \showarticletitle{Bounded regular sets}.
\newblock \bibinfo{journal}{\emph{Proc. Amer. Math. Soc.}}
  \bibinfo{volume}{17}, \bibinfo{number}{5} (\bibinfo{year}{1966}),
  \bibinfo{pages}{1043--1049}.
\newblock
\urldef\tempurl%
\url{https://doi.org/10.1090/S0002-9939-1966-0201310-3}
\showDOI{\tempurl}


\bibitem[Grobler et~al\mbox{.}(2023)]%
        {DBLP:journals/corr/abs-2301-08969}
\bibfield{author}{\bibinfo{person}{Mario Grobler}, \bibinfo{person}{Leif
  Sabellek}, {and} \bibinfo{person}{Sebastian Siebertz}.}
  \bibinfo{year}{2023}\natexlab{}.
\newblock \showarticletitle{Parikh Automata on Infinite Words}.
\newblock \bibinfo{journal}{\emph{CoRR}}  \bibinfo{volume}{abs/2301.08969}
  (\bibinfo{year}{2023}).
\newblock
\urldef\tempurl%
\url{https://doi.org/10.48550/arXiv.2301.08969}
\showDOI{\tempurl}
\showeprint[arXiv]{2301.08969}


\bibitem[Grumbach et~al\mbox{.}(2001)]%
        {GRS01}
\bibfield{author}{\bibinfo{person}{St{\'{e}}phane Grumbach},
  \bibinfo{person}{Philippe Rigaux}, {and} \bibinfo{person}{Luc Segoufin}.}
  \bibinfo{year}{2001}\natexlab{}.
\newblock \showarticletitle{Spatio-Temporal Data Handling with Constraints}.
\newblock \bibinfo{journal}{\emph{GeoInformatica}} \bibinfo{volume}{5},
  \bibinfo{number}{1} (\bibinfo{year}{2001}), \bibinfo{pages}{95--115}.
\newblock
\urldef\tempurl%
\url{https://doi.org/10.1023/A:1011464022461}
\showDOI{\tempurl}


\bibitem[Guha et~al\mbox{.}(2022)]%
        {DBLP:conf/fsttcs/GuhaJL022}
\bibfield{author}{\bibinfo{person}{Shibashis Guha},
  \bibinfo{person}{Isma{\"{e}}l Jecker}, \bibinfo{person}{Karoliina Lehtinen},
  {and} \bibinfo{person}{Martin Zimmermann}.} \bibinfo{year}{2022}\natexlab{}.
\newblock \showarticletitle{Parikh Automata over Infinite Words}. In
  \bibinfo{booktitle}{\emph{42nd {IARCS} Annual Conference on Foundations of
  Software Technology and Theoretical Computer Science, {FSTTCS} 2022, December
  18-20, 2022, {IIT} Madras, Chennai, India}} \emph{(\bibinfo{series}{LIPIcs},
  Vol.~\bibinfo{volume}{250})}, \bibfield{editor}{\bibinfo{person}{Anuj Dawar}
  {and} \bibinfo{person}{Venkatesan Guruswami}} (Eds.).
  \bibinfo{publisher}{Schloss Dagstuhl - Leibniz-Zentrum f{\"{u}}r Informatik},
  \bibinfo{pages}{40:1--40:20}.
\newblock
\urldef\tempurl%
\url{https://doi.org/10.4230/LIPIcs.FSTTCS.2022.40}
\showDOI{\tempurl}


\bibitem[Haase et~al\mbox{.}(2009)]%
        {DBLP:conf/concur/HaaseKOW09}
\bibfield{author}{\bibinfo{person}{Christoph Haase}, \bibinfo{person}{Stephan
  Kreutzer}, \bibinfo{person}{Jo{\"{e}}l Ouaknine}, {and}
  \bibinfo{person}{James Worrell}.} \bibinfo{year}{2009}\natexlab{}.
\newblock \showarticletitle{Reachability in Succinct and Parametric One-Counter
  Automata}. In \bibinfo{booktitle}{\emph{{CONCUR} 2009 - Concurrency Theory,
  20th International Conference, {CONCUR} 2009, Bologna, Italy, September 1-4,
  2009. Proceedings}} \emph{(\bibinfo{series}{Lecture Notes in Computer
  Science}, Vol.~\bibinfo{volume}{5710})},
  \bibfield{editor}{\bibinfo{person}{Mario Bravetti} {and}
  \bibinfo{person}{Gianluigi Zavattaro}} (Eds.). \bibinfo{publisher}{Springer},
  \bibinfo{pages}{369--383}.
\newblock
\urldef\tempurl%
\url{https://doi.org/10.1007/978-3-642-04081-8\_25}
\showDOI{\tempurl}


\bibitem[Hague and Lin(2011)]%
        {DBLP:conf/cav/HagueL11}
\bibfield{author}{\bibinfo{person}{Matthew Hague} {and}
  \bibinfo{person}{Anthony~Widjaja Lin}.} \bibinfo{year}{2011}\natexlab{}.
\newblock \showarticletitle{Model Checking Recursive Programs with Numeric Data
  Types}. In \bibinfo{booktitle}{\emph{Computer Aided Verification - 23rd
  International Conference, {CAV} 2011, Snowbird, UT, USA, July 14-20, 2011.
  Proceedings}} \emph{(\bibinfo{series}{Lecture Notes in Computer Science},
  Vol.~\bibinfo{volume}{6806})}, \bibfield{editor}{\bibinfo{person}{Ganesh
  Gopalakrishnan} {and} \bibinfo{person}{Shaz Qadeer}} (Eds.).
  \bibinfo{publisher}{Springer}, \bibinfo{pages}{743--759}.
\newblock
\urldef\tempurl%
\url{https://doi.org/10.1007/978-3-642-22110-1\_60}
\showDOI{\tempurl}


\bibitem[Hague and Lin(2012)]%
        {HL12}
\bibfield{author}{\bibinfo{person}{Matthew Hague} {and}
  \bibinfo{person}{Anthony~Widjaja Lin}.} \bibinfo{year}{2012}\natexlab{}.
\newblock \showarticletitle{Synchronisation- and Reversal-Bounded Analysis of
  Multithreaded Programs with Counters}. In \bibinfo{booktitle}{\emph{Computer
  Aided Verification - 24th International Conference, {CAV} 2012, Berkeley, CA,
  USA, July 7-13, 2012 Proceedings}}. \bibinfo{pages}{260--276}.
\newblock
\urldef\tempurl%
\url{https://doi.org/10.1007/978-3-642-31424-7\_22}
\showDOI{\tempurl}


\bibitem[Hague et~al\mbox{.}(2020)]%
        {DBLP:conf/cade/HagueLRW20}
\bibfield{author}{\bibinfo{person}{Matthew Hague}, \bibinfo{person}{Anthony~W.
  Lin}, \bibinfo{person}{Philipp R{\"{u}}mmer}, {and} \bibinfo{person}{Zhilin
  Wu}.} \bibinfo{year}{2020}\natexlab{}.
\newblock \showarticletitle{Monadic Decomposition in Integer Linear
  Arithmetic}. In \bibinfo{booktitle}{\emph{Automated Reasoning - 10th
  International Joint Conference, {IJCAR} 2020, Paris, France, July 1-4, 2020,
  Proceedings, Part {I}}} \emph{(\bibinfo{series}{Lecture Notes in Computer
  Science}, Vol.~\bibinfo{volume}{12166})},
  \bibfield{editor}{\bibinfo{person}{Nicolas Peltier} {and}
  \bibinfo{person}{Viorica Sofronie{-}Stokkermans}} (Eds.).
  \bibinfo{publisher}{Springer}, \bibinfo{pages}{122--140}.
\newblock
\urldef\tempurl%
\url{https://doi.org/10.1007/978-3-030-51074-9\_8}
\showDOI{\tempurl}


\bibitem[Ibarra(1978)]%
        {DBLP:journals/jacm/Ibarra78}
\bibfield{author}{\bibinfo{person}{Oscar~H. Ibarra}.}
  \bibinfo{year}{1978}\natexlab{}.
\newblock \showarticletitle{Reversal-Bounded Multicounter Machines and Their
  Decision Problems}.
\newblock \bibinfo{journal}{\emph{J. {ACM}}} \bibinfo{volume}{25},
  \bibinfo{number}{1} (\bibinfo{year}{1978}), \bibinfo{pages}{116--133}.
\newblock
\urldef\tempurl%
\url{https://doi.org/10.1145/322047.322058}
\showDOI{\tempurl}


\bibitem[Ibarra et~al\mbox{.}(2000)]%
        {DBLP:conf/mfcs/IbarraSDBK00}
\bibfield{author}{\bibinfo{person}{Oscar~H. Ibarra}, \bibinfo{person}{Jianwen
  Su}, \bibinfo{person}{Zhe Dang}, \bibinfo{person}{Tevfik Bultan}, {and}
  \bibinfo{person}{Richard~A. Kemmerer}.} \bibinfo{year}{2000}\natexlab{}.
\newblock \showarticletitle{Conter Machines: Decidable Properties and
  Applications to Verification Problems}. In
  \bibinfo{booktitle}{\emph{Mathematical Foundations of Computer Science 2000,
  25th International Symposium, {MFCS} 2000, Bratislava, Slovakia, August 28 -
  September 1, 2000, Proceedings}} \emph{(\bibinfo{series}{Lecture Notes in
  Computer Science}, Vol.~\bibinfo{volume}{1893})},
  \bibfield{editor}{\bibinfo{person}{Mogens Nielsen} {and}
  \bibinfo{person}{Branislav Rovan}} (Eds.). \bibinfo{publisher}{Springer},
  \bibinfo{pages}{426--435}.
\newblock
\urldef\tempurl%
\url{https://doi.org/10.1007/3-540-44612-5\_38}
\showDOI{\tempurl}


\bibitem[Jhala and Majumdar(2009)]%
        {SMC-survey}
\bibfield{author}{\bibinfo{person}{Ranjit Jhala} {and} \bibinfo{person}{Rupak
  Majumdar}.} \bibinfo{year}{2009}\natexlab{}.
\newblock \showarticletitle{Software model checking}.
\newblock \bibinfo{journal}{\emph{{ACM} Comput. Surv.}} \bibinfo{volume}{41},
  \bibinfo{number}{4} (\bibinfo{year}{2009}), \bibinfo{pages}{21:1--21:54}.
\newblock
\urldef\tempurl%
\url{https://doi.org/10.1145/1592434.1592438}
\showDOI{\tempurl}


\bibitem[Klaedtke and Rue{\ss}(2003)]%
        {DBLP:conf/icalp/KlaedtkeR03}
\bibfield{author}{\bibinfo{person}{Felix Klaedtke} {and}
  \bibinfo{person}{Harald Rue{\ss}}.} \bibinfo{year}{2003}\natexlab{}.
\newblock \showarticletitle{Monadic Second-Order Logics with Cardinalities}. In
  \bibinfo{booktitle}{\emph{Automata, Languages and Programming, 30th
  International Colloquium, {ICALP} 2003, Eindhoven, The Netherlands, June 30 -
  July 4, 2003. Proceedings}} \emph{(\bibinfo{series}{Lecture Notes in Computer
  Science}, Vol.~\bibinfo{volume}{2719})},
  \bibfield{editor}{\bibinfo{person}{Jos C.~M. Baeten},
  \bibinfo{person}{Jan~Karel Lenstra}, \bibinfo{person}{Joachim Parrow}, {and}
  \bibinfo{person}{Gerhard~J. Woeginger}} (Eds.).
  \bibinfo{publisher}{Springer}, \bibinfo{pages}{681--696}.
\newblock
\urldef\tempurl%
\url{https://doi.org/10.1007/3-540-45061-0\_54}
\showDOI{\tempurl}


\bibitem[Kuper et~al\mbox{.}(2000)]%
        {CDB-book}
\bibfield{author}{\bibinfo{person}{Gabriel Kuper}, \bibinfo{person}{Leonid
  Libkin}, {and} \bibinfo{person}{Jan Paredaens}.}
  \bibinfo{year}{2000}\natexlab{}.
\newblock \bibinfo{booktitle}{\emph{Constraint Databases}}.
\newblock \bibinfo{publisher}{Springer}.
\newblock


\bibitem[Kuske(2010)]%
        {K10}
\bibfield{author}{\bibinfo{person}{Dietrich Kuske}.}
  \bibinfo{year}{2010}\natexlab{}.
\newblock \showarticletitle{Is Ramsey's Theorem omega-automatic?}. In
  \bibinfo{booktitle}{\emph{27th International Symposium on Theoretical Aspects
  of Computer Science, {STACS} 2010, March 4-6, 2010, Nancy, France}}
  \emph{(\bibinfo{series}{LIPIcs}, Vol.~\bibinfo{volume}{5})},
  \bibfield{editor}{\bibinfo{person}{Jean{-}Yves Marion} {and}
  \bibinfo{person}{Thomas Schwentick}} (Eds.). \bibinfo{publisher}{Schloss
  Dagstuhl - Leibniz-Zentrum f{\"{u}}r Informatik}, \bibinfo{pages}{537--548}.
\newblock
\urldef\tempurl%
\url{https://doi.org/10.4230/LIPIcs.STACS.2010.2483}
\showDOI{\tempurl}


\bibitem[Legay(2008)]%
        {TORMC}
\bibfield{author}{\bibinfo{person}{Axel Legay}.}
  \bibinfo{year}{2008}\natexlab{}.
\newblock \showarticletitle{{T(O)RMC:} {A} Tool for (omega)-Regular Model
  Checking}. In \bibinfo{booktitle}{\emph{Computer Aided Verification, 20th
  International Conference, {CAV} 2008, Princeton, NJ, USA, July 7-14, 2008,
  Proceedings}} \emph{(\bibinfo{series}{Lecture Notes in Computer Science},
  Vol.~\bibinfo{volume}{5123})}, \bibfield{editor}{\bibinfo{person}{Aarti
  Gupta} {and} \bibinfo{person}{Sharad Malik}} (Eds.).
  \bibinfo{publisher}{Springer}, \bibinfo{pages}{548--551}.
\newblock
\urldef\tempurl%
\url{https://doi.org/10.1007/978-3-540-70545-1\_52}
\showDOI{\tempurl}


\bibitem[Leino(2023)]%
        {Leino-book}
\bibfield{author}{\bibinfo{person}{K.~Rustan~M. Leino}.}
  \bibinfo{year}{2023}\natexlab{}.
\newblock \showarticletitle{Program Proofs}.
\newblock  (\bibinfo{year}{2023}).
\newblock


\bibitem[Leroux(2021)]%
        {DBLP:conf/focs/Leroux21}
\bibfield{author}{\bibinfo{person}{J{\'{e}}r{\^{o}}me Leroux}.}
  \bibinfo{year}{2021}\natexlab{}.
\newblock \showarticletitle{The Reachability Problem for Petri Nets is Not
  Primitive Recursive}. In \bibinfo{booktitle}{\emph{62nd {IEEE} Annual
  Symposium on Foundations of Computer Science, {FOCS} 2021, Denver, CO, USA,
  February 7-10, 2022}}. \bibinfo{publisher}{{IEEE}},
  \bibinfo{pages}{1241--1252}.
\newblock
\urldef\tempurl%
\url{https://doi.org/10.1109/FOCS52979.2021.00121}
\showDOI{\tempurl}


\bibitem[Leroux and Schmitz(2019)]%
        {DBLP:conf/lics/LerouxS19}
\bibfield{author}{\bibinfo{person}{J{\'{e}}r{\^{o}}me Leroux} {and}
  \bibinfo{person}{Sylvain Schmitz}.} \bibinfo{year}{2019}\natexlab{}.
\newblock \showarticletitle{Reachability in Vector Addition Systems is
  Primitive-Recursive in Fixed Dimension}. In \bibinfo{booktitle}{\emph{34th
  Annual {ACM/IEEE} Symposium on Logic in Computer Science, {LICS} 2019,
  Vancouver, BC, Canada, June 24-27, 2019}}. \bibinfo{publisher}{{IEEE}},
  \bibinfo{pages}{1--13}.
\newblock
\urldef\tempurl%
\url{https://doi.org/10.1109/LICS.2019.8785796}
\showDOI{\tempurl}


\bibitem[Li et~al\mbox{.}(2020)]%
        {LiCWX20}
\bibfield{author}{\bibinfo{person}{Xie Li}, \bibinfo{person}{Taolue Chen},
  \bibinfo{person}{Zhilin Wu}, {and} \bibinfo{person}{Mingji Xia}.}
  \bibinfo{year}{2020}\natexlab{}.
\newblock \showarticletitle{Computing Linear Arithmetic Representation of
  Reachability Relation of One-Counter Automata}. In
  \bibinfo{booktitle}{\emph{Dependable Software Engineering. Theories, Tools,
  and Applications - 6th International Symposium, {SETTA} 2020, Guangzhou,
  China, November 24-27, 2020, Proceedings}} \emph{(\bibinfo{series}{Lecture
  Notes in Computer Science}, Vol.~\bibinfo{volume}{12153})},
  \bibfield{editor}{\bibinfo{person}{Jun Pang} {and} \bibinfo{person}{Lijun
  Zhang}} (Eds.). \bibinfo{publisher}{Springer}, \bibinfo{pages}{89--107}.
\newblock
\urldef\tempurl%
\url{https://doi.org/10.1007/978-3-030-62822-2\_6}
\showDOI{\tempurl}


\bibitem[Libkin(2003)]%
        {DBLP:journals/tocl/Libkin03}
\bibfield{author}{\bibinfo{person}{Leonid Libkin}.}
  \bibinfo{year}{2003}\natexlab{}.
\newblock \showarticletitle{Variable independence for first-order definable
  constraints}.
\newblock \bibinfo{journal}{\emph{{ACM} Trans. Comput. Log.}}
  \bibinfo{volume}{4}, \bibinfo{number}{4} (\bibinfo{year}{2003}),
  \bibinfo{pages}{431--451}.
\newblock
\urldef\tempurl%
\url{https://doi.org/10.1145/937555.937557}
\showDOI{\tempurl}


\bibitem[Lipton(1976)]%
        {lipton1976reachability}
\bibfield{author}{\bibinfo{person}{Richard Lipton}.}
  \bibinfo{year}{1976}\natexlab{}.
\newblock \showarticletitle{The reachability problem is exponential-space
  hard}.
\newblock \bibinfo{journal}{\emph{Yale University, Department of Computer
  Science, Report}}  \bibinfo{volume}{62} (\bibinfo{year}{1976}).
\newblock


\bibitem[Manna and Pnueli(1970)]%
        {MP70}
\bibfield{author}{\bibinfo{person}{Zohar Manna} {and} \bibinfo{person}{Amir
  Pnueli}.} \bibinfo{year}{1970}\natexlab{}.
\newblock \showarticletitle{Formalization of Properties of Functional
  Programs}.
\newblock \bibinfo{journal}{\emph{J. {ACM}}} \bibinfo{volume}{17},
  \bibinfo{number}{3} (\bibinfo{year}{1970}), \bibinfo{pages}{555--569}.
\newblock
\urldef\tempurl%
\url{https://doi.org/10.1145/321592.321606}
\showDOI{\tempurl}


\bibitem[Markgraf et~al\mbox{.}(2021)]%
        {MSL21}
\bibfield{author}{\bibinfo{person}{Oliver Markgraf}, \bibinfo{person}{Daniel
  Stan}, {and} \bibinfo{person}{Anthony~W. Lin}.}
  \bibinfo{year}{2021}\natexlab{}.
\newblock \showarticletitle{Learning Union of Integer Hypercubes with Queries -
  (with Applications to Monadic Decomposition)}. In
  \bibinfo{booktitle}{\emph{Computer Aided Verification - 33rd International
  Conference, {CAV} 2021, Virtual Event, July 20-23, 2021, Proceedings, Part
  {II}}} \emph{(\bibinfo{series}{Lecture Notes in Computer Science},
  Vol.~\bibinfo{volume}{12760})}, \bibfield{editor}{\bibinfo{person}{Alexandra
  Silva} {and} \bibinfo{person}{K.~Rustan~M. Leino}} (Eds.).
  \bibinfo{publisher}{Springer}, \bibinfo{pages}{243--265}.
\newblock
\urldef\tempurl%
\url{https://doi.org/10.1007/978-3-030-81688-9\_12}
\showDOI{\tempurl}


\bibitem[Nelson and Oppen(1980)]%
        {NO80}
\bibfield{author}{\bibinfo{person}{Greg Nelson} {and} \bibinfo{person}{Derek~C.
  Oppen}.} \bibinfo{year}{1980}\natexlab{}.
\newblock \showarticletitle{Fast Decision Procedures Based on Congruence
  Closure}.
\newblock \bibinfo{journal}{\emph{J. {ACM}}} \bibinfo{volume}{27},
  \bibinfo{number}{2} (\bibinfo{year}{1980}), \bibinfo{pages}{356--364}.
\newblock
\urldef\tempurl%
\url{https://doi.org/10.1145/322186.322198}
\showDOI{\tempurl}


\bibitem[Podelski and Rybalchenko(2004)]%
        {PR04}
\bibfield{author}{\bibinfo{person}{Andreas Podelski} {and}
  \bibinfo{person}{Andrey Rybalchenko}.} \bibinfo{year}{2004}\natexlab{}.
\newblock \showarticletitle{Transition Invariants}. In
  \bibinfo{booktitle}{\emph{19th {IEEE} Symposium on Logic in Computer Science
  {(LICS} 2004), 14-17 July 2004, Turku, Finland, Proceedings}}.
  \bibinfo{publisher}{{IEEE} Computer Society}, \bibinfo{pages}{32--41}.
\newblock
\urldef\tempurl%
\url{https://doi.org/10.1109/LICS.2004.1319598}
\showDOI{\tempurl}


\bibitem[Presburger(1929)]%
        {presburger1929}
\bibfield{author}{\bibinfo{person}{Moj\.{z}esz Presburger}.}
  \bibinfo{year}{1929}\natexlab{}.
\newblock \showarticletitle{{Über die Vollständigkeit eines gewissen Systems
  der Arithmetik ganzer Zahlen, in welchem die Addition als einzige Operation
  hervortritt}}.
\newblock \bibinfo{journal}{\emph{Comptes Rendus du I congres de Mathematiciens
  de Pays Slaves}} (\bibinfo{year}{1929}).
\newblock


\bibitem[Qadeer and Rehof(2005)]%
        {QR05}
\bibfield{author}{\bibinfo{person}{Shaz Qadeer} {and} \bibinfo{person}{Jakob
  Rehof}.} \bibinfo{year}{2005}\natexlab{}.
\newblock \showarticletitle{Context-Bounded Model Checking of Concurrent
  Software}. In \bibinfo{booktitle}{\emph{Tools and Algorithms for the
  Construction and Analysis of Systems, 11th International Conference, {TACAS}
  2005, Held as Part of the Joint European Conferences on Theory and Practice
  of Software, {ETAPS} 2005, Edinburgh, UK, April 4-8, 2005, Proceedings}}
  \emph{(\bibinfo{series}{Lecture Notes in Computer Science},
  Vol.~\bibinfo{volume}{3440})}, \bibfield{editor}{\bibinfo{person}{Nicolas
  Halbwachs} {and} \bibinfo{person}{Lenore~D. Zuck}} (Eds.).
  \bibinfo{publisher}{Springer}, \bibinfo{pages}{93--107}.
\newblock
\urldef\tempurl%
\url{https://doi.org/10.1007/978-3-540-31980-1\_7}
\showDOI{\tempurl}


\bibitem[Quaas et~al\mbox{.}(2017)]%
        {QuaasSW17}
\bibfield{author}{\bibinfo{person}{Karin Quaas}, \bibinfo{person}{Mahsa
  Shirmohammadi}, {and} \bibinfo{person}{James Worrell}.}
  \bibinfo{year}{2017}\natexlab{}.
\newblock \showarticletitle{Revisiting reachability in timed automata}. In
  \bibinfo{booktitle}{\emph{32nd Annual {ACM/IEEE} Symposium on Logic in
  Computer Science, {LICS} 2017, Reykjavik, Iceland, June 20-23, 2017}}.
  \bibinfo{publisher}{{IEEE} Computer Society}, \bibinfo{pages}{1--12}.
\newblock
\urldef\tempurl%
\url{https://doi.org/10.1109/LICS.2017.8005098}
\showDOI{\tempurl}


\bibitem[Rackoff(1978)]%
        {Rackoff78}
\bibfield{author}{\bibinfo{person}{Charles Rackoff}.}
  \bibinfo{year}{1978}\natexlab{}.
\newblock \showarticletitle{The covering and boundedness problems for vector
  addition systems}.
\newblock \bibinfo{journal}{\emph{Theoretical Computer Science}}
  \bibinfo{volume}{6}, \bibinfo{number}{2} (\bibinfo{year}{1978}),
  \bibinfo{pages}{223--231}.
\newblock


\bibitem[Ramsey(1930)]%
        {Ramsey30}
\bibfield{author}{\bibinfo{person}{F.~P. Ramsey}.}
  \bibinfo{year}{1930}\natexlab{}.
\newblock \showarticletitle{{On a Problem of Formal Logic}}.
\newblock \bibinfo{journal}{\emph{Proceedings of the London Mathematical
  Society}} \bibinfo{volume}{s2-30}, \bibinfo{number}{1} (\bibinfo{date}{01}
  \bibinfo{year}{1930}), \bibinfo{pages}{264--286}.
\newblock
\showISSN{0024-6115}
\urldef\tempurl%
\url{https://doi.org/10.1112/plms/s2-30.1.264}
\showDOI{\tempurl}


\bibitem[Schmerl and Simpson(1982)]%
        {schmerlsimpson1982}
\bibfield{author}{\bibinfo{person}{James~H Schmerl} {and}
  \bibinfo{person}{Stephen~G Simpson}.} \bibinfo{year}{1982}\natexlab{}.
\newblock \showarticletitle{On the role of Ramsey quantifiers in first order
  arithmetic1}.
\newblock \bibinfo{journal}{\emph{The Journal of Symbolic Logic}}
  \bibinfo{volume}{47}, \bibinfo{number}{2} (\bibinfo{year}{1982}),
  \bibinfo{pages}{423--435}.
\newblock


\bibitem[Shostak(1984)]%
        {Sho84}
\bibfield{author}{\bibinfo{person}{Robert~E. Shostak}.}
  \bibinfo{year}{1984}\natexlab{}.
\newblock \showarticletitle{Deciding Combinations of Theories}.
\newblock \bibinfo{journal}{\emph{J. {ACM}}} \bibinfo{volume}{31},
  \bibinfo{number}{1} (\bibinfo{year}{1984}), \bibinfo{pages}{1--12}.
\newblock
\urldef\tempurl%
\url{https://doi.org/10.1145/2422.322411}
\showDOI{\tempurl}


\bibitem[Sontag(1985)]%
        {sontagRealAdditionPolynomial1985}
\bibfield{author}{\bibinfo{person}{Eduardo~D. Sontag}.}
  \bibinfo{year}{1985}\natexlab{}.
\newblock \showarticletitle{Real Addition and the Polynomial Hierarchy}.
\newblock \bibinfo{journal}{\emph{Inform. Process. Lett.}}
  \bibinfo{volume}{20}, \bibinfo{number}{3} (\bibinfo{date}{April}
  \bibinfo{year}{1985}), \bibinfo{pages}{115--120}.
\newblock
\showISSN{00200190}
\urldef\tempurl%
\url{https://doi.org/10.1016/0020-0190(85)90076-6}
\showDOI{\tempurl}


\bibitem[To(2009)]%
        {To9}
\bibfield{author}{\bibinfo{person}{Anthony~Widjaja To}.}
  \bibinfo{year}{2009}\natexlab{}.
\newblock \showarticletitle{Model Checking {FO(R)} over One-Counter Processes
  and beyond}. In \bibinfo{booktitle}{\emph{Computer Science Logic, 23rd
  international Workshop, {CSL} 2009, 18th Annual Conference of the EACSL,
  Coimbra, Portugal, September 7-11, 2009. Proceedings}}.
  \bibinfo{pages}{485--499}.
\newblock
\urldef\tempurl%
\url{https://doi.org/10.1007/978-3-642-04027-6\_35}
\showDOI{\tempurl}


\bibitem[To and Libkin(2008)]%
        {TL08}
\bibfield{author}{\bibinfo{person}{Anthony~Widjaja To} {and}
  \bibinfo{person}{Leonid Libkin}.} \bibinfo{year}{2008}\natexlab{}.
\newblock \showarticletitle{Recurrent Reachability Analysis in Regular Model
  Checking}. In \bibinfo{booktitle}{\emph{Logic for Programming, Artificial
  Intelligence, and Reasoning, 15th International Conference, {LPAR} 2008,
  Doha, Qatar, November 22-27, 2008. Proceedings}}. \bibinfo{pages}{198--213}.
\newblock
\urldef\tempurl%
\url{https://doi.org/10.1007/978-3-540-89439-1\_15}
\showDOI{\tempurl}


\bibitem[Veanes et~al\mbox{.}(2017)]%
        {DBLP:journals/jacm/VeanesBNB17}
\bibfield{author}{\bibinfo{person}{Margus Veanes}, \bibinfo{person}{Nikolaj~S.
  Bj{\o}rner}, \bibinfo{person}{Lev Nachmanson}, {and} \bibinfo{person}{Sergey
  Bereg}.} \bibinfo{year}{2017}\natexlab{}.
\newblock \showarticletitle{Monadic Decomposition}.
\newblock \bibinfo{journal}{\emph{J. {ACM}}} \bibinfo{volume}{64},
  \bibinfo{number}{2} (\bibinfo{year}{2017}), \bibinfo{pages}{14:1--14:28}.
\newblock
\urldef\tempurl%
\url{https://doi.org/10.1145/3040488}
\showDOI{\tempurl}


\bibitem[Weispfenning(1997)]%
        {weispfenning1997}
\bibfield{author}{\bibinfo{person}{Volker Weispfenning}.}
  \bibinfo{year}{1997}\natexlab{}.
\newblock \showarticletitle{Complexity and Uniformity of Elimination in
  {Presburger} Arithmetic}. In \bibinfo{booktitle}{\emph{Proceedings of the
  1997 International Symposium on Symbolic and Algebraic Computation, {ISSAC}
  1997, Maui, Hawaii, USA, July 21-23, 1997}},
  \bibfield{editor}{\bibinfo{person}{Bruce~W. Char}, \bibinfo{person}{Paul~S.
  Wang}, {and} \bibinfo{person}{Wolfgang K{\"{u}}chlin}} (Eds.).
  \bibinfo{publisher}{{ACM}}, \bibinfo{pages}{48--53}.
\newblock
\urldef\tempurl%
\url{https://doi.org/10.1145/258726.258746}
\showDOI{\tempurl}


\bibitem[Weispfenning(1999)]%
        {Weispfenning99}
\bibfield{author}{\bibinfo{person}{Volker Weispfenning}.}
  \bibinfo{year}{1999}\natexlab{}.
\newblock \showarticletitle{Mixed Real-Integer Linear Quantifier Elimination}.
  In \bibinfo{booktitle}{\emph{Proceedings of the 1999 International Symposium
  on Symbolic and Algebraic Computation, {ISSAC} '99, Vancouver, B.C., Canada,
  July 29-31, 1999}}, \bibfield{editor}{\bibinfo{person}{Keith~O. Geddes},
  \bibinfo{person}{Bruno Salvy}, {and} \bibinfo{person}{Samuel~S. Dooley}}
  (Eds.). \bibinfo{publisher}{{ACM}}, \bibinfo{pages}{129--136}.
\newblock
\urldef\tempurl%
\url{https://doi.org/10.1145/309831.309888}
\showDOI{\tempurl}


\bibitem[Williams(1986)]%
        {williams1986fourier}
\bibfield{author}{\bibinfo{person}{H~Paul Williams}.}
  \bibinfo{year}{1986}\natexlab{}.
\newblock \showarticletitle{Fourier's method of linear programming and its
  dual}.
\newblock \bibinfo{journal}{\emph{The American mathematical monthly}}
  \bibinfo{volume}{93}, \bibinfo{number}{9} (\bibinfo{year}{1986}),
  \bibinfo{pages}{681--695}.
\newblock
\urldef\tempurl%
\url{https://doi.org/10.2307/2322281}
\showDOI{\tempurl}


\end{thebibliography}

\end{document}